\documentclass[nonblindrev]{informs3} 

\DoubleSpacedXI 

\usepackage{endnotes}
\let\footnote=\endnote

\usepackage{natbib}
\bibpunct[, ]{(}{)}{,}{a}{}{,}%
\def\bibfont{\small}%
\def\BIBand{and}%

\usepackage{url}
\usepackage{nicefrac}
\usepackage{amssymb}
\usepackage{amsmath}
\usepackage{hyperref}
\hypersetup{
    colorlinks,
    linkcolor={red!50!black},
    citecolor={blue!50!black},
    urlcolor={blue!80!black}
}
\usepackage{breakurl}
\newcommand{\LeftEqNo}{\let\veqno\@@leqno}
\usepackage{environ}
\usepackage{xcolor}
\usepackage{float}
\usepackage{graphicx}
\usepackage{amsmath}
\usepackage{amssymb}
\usepackage[noend]{algpseudocode}
\usepackage{algorithm}
\usepackage{tikz}
\usepackage{float}
\usepackage{url}
\usepackage{graphicx}
\usepackage{array}
\usepackage{enumerate}
\usepackage{caption}
\usepackage{subcaption}
\usepackage{booktabs}

\newcommand{\setcycles}{\mathcal{C}}
\newcommand{\filledges}{E^c}
\newcommand{\exterior}{\textnormal{ext}}
\newcommand{\interior}{\textnormal{int}}

\newcommand{\edge}[2] { \left\{ #1,#2 \right\} }
\newcommand{\coef}[2] { \mu_{\left\{#1,#2\right\}} }

\newcolumntype{C}[1]{>{\centering\let\newline\\\arraybackslash\hspace{0pt}}m{#1}}

\TheoremsNumberedThrough     
\ECRepeatTheorems

\EquationsNumberedThrough    

\begin{document}


\RUNAUTHOR{Bergman et al.}

\RUNTITLE{On the Minimum Chordal Completion Polytope}

\TITLE{On the Minimum Chordal Completion Polytope}

\ARTICLEAUTHORS{%
\AUTHOR{David Bergman}
\AFF{Operations and Information Management, University of Connecticut,  Storrs, Connecticut 06260\\ 
\EMAIL{david.bergman@business.uconn.edu} \URL{}}

\AUTHOR{Carlos H. Cardonha}
\AFF{IBM Research, Brazil, S\~ao Paulo 04007-900\\ \EMAIL{carloscardonha@br.ibm.com } \URL{}}

\AUTHOR{Andre A. Cire}
\AFF{Department of Management, University of Toronto Scarborough, Toronto, Ontario M1C-1A4, Canada,
\\ \EMAIL{acire@utsc.utoronto.ca} \URL{}}

\AUTHOR{Arvind U. Raghunathan}
\AFF{Mitsubishi Electric Research Labs, 201 Broadway, Cambridge,
\\ \EMAIL{raghunathan@merl.com} \URL{}}

} 

\ABSTRACT{%
A graph is chordal if every cycle of length at least four contains a \textit{chord}, that is,
an edge connecting two nonconsecutive vertices of the cycle. Several classical applications in sparse linear systems, database management, computer vision, and semidefinite programming can be reduced to finding the minimum number of edges to add to a graph so that it becomes chordal, known as the \textit{minimum chordal completion problem} (MCCP). In this article we propose a new formulation for the MCCP which does not rely on finding perfect elimination orderings of the graph, as has been considered in previous work. We introduce several families of facet-defining inequalities for cycle subgraphs and investigate the underlying separation problems, showing that some key inequalities are NP-Hard to separate. We also show general properties of the proposed polyhedra, indicating certain conditions and methods through which facets and inequalities associated with the polytope of a certain graph can be adapted in order to become valid and eventually facet-defining for some of its subgraphs or supergraphs.
Numerical studies combining heuristic separation methods based on a threshold rounding and lazy-constraint generation indicate that our approach substantially outperforms existing methods for the MCCP, solving many benchmark graphs to optimality for the first time.
}%


\KEYWORDS{Networks/graphs; Applications: Networks/graphs ;  Programming: Integer: Algorithms: Cutting plane/facet}
\HISTORY{}

\maketitle

%


\section{Introduction}

%

Given a simple undirected graph $G=(V,E)$, the \textit{minimum chordal completion problem} (MCCP) asks for the minimum number of edges to add to $E$ so that the graph becomes \textit{chordal}; that is, every cycle of length at least four in $G$ has an edge connecting two non-consecutive vertices (i.e., a \textit{chord}). Figure \ref{fig:SimpleGraph}(b) depicts an example of a minimum chordal completion of the graph in Figure \ref{fig:SimpleGraph}(a), 
where a chord $\{v_1,v_3\}$ is added because of the chordless cycle $(v_0,v_1,v_4,v_3)$. The problem is also referred to as the \textit{minimum triangulation problem} or the \textit{minimum fill-in problem}.

The MCCP is a classical combinatorial optimization problem with a variety of applications spanning both the computer science and the operations research literature. Initial algorithmic aspects for directed graphs were 
investigated by \cite{Rose1976} and \cite{Rose1978}, motivated by problems arising in Gaussian elimination of sparse linear equality systems. The most general version of the MCCP was only proven NP-Hard later by \cite{Yan1981}, and since then minimum graph chordalization methods have been applied in database management \citep{BeeFagMaiYan83,TarYan84}, sparse matrix computation \citep{GroJohSaWol1984,Fomin2013}, artificial
intelligence \citep{Lauritzen1990}, computer vision \citep{ChuMum1994}, 
and in several other contexts; see, e.g., the survey by \cite{Heggernes2006}. Most recently, solution methods for the MCCP have gained a central role in semidefinite and nonlinear optimization, in particular for exploiting sparsity of linear and nonlinear constraint matrices \citep{NakFujFukKojMur2003, KimKojMevYam2011,Vandenberghe2015}.

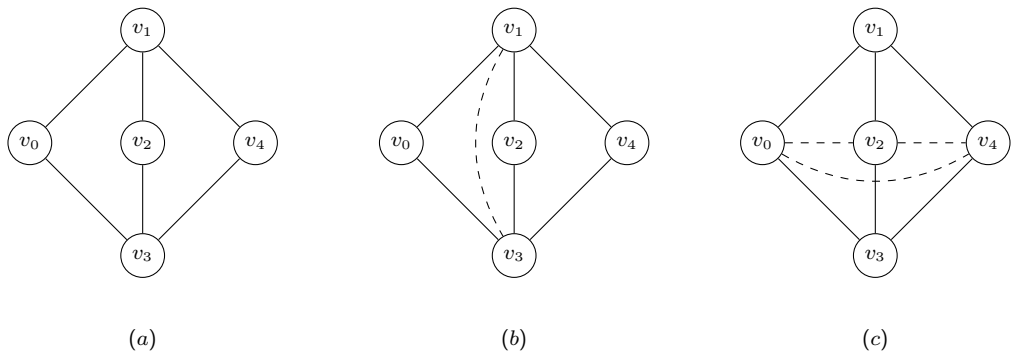
\begin{figure}[t]
\begin{center}
\begin{minipage}{.25\textwidth}
\begin{tikzpicture}[scale=0.75][font=\sffamily,\scriptsize]

\node (1) [draw,circle] [draw] at (0,0) {$v_0$};
\node (2) [draw,circle] [draw] at (2,2) {$v_1$};
\node (3) [draw,circle] [draw] at (2,0) {$v_2$};
\node (4) [draw,circle] [draw] at (2,-2) {$v_3$};
\node (5) [draw,circle] [draw] at (4,0) {$v_4$};

\node (6) [draw=none] at (2,-3.5) {($a$)};

\path[-](1) edge (2);
\path[-](1) edge (4);
\path[-](2) edge (3);
\path[-](3) edge (4);
\path[-](2) edge (5);
\path[-](4) edge (5);
\end{tikzpicture}
\end{minipage}%
\hspace{2.5ex}
\begin{minipage}{.25\textwidth}
\centering
\begin{tikzpicture}[scale=0.75][font=\sffamily,\scriptsize]

\node (1) [draw,circle] [draw] at (0,0) {$v_0$};
\node (2) [draw,circle] [draw] at (2,2) {$v_1$};
\node (3) [draw,circle] [draw] at (2,0) {$v_2$};
\node (4) [draw,circle] [draw] at (2,-2) {$v_3$};
\node (5) [draw,circle] [draw] at (4,0) {$v_4$};

\node (6) [draw=none] at (2,-3.5) {($b$)};

\path[-](1) edge (2);
\path[-](1) edge (4);
\path[-](2) edge (3);
\path[-](3) edge (4);
\path[-](2) edge (5);
\path[-](4) edge (5);
\path[dashed,bend right = 30](2) edge (4);

\end{tikzpicture}
\end{minipage}
\hspace{2.5ex}
\begin{minipage}{.25\textwidth}
\centering
\begin{tikzpicture}[scale=0.75][font=\sffamily,\scriptsize]

\node (1) [draw,circle] [draw] at (0,0) {$v_0$};
\node (2) [draw,circle] [draw] at (2,2) {$v_1$};
\node (3) [draw,circle] [draw] at (2,0) {$v_2$};
\node (4) [draw,circle] [draw] at (2,-2) {$v_3$};
\node (5) [draw,circle] [draw] at (4,0) {$v_4$};
\path[dashed,bend right = 30](1) edge (5);

\node (6) [draw=none] at (2,-3.5) {($c$)};

\path[-](1) edge (2);
\path[-](1) edge (4);
\path[-](2) edge (3);
\path[-](3) edge (4);
\path[-](2) edge (5);
\path[-](4) edge (5);
\path[dashed](1) edge (3);
\path[dashed](3) edge (5);
\end{tikzpicture}
\end{minipage}
\caption{(a)~An example graph by \cite{Heggernes2006}. (b)~A minimum chordal completion of the graph.  (c)~A minim\underline{al}, non-optimal chordal completion of the graph.}
\label{fig:SimpleGraph}
\end{center}
\end{figure}

The literature on exact computational approaches for the MCCP is, however,  surprisingly scarce. A large stream of research in the computer science community has focused on the identification of polynomial-time algorithms for structured classes of graphs, such as the works by \cite{Cha1996,Moh1997,KloKraWon1998}, and \cite{BodKloKraHue1998}. For general graphs, 
\cite{Kaplan1999} proposed the first exact fixed-parameter tractable algorithm for the MCCP, and recently \cite{FomVil2012} presented a substantially faster procedure with sub-exponential parameterized time complexity. Such algorithms are typically extremely challenging to implement and have not been used in computational experiments involving existing datasets.

In regards to practical optimization approaches, the primary focus has been on heuristic methodologies with no optimality guarantees, including techniques by  \cite{MezMos2010,Rollon2011}, and \cite{BerBorHegSimVil2006,BerHegSim2003}. The state-of-the-art heuristic, proposed by \cite{GeoLiu1989},
is a simple and efficient algorithm based on ordering the vertices by their degree. Previous articles  have also developed methodologies for finding a \textit{minimal} chordal completion. 
Note that a minimal chordal completion is not necessarily \textit{minimum}, as depicted in Figure \ref{fig:SimpleGraph}(c), but does provide a heuristic solution to the MCCP.

To the best of our knowledge, the first mathematical programming model for the MCCP is derived from a simple modification of the formulation by \cite{Feremans02} for determining the \textit{tree-widths} of graphs. The model
is based on a result by \cite{Fulkerson1965}, which states that a graph
is chordal if and only it has a \textit{perfect elimination ordering}, to
be detailed further in this paper. \cite{Yuceoglu2015}
has recently provided a first polyhedral analysis and computational testing of this formulation, deriving
facets and other valid inequalities for specific classes of graphs. Alternatively, \cite{BerRaghu15} introduced a Benders approach to the MCCP that relies on a simple class of valid inequalities, outperforming a simple constraint programming backtracking search algorithm.  


\paragraph{Our Contributions}

In this paper we investigate a novel computational approach to the MCCP that extends the preliminary work of \cite{BerRaghu15}. Our technique is based on a mathematical programming model composed of exponentially many constraints -- the \textit{chordal inequalities} -- that are defined directly on the edge space of the graph and do not depend on the perfect elimination ordering property, as opposed to earlier formulations. We investigate the polyhedral structure of our model, which reveals that the proposed inequalities are part of a special class of constraints that induce exponentially many facets for cycle subgraphs. This technique can be generalized to lift other inequalities and greatly strengthen the corresponding linear programming relaxation.

Building on these theoretical results, we propose a hybrid solution method that alternates a lazy-constraint generation with a heuristic separation procedure. The resulting approach is compared to the current state-of-the-art models in the literature, and it is empirically shown to improve solution times and known optimality gaps for standard graph benchmarks, often by orders of magnitude.

\smallskip

\paragraph{Organization of the Paper.}

After introducing our notation in \S\ref{sec:notation}, in \S\ref{sec:IPFormulation} we describe a new integer programming (IP) formulation for the MCCP
and characterize its polytope in \S\ref{sec:polytopeMCCP}, also proving dimensionality results and simple upper bound facets. \S\ref{sec:cycleGraphFacets} provides an in-depth analysis of the polyhedral structure of cycle graphs, introducing four classes of facet-defining inequalities. In \S\ref{sec:GeneralPolyhedralProperties} we prove general properties of the polytope, including some results relating to the lifting of facets. Finally, in \S \ref{sec:solutionMethod} we propose a hybrid solution technique that considers both a lazy-constraint generation and a heuristic separation method based on a threshold rounding procedure, and also present a simple primal heuristic for the problem. We provide a numerical study in \S \ref{sec:numericalExperiments}, indicating that our approach substantially outperforms existing methods, in particular solving many benchmark graphs to optimality for the first time.

%
%

\section{Notation and Terminology}
\label{sec:notation}

For the remainder of the paper, we assume that each graph $G = (V,E)$ is connected, undirected, and does not contain self-loops or multi-edges. For any set $S$, ${{S}\choose{2}}$ denotes the family of two-element subsets of $S$. Each edge $e \in E \subseteq {V \choose 2}$ is a two-element subset of vertices in $V$.  The \emph{complement edge set} (or \emph{fill edges}) $\filledges$ of $G$ is the set of edges missing from $G$, that is, $\filledges = {V \choose 2} \backslash E$. We denote by $m$ and $m^c$ the cardinality of the edge set and of the complement edge set of~$G$, respectively (i.e., $m=|E|$ and $m^c = |\filledges|$). 
The graph \emph{induced} by a set $V' \subseteq V$ is the graph $G[V'] = \left(V',E' \right)$ whose edge set is such that $E' = E \cap {V' \choose 2}$.  
If multiple graphs are being considered in a context, we include ``$(G)$'' in the notation to avoid ambiguity; i.e., $V(G')$ and $E(G')$ represent the vertex and edge set of a graph $G'$, respectively. 
Moreover, for every integer $k \geq 0$, we let $[k] := \{1, 2, \ldots, k\}$.


For any ordered list $C = (v_0,v_1,\ldots,v_{k-1})$ of $k$ distinct vertices of $V$, let
$V(C) = \{v_0,v_1,\ldots,v_{k-1}\}$ be the set of vertices composing $C$
and let~$|V(C)|$ be the \textit{size} $|C|$ of $C$. 
The \textit{exterior} of $C$ is the family  
\[
\exterior(C) =  \{\{v_{k-1}, v_0\}\} \cup  \bigcup_{i \in [k-1]} \{  \{v_{i-1}, v_{i}\} \},
\] and the \emph{interior} of $C$ 
is the family of two-element subsets of $V(C)$ that do not belong to $\exterior(C)$,
that is, 
\[
\interior(C) = {V(C) \choose 2} \setminus \exterior(C).
\] 
We call $C$ a \emph{cycle} if $\exterior(C)\subseteq E$. If $C$ is a cycle, any element of $\interior(C)$ is referred to as a \emph{chord}.  
A cycle $C$ for which the induced graph $G[V(C)]$ contains no chords is a \emph{chordless} cycle.  $G$ is said to be \emph{chordal} if the maximum size of any chordless cycle is three.  
A chordless cycle with $k$ vertices is called a $k$\emph{-chordless cycle}. 


Let $G = (V,E)$.  
Every subset of fill edges $F \subseteq \filledges$ is a \emph{completion} of~$G$, and $G+F$ represents  the graph that results from the addition of edges in $F$ to $E$; that is, $G+F := (V,E \cup F)$. A \emph{chordal completion} of $G$ is any subset of fill edges $F \subseteq E^c$  for which the completion $G+F$ is chordal.  A \emph{minimal chordal completion} $F$ is a chordal completion such that, for any proper subset $F' \subset F$, $F'$ is not a chordal completion for $G$.  A \emph{minimum chordal completion} is a minimal chordal completion of minimum cardinality, and  the minimum chordal completion problem (MCCP) is about the identification of such a subset of~$\filledges$. 

\begin{example} Consider the graph in Figure~\ref{fig:SimpleGraph}(a).  This graph has three chordless cycles, $C_1 = (v_0,v_1,v_2,v_3)$, $C_2 = (v_1,v_2,v_3,v_4)$ and $C_3 = (v_0,v_1,v_4,v_3)$.  Figure~\ref{fig:SimpleGraph}(c) shows a chordal completion consisting of edges $\{ v_0,v_2 \},\{ v_2,v_4 \}$, and $\{ v_0,v_4 \}$.  Removing any of these edges will result in a graph that is not chordal, so this chordal completion is minimal.  Figure~\ref{fig:SimpleGraph}(b) shows a smaller minimal chordal completion, a minimum chordal completion consisting only of edge $\{v_1,v_3\}$. 
\end{example}

\section{IP Formulation for the MCCP}
\label{sec:IPFormulation}

We now describe the basic IP formulation investigated  in this paper. Given a graph $G=(V,E)$, each binary variable $x_f$ in our model indicates whether the fill 
edge $f \in \filledges$ is part of a chordal completion for $G$. That is, if we define the set $E(x) := \left\{ f \in \filledges: x_f = 1 \right\}$ for each vector $x \in [0,1]	^{m^c}$, the set of feasible solutions to our model is given by
\[
X(G) := \left\{ x \in \{ 0,1 \}^{m^c} : G + E(x) \;\; \textnormal{is chordal} \right\}.
\]
Thus, each $x \in X(G)$ equivalently represents the characteristic vector of a chordal completion of~$G$. We use $G(x)$ to denote $G + E(x)$, i.e., $G(X) = (V,E\cup E(x))$, and~$x(F)$ to represent the characteristic vector of completion~$F \subseteq \filledges$, i.e., $x(F)_f = 1$ if $f \in F$ and $x(F)_f = 0$ otherwise.

Let $\setcycles$ be the family of all possible ordered lists composed of distinct vertices of $V$, i.e., every  $C \in \setcycles$ can be written as a sequence $C = (v_0, v_1, \dots, v_{k-1})$ for some $k \leq |V|$. Also, let $F_G(C) = F(C) \subseteq \filledges$ be the set of fill edges that are missing in $\exterior(C)$ for $C$ to induce a cycle in~$G$,  
that is,
\[
F(C) := \exterior(C) \setminus E.
\]

We propose the following model for the MCCP:

\begin{minipage}{\textwidth}
\begin{align*}
\textnormal{minimize}\quad & \sum_{f \in \filledges} x_f \tag{IPC} \label{ipMCCP} \\
\textnormal{s.t.}\quad & 
\sum_{f \in \interior(C)} x_f 
- (|C| - 3) \left( \sum_{f \in F(C)} x_f - |F(C)| + 1 \right)
\ge 0, \;\; \textnormal{for all $C \in \setcycles$},  \tag{I1} \label{ineq:ci} \\
& \hspace{45ex} \textnormal{ with $\interior(C) \cap E = \emptyset$} \\
& x \in \{0,1\}^{m^c}.
\end{align*}
\end{minipage}

\medskip

The set of inequalities (\ref{ineq:ci}) will be denoted by \textit{chordal inequalities} henceforth. Note that every sequence~$C$ such that $F(C) = \emptyset$ and $\interior(C) \cap E = \emptyset$  describes a cycle in~$G$, and its associated inequality (\ref{ineq:ci}) simplifies to
\[
\sum_{f \in \interior(C)} x_f \ge |C| - 3.
\]
Lemma \ref{lem:ThreeInteriorEdges} shows that inequalities (\ref{ineq:ci}) are valid in this special case.


\begin{lemma} \label{lem:ThreeInteriorEdges}
If $C$ is a chordless cycle of~$G$ such that $|C| \geq 4$, then any chordal completion of~$G$ contains at least $|C|-3$ edges that belong to~$\interior(C)$.
\end{lemma}
\proof{Proof.}
We proceed by induction on $k$, the cardinality of sequence $C = (v_0,v_1,\ldots, v_{k-1})$.  For $k=4$, if fewer than $4-3 = 1$ edges (i.e., 0 edges) belonging to~$\interior(C)$ are added to $G$, the graph trivially remains not chordal. Assume now $k \geq 5$. By definition, any chordless completion of~$G$ must contain at least one edge in~$\interior(C)$, so let us suppose without loss of generality that a chord $\edge{v_0}{v_p}$ is added to $G$ for some $1 < p < k-1$.


If $p=2$ (or $p = k-2$), then $C' = (v_0,v_2,\ldots, v_{k-1})$ is a chordless cycle of length $k-1$. By the inductive hypothesis, any chordal completion requires at least $(k-1)-3 = k-4$ edges in~$\interior(C')$ (which are also in~$\interior(C)$).  By symmetry, the same analysis and result hold if~$p = k-2$. Therefore any chordal completion of~$G$ employing chord~$\edge{v_0}{v_p}$ for $p = 2$ (or $p = k-2$) requires at least $1 + k-4 = k-3$ edges in~$\interior(C)$.

Finally, let $3 \leq p \leq k-3$. In this case, chord~$\edge{v_0}{v_p}$ cuts $C$ and creates two chordless cycles: $C^1 = (v_0, v_1, \ldots, v_p)$ and $C^2 = (v_0, v_p, v_{p+1}, \ldots, v_{k-1})$.  By the inductive hypothesis, any chordless completion of~$G$ requires at least $p-3$ edges in~$\interior(C^1)$ and at least $k-p+2-3 = k-p-1$ edges in~$\interior(C^2)$; as $\interior(C^1) \cap \interior(C^2) = \emptyset$ and $\interior(C^1) \cup \interior(C^2) \subset \interior(C)$, any chordal completion of~$G$ employing chord~$\edge{1}{v_p}$ for  $3 \leq p \leq k-3$ also requires at least $(1) + (p-3) + (k-p-1) = k - 3$ edges in~$\interior(C)$, as desired. \Halmos
\endproof

Based on the previous lemma, we show below that the set of inequalities (\ref{ineq:ci}) is valid for all sequences in~$\setcycles$ and, as a consequence, that (\ref{ipMCCP}) is a valid formulation for the MCCP.
\begin{proposition}
\label{prop:correctnessIPMCCP}
The model (\ref{ipMCCP}) is a valid formulation for the MCCP.
\end{proposition}
\proof{Proof.}
We first show that there is an one-to-one correspondence between $X(G)$ and solutions to (\ref{ipMCCP}).
Let $x^*$ be a feasible solution to (\ref{ipMCCP}) and suppose that the graph $G+E(x^*)$ contains a chordless cycle $C$ with more than 3 vertices. 
By definition, 
$\sum_{f \in \interior(C)} x^*_f = 0 < |C|-3$, thus contradicting the feasibility of~$x^*$. Conversely, let $E^* \subseteq \filledges$ be such that $G+E^*$ is chordal, and suppose
$x^* = \{x \in \{0,1\}^{m^c} : x_f = 1 \Leftrightarrow f \in E^*\}$ is infeasible
to (\ref{ipMCCP}). Assume that one of the violated inequalities is associated with sequence~$C^*$; note that 
$\sum_{f \in F(C^*)} x_f = |F(C^*)|$, as the associated inequality would be trivially satisfied otherwise. Therefore, $C^*$ is a cycle in $G + E^*$ and we must have $\sum_{f \in \interior(C^*)} x_f < |C|-3$, which, by Lemma~\ref{lem:ThreeInteriorEdges}, contradicts the fact that $G + E^*$ is chordal.
Finally, since the one-to-one correspondence holds and the objective function of (\ref{ipMCCP}) minimizes the number of added edges, the result follows. \Halmos
\endproof

\section{MCCP Polytope Dimension and Simple Upper Bound Facets}
\label{sec:polytopeMCCP}

This section begins our investigation of the convex hull of the feasible set of chordal completions $X(G)$, which will lead to special properties that can be exploited by computational methods for the MCCP. We identify the dimension of the polytope and provide a proof that the simple upper bound inequalities $x_f \le 1$ are facet-defining.

\begin{theorem}
If $G=(V,E)$ is not a complete graph (and hence not trivially chordal),

$\hspace{2ex}$ a. $\textnormal{conv}(X(G))$ is full-dimensional;

$\hspace{2ex}$ b. $x_f \le 1$ is facet-defining for all $f \in \filledges$. 
\end{theorem}

\proof{Proof.}
We first show (a). Let $e \in \{0,1\}^{m^c}$ be the vector consisting only of ones and $e^j \in \{0,1\}^{m^c}$ be the unit vector for coordinate~$j$.  
By definition, $G + E(e)$ is the complete graph (which is trivially chordal), whereas $G + E(e-e^j)$ is the complete graph with only the edge associated with coordinate~$j$ missing. Graphs 
$G + E(e-e^j)$ are chordal for every~$e^j$;   this follows because the graph induced by the set of vertices in any cycle $C$ of cardinality $k \geq 4$ contains at least ${k \choose 2} - 1$ edges and, consequently, $|\interior(C)| \geq 
{k \choose 2} - 1 - k = \frac{k^2-3k-2}{2}$, which 
is greater than or equal to 1 for $k \geq 4$.
The set of $m^c+1$ vectors $\{e\} \cup \{e-e^1,e-e^2,\ldots,e-e^{m^c}\}$  is affinely independent and contained in the set~$X(G)$, and so it follows that $\textnormal{conv}(X(G))$ is a full-dimensional polytope.

For $(b)$, let $f \in \filledges$ and notice that the set of $m^c$ vectors $\{e\} \cup \{e-e^{f'}: f' \in \filledges \setminus \{f\} \}$
is affinely independent and satisfy  $x_f = 1$. \Halmos
\endproof
%

%

\section{Cycle Graph Facets}
\label{sec:cycleGraphFacets}

We restrict our attention now to \textit{cycle graphs}, i.e., graphs consisting of a single cycle. Cycle graphs are the building blocks of computational methodologies for the MCCP, since finding a chordal completion of a graph naturally concerns identifying chordless cycles and eliminating them  by adding chords.  We present four classes of facet-defining inequalities for cycle graphs in this section. 

Let $G=(V,E)$ be a cycle graph associated with a $k$-vertex chordless cycle $C=(v_0, \dots, v_{k-1})$, i.e., $V = V(C)$ and $E = \exterior(C)$. Assume all additions and subtractions involving indices of vertices are modulo-$k$. The proofs presented in this section show only the validity of the inequalities; arguments proving that they are 
facet-defining for cycle graphs are presented in Section~\ref{sec:facetProofsAdditionalInequalities}.

\begin{proposition}
\label{prop:inq1}
Let $G=(V,E)$ be a cycle graph associated with cycle $C=(v_0,\dots,v_{k-1})$, $k \ge 4$. 
The chordal inequality (\ref{ineq:ci}) associated with~$C$, which in this case simplifies to
\[
\sum_{f \in \interior(C)} x_f \ge |C| - 3,
\]
 is facet-defining for $\textnormal{conv}(X(G))$. \Halmos
\end{proposition}
The proof of the validity of the inequality in Proposition~\ref{prop:inq1} follows directly from Lemma~\ref{lem:ThreeInteriorEdges}.

\begin{proposition}
\label{prop:inq2}
If $k \ge 4$, the inequality 
\begin{align*}
x_{ \{ v_{i-1},v_{i+1} \} } +  \sum_{f : v_i \in f, \{v_{i-1},v_{i+1}\} \cap f = \emptyset } x_{f} \geq \: 1, \quad \; \textnormal{for all $i \in \{1,\dots,k\}$} \tag{I2} \label{inq:2}
\end{align*}
is valid and facet-defining for $\textnormal{conv}(X(G))$.
\end{proposition}
\proof{Proof.}
Suppose that (\ref{inq:2}) is violated by some $x \in \textnormal{conv}(X(G))$, i.e., that for some $\{v_{i-1},v_{i},v_{i+1}\} \subset V$, 
$x_{ \{ v_{i-1},v_{i+1} \} } +  \sum_{f : v_i \in f, \{v_{i-1},v_{i+1}\} \cap f = \emptyset } x_{f} = 0$. 
As $k \geq 4$, 
a shortest path~$P$ from $v_{i-1}$ to $v_{i+1}$ in $G(x)$ that does not include $v_i$ traverses at least two edges. The sequence defined by the concatenation of~$P$ with $(v_{i-1}, v_i, v_{i+1})$ 
defines thus a $k'$-chordless cycle of $G(x)$ for $k' \geq 4$, a contradiction. \Halmos
\endproof

For the next proposition, some additional notation is in order. For any two vertices $v_i$ and $v_j$ with $i < j$, let $d_C(v_i,v_j) := \min \{j-i, k-j+i\}$ be the ``distance'' between $v_i$ and $v_j$ in $C$, and assume $d_C(v_i,v_j) := d_C(v_j,v_i)$ if $i > j$. 

\begin{proposition}
\label{prop:inq3}
If $k \ge 5$, the inequality 
\begin{align*}
&\sum_{f \in \left\{ \{v_i,v_j\} \in \filledges \,:\, d_C(v_i,v_j) = 2 \right\} } x_{f} \geq \: 2 \tag{I3} \label{inq:3}
\end{align*}
is valid and facet-defining for $\textnormal{conv}(X(G))$.
\end{proposition}
\proof{Proof.}

Inequality~(\ref{inq:3}) states that at least two out of the $|C|$ pairs of vertices
of distance 2 must appear in any chordal completion of~$G$. 
Without loss of generality, let $f' = \{v_0,v_{j_1}\}$ be the edge of $\interior(C)$ composing some completion~$F$ of~$G$ that connects the ``closest'' vertices with respect to $d_C$. 
If $j_1 \geq 3$, then $C' = (v_0,v_1,\ldots,v_{j_1})$ is a chordless cycle in $G + F$, a contradiction; therefore, $j_1 = 2$. 

As $k \geq 5$, $C' = (v_0,v_2,v_3,\ldots,v_{k-1})$ is a chordless  cycle in $G + F$ with at least 4 vertices, so at least one edge of $\mathrm{int}(C')$ must be present in $G + F$. Let~$f'' =  \{v'_i,v'_{j_2}\}$ be the edge of $\interior(C')$ that connects the ``closest'' vertices with respect to $d_{C'}$. 
An argument similar to the one used above shows that~$f''$ connects two vertices of distance 2 in~$C'$, so we have two cases to analyse. First, if~$f'' \neq \{v_0,v_2\}$, the result follows directly. Otherwise, we either have $|C'| = 4$, in which case~$d_{C'}(v_0,v_2) = d_{C}(v_0,v_2) = 2$, as desired, or we have a chordless cycle $C'' = (v_0,v_3,\ldots,v_{k-1})$ in $G + F$ with at least 4 vertices, on which we can apply the same arguments; as an eventual sequence of cycles emerging from this construction will eventually lead to a cycle of length~$4$, the result holds. 
\Halmos
\endproof

\begin{proposition}
\label{prop:inq4}
If $k \ge 5$, the inequality 
\begin{align*}
\sum_{f \in \interior(C) \: \backslash \: \left\{ \left\{ v_{j-1},v_{j+1} \right\} , \left\{ v_j, v_i \right\}  \right\} } x_{f} \geq  |C| - 4, \quad \textnormal{for all} \; i,j \in \{1,\dots,k\}, d_C(\{v_j,v_i\}) \geq 2 \tag{I4} \label{inq:4} 
\end{align*}
is valid and facet defining for $\textnormal{conv}(X(G))$.
\end{proposition}
\proof{Proof.}
Given vertices $v_i$ and $v_j$ such that $d_C(\{v_j,v_i\}) \geq 2$, inequality~(\ref{inq:4}) enforces the inclusion of at least $|C|-4$ edges of $\interior(C) \setminus \{ \{v_{j-1},v_{j+1}\}, \{v_j,v_i\} \}$
in any chordal completion of $G$. 
Without loss of generality, let $j=0$ and let~$i$ be any value in $[2,k-3]$. Suppose by contradiction that there exists $x^0 \in X(G)$ such that
\begin{eqnarray}\label{prop4:ineq}
\sum_{f \in \interior(C) \: \backslash \: \left\{ \left\{ v_{k-1},v_{1} \right\} , \left\{ v_0, v_i \right\}  \right\} } x^0_{f} < |C| - 4.
\end{eqnarray}

By Lemma~\ref{lem:ThreeInteriorEdges}, we have that 
\[ \sum_{f \in \interior(C)} x^0_f \geq |C|-3. \]
This implies that $x^0_{v_{k-1},v_1} = x^0_{v_0,v_i} = 1$, for otherwise inequality~\ref{prop4:ineq} would be violated. Thus, the sequences 
$C^1 = \left( v_0,v_1,\ldots,v_i \right)$ and $C^2 = \left( v_0,v_i,v_{i+1},\ldots,v_{k-1} \right)$ are cycles in $G(x^0)$.  Again, by Lemma~\ref{lem:ThreeInteriorEdges}, at least $\left|C^\ell\right|-3$ fill edges must be present in $\interior(C^\ell)$, $\ell = 1,2$.  
This is only possible if at least $i-2$ edges of $\interior(C^1)$ and at least $|C|-i-2$ edges of $\interior(C^2)$ belong to the set of fill edges described by~$x^0$.
As $\interior(C^1) \cap \interior(C^2) = \emptyset$ and 
$\interior(C^1) \cup \interior(C^2) \subseteq \interior(C)$, we have that
\[ \sum_{f \in \interior(C) \: \backslash \: \left\{ \left\{ v_{i-1},v_{i+1} \right\} , \left\{ v_0, v_i \right\}  \right\} } x^0_{f} \geq \sum_{f \in \mathrm{int}(C^1)} x^0_f + \sum_{f \in \mathrm{int}(C^2)} x^0_f \geq |C|-4,\]
contradicting thus inequality~(\ref{prop4:ineq}). 
\Halmos
\endproof

\medskip

\begin{example} \label{ex:CycleFacets}
Let~$G$ be a cycle graph associated with 6-cycle $C = (v_0, v_1, v_2, v_3, v_4, v_5)$.
An example of inequality (\ref{inq:2}) with $i = 1$ is
\[ x_{\{v_0,v_2\}} + \left( x_{\{v_1,v_3\}} + x_{\{v_1,v_4\}} + x_{\{v_1,v_5\}} \right) \geq 1, \]
which enforces the inclusion of at least one edge $f  = \{v_1,v_k\}$, $k \in \{3,4,5\}$, for every
chordal completion of~$G$ that does not contain edge $(v_0,v_2)$.


Inequality (\ref{inq:3}) translates to 
\[ x_{\{v_5,v_1\}} + x_{\{v_0,v_2\}} + x_{\{v_1,v_3\}} + x_{\{v_2,v_4\}} + x_{\{v_3,v_5\}} + x_{\{v_4,v_0\}} \geq 2, \]
which enforces that at least 2 of the 6 pairs of vertices that are separated by one vertex must appear in any chordal completion of~$G$.

Finally, inequality (\ref{inq:4}) for $i = 4$ and $j = 1$ is given by
\begin{align*}
x_{\{v_0,v_3\}} + x_{\{v_0,v_4\}} + x_{\{v_1,v_3\}} + x_{\{v_1,v_5\}} +x_{\{v_2,v_4\}} + x_{\{v_2,v_5\}} + x_{\{v_3,v_5\}} \geq 2,
\end{align*}
which enforces that at least 2 of the edges in $\interior(C) \setminus \{  \{v_0,v_2\}, \{v_1,v_4\}  \}$
must be included in any chordal completion of~$G$. 

\end{example}

\section{General Polyhedral Properties}
\label{sec:GeneralPolyhedralProperties}

This section provides theoretical insights into the polyhedral structure of the MCCP polytope.  In particular, the first result, provided in Theorem~\ref{thm:validsubinequality}, shows that any inequality proven to be valid on an induced subgraph 
can be extended into a valid inequality for the original graph.  This result is important for the development of practical solution methodology because it shows that finding valid inequalities/facets on particular substructures, such as cycles, can help in the generation of valid inequalities for larger graphs containing these substructures.  
 Theorem~\ref{thm:FacetsforInducedSubcycles} shows how facets for cycles can be lifted to facets of graphs that are subgraphs of cycles, leading to the result in Corollary~\ref{cor:i2inducedfacet} relating to when the inequalities (\ref{ineq:ci}) in their general form in model (\ref{ipMCCP}) are facet-defining.  The final result of the section, Theorem~\ref{thm:ChordFacets}, proves and describes how facets for small cycles can be lifted to facets of larger cycles.

First, we show a lemma that will be used in the proof of Theorem~\ref{thm:validsubinequality}.

\begin{lemma} \label{lem:InducedSubgraph}
If $G = (V,E)$ is chordal, then $G[W]$ is chordal for any $W \subseteq V$.
\end{lemma}
\proof{Proof.}
A chordless cycle $C$ in $G[W]$ must be a chordless cycle in $G$. Thus, if $G[W]$ is not chordal, then $G$ cannot be as well. 
$\Halmos$
\endproof


\begin{theorem}
\label{thm:validsubinequality}
Let $G = (V,E)$ be an arbitrary graph and $W \subseteq V$ be any subset of vertices.  If $a'x' \ge b$ is a valid inequality for $X(G[W])$, then $ax \geq b$ is a valid inequality for~$X(G)$, where $a_f = a'_f$ if $f \in E^C(G[W])$ and $a_f = 0$ otherwise.
\end{theorem}

\proof{Proof.}
By way of contradiction, suppose that $a'x' \ge b$ is valid for~$X(G[W])$ and let $F$ be a chordal completion of~$G$ such that 
$ax(F) < b$. From the construction of~$a$, we have $ax(F) = a'x(F \cap E(G[W])) < b$, and 
by Lemma~\ref{lem:InducedSubgraph}, $G[W] + F \cap E(G[W])$ must be chordal and, consequently, we must have $a'x(F \cap E(G[W])) \geq b$, establishing thus a contradiction.  $\Halmos$
\endproof

%


We now present a result that goes in the opposite direction of Theorem~\ref{thm:validsubinequality}. Namely, it shows how facet-defining inequalities for a cycle graph~$G$ can be transformed into facet-defining inequalities for subgraphs of~$G$; note that subgraphs of cycle graphs consist of collections of paths. This result allows us to show that inequality~(\ref{ineq:ci}) is facet-defining for subgraphs of cycle graphs.



\begin{theorem}\label{thm:FacetsforInducedSubcycles}
Let $G'=(V,E')$ be  
a cycle graph  associated with the cycle $C = (v_0,v_1,\ldots,v_{k-1}) \in \setcycles$ and $G=(V,E)$ be a subgraph of~$G'$ such that  $G'= G + F_G(C),$  
$F_{G}(C) = E' \setminus E$.
If $ax \geq b$ is facet-defining for $\textnormal{conv}(X(G'))$, $a \geq \textbf{0}$, and
 $a' \in \mathbb{R}^{|E^C|}$, with $a'_f = a_f$ if  $f \in E^C \setminus F_G(C)$ and $a'_f = 0$ otherwise, the inequality
\[ a'x \geq b \left( \sum_{f \in F_G(C)} x_{f} - |F_G(C)| + 1  \right) \]
is facet-defining for $\textnormal{conv}(X(G))$. \Halmos
\end{theorem}


\begin{figure}[t]
\centering
\begin{minipage}{.25\textwidth}
\begin{tikzpicture}[scale=0.75][font=\sffamily,\scriptsize]

\node (1) [draw,circle] [draw] at (-1.2,0) {$v_0$};
\node (2) [draw,circle] [draw] at (1.2,0) {$v_1$};
\node (3) [draw,circle] [draw] at (2,-2) {$v_2$};
\node (4) [draw,circle] [draw] at (0,-3.5) {$v_3$};
\node (5) [draw,circle] [draw] at (-2,-2) {$v_4$};

\node (6) [draw=none] at (0,-4.5) {($a$)};


\path[-](1) edge (2);
\path[dashed](2) edge (3);
\path[dashed](3) edge (4);
\path[-](4) edge (5);
\path[dashed](5) edge (1);
\end{tikzpicture}
\end{minipage}%
\hspace{2.5ex}
\begin{minipage}{.25\textwidth}
\begin{tikzpicture}[scale=0.75][font=\sffamily,\scriptsize]

\node (1) [draw,circle] [draw] at (-1.2,0) {$v_0$};
\node (2) [draw,circle] [draw] at (1.2,0) {$v_1$};
\node (3) [draw,circle] [draw] at (2,-2) {$v_2$};
\node (4) [draw,circle] [draw] at (0,-3.5) {$v_3$};
\node (5) [draw,circle] [draw] at (-2,-2) {$v_4$};

\node (6) [draw=none] at (0,-4.5) {($b$)};

\path[-](1) edge (2);
\path[-](2) edge (3);
\path[-](3) edge (4);
\path[-](4) edge (5);
\path[-](5) edge (1);
\end{tikzpicture}
\end{minipage}

\caption{(a) A graph $G$ with $V = \{v_0,v_1,v_2,v_3,v_4\}$ and $E = \left\{ \{v_0,v_1\},\{v_3,v_4\} \right\}$.  Solid lines represent graph edges and dashed lines are fill edges whose addition to $G$ makes a chordal cycle of length 5. (b) Cycle graph with 5 vertices.}
\label{fig:InducedSubcycle}
\end{figure}
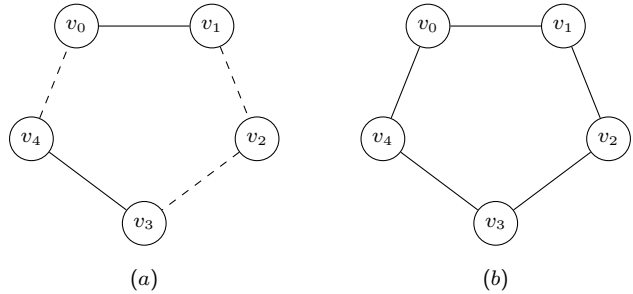

Theorem \ref{thm:FacetsforInducedSubcycles} (proved in Section~\ref{ec:generalProperties}) immediately leads to the following result:
\begin{corollary}
\label{cor:i2inducedfacet}
For any graph $G = (V,E)$ and for any sequence $C \in \setcycles$ such that 
 $\interior(C) \cap E = \emptyset$, the chordal inequality (\ref{ineq:ci}) 
is facet-defining for the MCCP polytope of~$G[V(C)]$.
\end{corollary}

\begin{example} 
\label{ex:InducedSubcycle}
Consider the graph $G$ in Figure~\ref{fig:InducedSubcycle}(a); solid lines represent graph edges in this example. 
As in the statement of Theorem~\ref{thm:FacetsforInducedSubcycles}, we have 
\begin{align*}
& C = \left(v_0,v_1,v_2,v_3,v_4\right), \;E(G) = \left\{ \{v_0,v_1\},\{v_3,v_4\} \right\}, 
\; \textnormal{and} \;  F_G(C) =  \left\{ \{v_1,v_2\},\{v_2,v_3\},\{v_4,v_0\}   \right\}.  \nonumber 
\end{align*}
Graph $G + F(C)$ is 5-chordless cycle.
One facet-defining inequality for this cycle, according to Proposition \ref{prop:inq1}, is the simplified version of the chordal inequality, given by 
\[ x_{\{v_0,v_2\}} + x_{\{v_0,v_3\}} + x_{\{v_1,v_3\}} + x_{\{v_1,v_4\}} + x_{\{v_2,v_4\}} \geq 2.  \]
Corollary~\ref{cor:i2inducedfacet} stipulates that the inequality below is facet defining for $G$:
\begin{align*}
&x_{\{v_0,v_2\}} + x_{\{v_0,v_3\}} + x_{\{v_1,v_3\}} + x_{\{v_1,v_4\}} + x_{\{v_2,v_4\}} \geq 2 \cdot \left(  x_{\{v_1,v_2\}} + x_{\{v_2,v_3\}} + x_{\{v_4,v_0\}} - 3+1 \right).  \nonumber
\end{align*} 
By Theorem~\ref{thm:validsubinequality}, these inequalities will also be valid even if $G$ is a subgraph of a larger graph. 
\end{example}

\medskip


We now define a method for lifting facet-defining inequalities defined on smaller cycles into facet-defining inequalities for large cycles. This is done by considering the inclusion of chords into the inequalities, which reveals a lifting property of MCCPs that can be used to strengthen known inequalities. We present this result in Theorem \ref{thm:ChordFacets}, which is proved in Section~\ref{ec:generalProperties}.

\begin{theorem} \label{thm:ChordFacets}
Let $G=(V,E)$ be a cycle graph associated with the cycle $C = (v_0,v_1,\ldots,v_{k-1})$ 
and let $f^* = \{v_s,v_{t}\} \in \filledges$, $0 \le s < t \le k-1$, be any chord of~$C$. If the inequality $ax \geq b$, $a \geq 0$, is facet-defining for the MCCP polytope of cycle graph $G' = (V',E')$ associated with the cycle 
$C' = (v_s,v_{s+1},\ldots,v_t)$ (i.e., $G' = G[V(C')] + \{f^*\}$), then
\[
a'x \geq b \cdot x_{f^*}
\]
is facet-defining for $\textnormal{conv}(X(G))$,  where 
$a'_f = a_f$ if  $f \in \interior(C')$ and $a'_f = 0$ otherwise. $\Halmos$
\end{theorem}
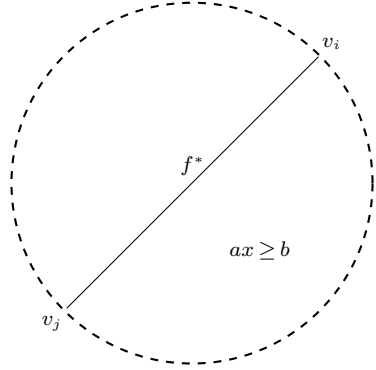
\begin{figure}
\centering
\begin{tikzpicture}[scale=0.6][font=\sffamily,\scriptsize]
\node (i) at (3.1,3.1) {$v_i$};
\node (j) at (-3.1,-3.1) {$v_j$};
\path[-](i) edge node [above] {$f^*$} (j) ;
\node (w) at (1.5,-1.5) {$ax \geq b$};
\draw[thick,dashed] (0,0) circle (4cm);
\end{tikzpicture}
\caption{Depiction of inequality lifting from Theorem~\ref{thm:ChordFacets}.}
\label{fig:ChordFacetProof}
\end{figure}

Figure~\ref{fig:ChordFacetProof} provides a depiction of Theorem~\ref{thm:ChordFacets}.  If $ax \geq b$ is facet-defining for any induced subcycle obtained by adding edge $f^*$, then $ax \geq b \cdot x_{f^*}$ will be facet-defining for the original cycle graph.

\begin{example} \label{ex:ChordFacets}
Consider the graph in Figure \ref{fig:InducedSubcycle}(b), and let~$G' = G[\{v_1, v_2, v_3, v_4\} ] + \{\{v_1,v_4\}\}$ be a cycle graph associated with cycle $C'=(v_1, v_2, v_3, v_4)$. From Proposition~\ref{prop:inq1}, we have that the following chordal inequality is facet-defining for $\textnormal{conv}(X(G'))$:
\begin{align*}
\sum\limits_{f \in \interior(C')}x_f  =  x_{\{v_1,v_3\}} + x_{\{v_2,v_4\}}  \geq |C'|-3 = 1.
\end{align*}
This inequality can be modified in order to become facet-defining for $\textnormal{conv}(X(G))$ through Theorem~\ref{thm:ChordFacets}, yielding
\begin{align*}
x_{\{v_1,v_3\}} + x_{\{v_2,v_4\}} \ge x_{\{v_1,v_4\}} \implies
x_{\{v_1,v_3\}} + x_{\{v_2,v_4\}} - x_{\{v_1,v_4\}} \ge 0.
\end{align*}
\end{example}

\section{Solution Method for the MCCP}
\label{sec:solutionMethod}

We now describe a procedure to solve formulation (\ref{ipMCCP}) for general graphs based on the structural results showed in the previous sections. Since the model has exponentially many constraints, our solution technique is based on a hybrid branch-and-bound procedure that applies separation, lazy-constraint generation, and a primal heuristic to the problem. 

\subsection{Separation complexity}
\label{sec:SeparationComplexity}

Given a graph $G=(V,E)$, we first consider the problem of identifying inequalities violated by a completion of $G$. Recall that, for a vector $x \in \{0,1\}^{m^c}$,
we define $E(x) := \left\{ f \in \filledges: x_f = 1 \right\}$. 

\begin{proposition}
\label{prop:lazycuts}
Let $x^* \in \{0,1\}^{m^c}$ be a solution vector and $G' = (V,E') =  
 G + E(x^*)$. If $G'$ is not chordal, then at least one inequality of each family (\ref{ineq:ci})-(\ref{inq:3}) (and possibly one of family~(\ref{inq:4})) violated by $x^*$ can be found in $O(|V|^3(|E'| + |V|\log|V|))$.
\end{proposition}
\proof{Proof.}
An immediate consequence of Proposition \ref{prop:correctnessIPMCCP} is that some inequality (\ref{ineq:ci}) is violated if and only if $G'$ has a chordless cycle. Such a cycle can be identified according to the following procedure: For each triple of vertices $(v,w,u)$ in~$V$ such that  $\{\{v,w\},\{w,u\}\} \subseteq E'$ and $\{v,u\} \not \in E'$, find the shortest path $p$ between $v$ and $u$ that does not traverse $w$. If such a path exists, cycle $C := (p, w)$ is chordless 
and has length at least 4.
Given such a chordless cycle~$C$, one violated inequality of each type (\ref{ineq:ci})-(\ref{inq:3}) (and type~(\ref{inq:4}) as well if $|C| \ge 5$) can be derived in linear time in $|C|$. Since there are $|V|^3$ triples and the shortest path between $u$ and $v$ (that does not include $w$) can be found using Dijkstra's algorithm in time $O(|E'| + |V|\log|V|)$
\citep{Cormen2009}, the result follows. \Halmos
\endproof

The separation problem of (\ref{ineq:ci})-(\ref{inq:4}) over fractional points is,
however, much more challenging. We state the results below concerning this question.

\begin{theorem}
\label{thm:separation}
Given a fractional point $x^* \in [0,1]^{m^c}$:

$\hspace{2ex}$ a. The separation problem of (\ref{ineq:ci}) is NP-Complete. 

$\hspace{2ex}$ b. Inequalities (\ref{inq:2}) can be separated in $O(|V|^5)$.

$\hspace{2ex}$ c. Inequalities (\ref{inq:3}) can be separated in $O(|V|^8)$.

$\hspace{2ex}$ d. The separation problem of (\ref{inq:4}) is NP-Complete.

\end{theorem}

\begin{proof}{Proof.} Due to space limitations, we present below only proof sketches for these results. The full version of each proof is presented in Section \ref{ec:separationProofsAlgorithms} of the online supplement.

$\hspace{2ex}$ \textbf{a.}  The proof reduces the \emph{quadratic assignment problem} (QSCP), a classical and well-studied NP-hard problem, to the \emph{$\alpha$-quadratic shortest cycle problem} ($\alpha$-QSCP), introduced in this paper. In the QSCP, we are given a graph $G = (V, E)$ and a quadratic cost function $q:V \times V \rightarrow [0,1]$, with $q(u,v) = 0$ if $(u,v) \in E$. A feasible solution of QSCP is a simple chordless cycle $C = (v_1,v_2,\ldots, v_{|C|})$  whose cost is $p(C) = \sum_{\{u,v\} \in E(G[C])^C}q(u,v)   -|C|$. The $\alpha$-QSCP is  the decision version of QSCP  in which the goal is to decide whether~$G$ has  a simple chordless cycle~$C$ such that $p(C) < \alpha$.  We employ a reduction of the quadratic assignment problem to $(-3)$-QSCP that resembles the ones used by~\cite{rostami2015} for the quadratic shortest path problem. Finally, the $-3$-QSCP is reduced to the problem of separating the inequality (\ref{ineq:ci}), completing the proof.

$\hspace{2ex}$ \textbf{b.} An auxiliary graph $G'$, a complete digraph on $|V(G)|$ nodes, is constructed for which the separation problem is reduced to finding, for every triple of vertices $(v_1,v_2,v_3)$, the shortest path from $v_1$ to $v_3$ that does not include $v_2$.   The number of sequences for which this verification needs to be performed is $O(|V(G)|^3)$, and the identification of such a path can be made in time $O(|V(G)|)^2$. 

$\hspace{2ex}$ \textbf{c.} As in \textbf{b}, An auxiliary graph $G'$, specifically a complete digraph on $|V(G)|^2$ nodes, is constructed for which the separation problem is reduced to finding at most ~$O(|V(G)|^4)$ shortest paths, each of which can be performed in polynomial time. 

$\hspace{2ex}$ \textbf{d.} The proof is similar to that of \textbf{a}, except that we use a reduction from $-4$-QSCP*, a slight variant of $-3$-QSCP. $\Halmos$
 
\end{proof}

\subsection{Heuristic separation algorithms}

In view of Proposition \ref{prop:lazycuts} and Theorem \ref{thm:separation}, we tackle model (\ref{ipMCCP}) by applying a typical branch-and-bound procedure that alternates between heuristic separation and lazy-constraint generation.

For the lazy generation part, at every integer node of the branching tree we apply the procedure presented in the proof of Proposition \ref{prop:lazycuts} to separate at least one violated inequality (\ref{ineq:ci})-(\ref{inq:4}), similar to a
combinatorial Benders methodology \citep{Codato2006}. 
Propositions \ref{prop:correctnessIPMCCP} and \ref{prop:lazycuts} 
ensures that this approach 
yields an (feasible and) optimal solution to (\ref{ipMCCP}), since a violated inequality is not found if and only if the resulting graph is chordal.

Nonetheless, adding violated inequalities only at integer points typically yield weak bounds at intermediate nodes of the branching tree. Since a complete separation
of fractional points is not viable due to Theorem \ref{thm:separation}, we consider a heuristic \textit{threshold} procedure. Namely, 
given a point $x^* \in [0,1]^{m^c}$ and a threshold $\delta \in (0,1)$, let
\[
E^\delta(x) := \left\{ f \in \filledges: x_f \ge \delta \right\}.
\]
We can use the procedure from Proposition \ref{prop:lazycuts} to find violating inequalities for the graph $G+E^\delta(x^*)$. Such inequalities may not be necessarily violated by $x^*$, and require thus a (simple) extra verification testing step. Even though the threshold policy does not guarantee that at least one violated inequality is found, it can be performed efficiently and, as our numerical experiments indicate, it is a fundamental component for the good performance of the proposed solution technique. 

\subsection{Primal Heuristic}
\label{sec:PrimalHeuristic}

We have also incorporated a primal heuristic to be applied at infeasible integer
nodes of the branching tree. The method is based on the state-of-the-art heuristic for the problem, designed by \cite{GeoLiu1989}. Specifically, the vertices of the graph
are sorted in ascending order according to their 
degree, thereby defining a sequence $S = (v_1, v_2, \dots, v_{|V|})$. The vertices are then picked one at a time, in the order indicated by~$S$. For each vertex $v_i$, edges are added to $G$ so that $S$ defines
a \textit{perfect elimination ordering}, i.e., $v_i$ and its neighbours on 
set $\{v_{i+1}, v_{i+2}, \dots, v_{|V|}\}$ induce a clique, which makes $G$ chordal.
This procedure has complexity $O(|V|^2\,|E|)$.

For any integer point $x \in \{0,1\}^{m^c}$ found during the branch-and-bound procedure, if
$G' := G+E(x)$ is not chordal, we can apply \cite{GeoLiu1989}'s heuristic in order to chordalize~$G'$ and obtain a
feasible solution to the problem.
The application of this procedure at the root node ensures we can identify solutions which are at least as good as those provided by the heuristic.

\section{Numerical Experiments}
\label{sec:numericalExperiments}

In this section we present an experimental evaluation of the solution methods introduced in this paper.  
The experiments ran on an Intel(R) Xeon(R) CPU E5-2640 v3 at 2.60GHz
with 128 GB RAM. We used the integer programming solver IBM
ILOG CPLEX 12.6.3 \citep{CPLEXRef} in all experiments, with a time-limit of 3,600 seconds and one thread.

\subsection{Instances}
\label{sec:Instance}

Four family of instances were used for the experimental evaluation: \emph{relaxed caveman graphs}, \emph{grid graphs}, \emph{queen graphs}, and \emph{DIMACS} graphs. They are described as follows.

Relaxed caveman graphs \citep{Judd2011} represent typical social networks, where small pockets of individuals are tightly connected and have sporadic connections to other groups. This family of instances has been employed  previously in the evaluation of algorithms for combinatorial optimization problems \citep{BerCir2016}.  Each instance is generated randomly based on three parameters, $\alpha, \beta \in \mathbb{Z}^+$ and $\gamma \in (0,1)$.  Starting from a set of $\beta$ disjoint cliques of size $\alpha$, each edge is examined and, with probability $\gamma$, one of its endpoints is switched to a vertex belonging to another clique; all operations are made uniformly at random.  An example of a relaxed caveman graph is depicted in Figure~\ref{fig:caveman}, where $\alpha = \beta = 6$ and $\gamma = 0.2$. 
\begin{figure}[h!]
\centering
\includegraphics[scale=0.40]{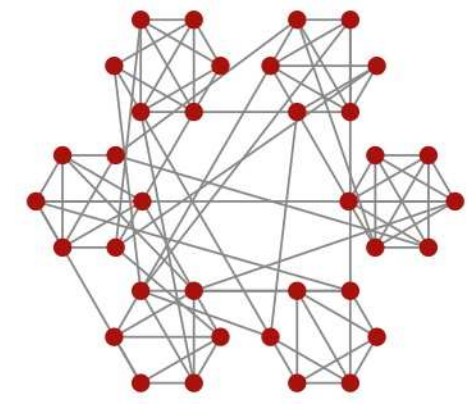}
\caption{Relaxed caveman graph. Picture from Judd et al \cite{Judd2011}.}
\label{fig:caveman}
\end{figure}

The structure of relaxed caveman graphs is particularly useful for evaluating algorithms for chordal completions. Namely, the modifications in the edges lead to large chordless cycles, enforcing thus the inclusion of several edges in chordal completions. For our experiments, ten instances of each possible configuration involving $\alpha,\beta \in \{4,5,6,7,8\}$ and $\gamma = 0.30$  were generated.  This set of graphs will be henceforth denoted simply by caveman instances.

The next set of instances, \emph{grid graphs}, correspond to graphs whose node set can be partitioned into a set of $r$ \textit{rows} $R_1,\dots,R_{r}$ and $c$ \textit{columns} $C_1,\dots,C_{c}$. Each vertex is denoted by $v_{r,c}$ if $v_{r,c} \in R_r \cap C_c$.  Vertices $v_{i,j}$ and $v_{k,l}$ are adjacent if and only if either $i = k$ and $|j - l| = 1$, or  $j = l$ and $|i - k| = 1$.  Note that grid graphs also contain large chordless cycles. 



We also used \emph{queen graphs} for our experiments. The queen graphs are extensions of grid graphs with additional edges representing longer hops as well as diagonal movements.  More precisely, there exists an edge connecting $v_{i,j}$ to $v_{i',j'}$ if and only if one of the three following conditions is satisfied for some $k$: (1) $i = i' \pm k$ and $j = j' \pm k$, (2) $i = i'$ and $j = j' \pm k$, or (3) $i = i' \pm k$ and $j = j'$.  The configurations of grid graphs and of queen graphs used in our experiments are equivalent to those used by \cite{Yuceoglu2015}.

The final set of instances consists of the classical DIMACS graph coloring instances, which can be downloaded from \href{http://dimacs.rutgers.edu/Challenges/}{http://dimacs.rutgers.edu/Challenges/}. These instances are frequently used in computational evaluations of graph algorithms.  

\subsection{Other Approaches}
\label{sec:OtherAlgos}

The state-of-the-art approaches to the MCCP reported in the literature are a branch-and-cut approach by \cite{Yuceoglu2015} and a Benders decomposition approach by \cite{BerRaghu15}, henceforth denoted by \texttt{YUC} and \texttt{BEN}, respectively.  

\texttt{YUC} is based on a perfect elimination ordering (PEO) model of the MCCP. The model finds a PEO that minimizes the number of fill-in edges, and can be written as follows.
\begin{align}
\min        \quad         &\sum_{(i,j) \in {E^c}} y_{ij} + y_{ji} \nonumber \\
\textnormal{s.t.}\quad    &x_{ij} + x_{ji} = 1, \quad &\textnormal{for all $\{i,j\}  \in {E}$} 
\label{peo1} \\
                          &x_{ij} + x_{ji} \le 1, \quad &\textnormal{for all $\{i,j\} \in {E^c}$} \label{peo2} \\
                          &x_{ij} + x_{jk} - x_{ik} \le 1, \quad &\textnormal{for all $i,j,k \in V, i \neq j, j \neq k, i \neq k$} \label{peo3} \\
                          &y_{ij} \le x_{ij}, y_{ji} \le x_{ji}, \quad &\textnormal{for all $i,j \in V, i \neq j$} \label{peo4} \\
                          &y_{ij} = x_{ij}, y_{ji} = x_{ji}, \quad &\textnormal{for all $\{i,j\} \in E$} \label{peo5} \\
                          &x_{jk} + y_{ij} + y_{ik} - y_{jk} \le 2, \quad &\textnormal{for all $i,j,k \in V, i \neq j, j \neq k, i \neq k$} \label{peo6} \\
                          &y_{ij}, x_{ij} \in \{0,1\}, \quad &\textnormal{for all $i,j \in V, i \neq j$} \label{peo7}
\end{align}
In the model above, denoted by \texttt{PEO}, a binary variable $y_{ij}$ indicates whether edge $\{i,j\}$ is added to  $G$ and binary variable $x_{ij}$ indicates whether vertex $i$ precedes $j$ in the resulting ordering. Constraints (\ref{peo1}) enforce the existence of a precedence relation between vertices $i$ and $j$ if $\{i,j\} \in E$, whereas constraints (\ref{peo2}) prevent $i$ and~$j$ from preceding each other  simultaneously in an elimination ordering. Constraints (\ref{peo3}) ensure the transitive closure of precedence relations is satisfied. Constraints (\ref{peo4}) and (\ref{peo5}) indicate that a precedence relation between edges $i$ and $j$ can exist if and only if $\{i,j\} \in E$. Finally, constraints (\ref{peo6}) impose that the final ordering must be a perfect elimination ordering. 

\texttt{YUC} employs a branch-and-cut approach that is based on a polyhedral analysis of the convex hull of solutions to \texttt{PEO}. Additionally, the algorithm also considers valid inequalities for special structured graphs, such as grid and queen graphs.

The other approach tested against is \texttt{BEN}, the precursor of the approach described in the present work.  In \cite{BerRaghu15}, a formulation consisting only of inequalities (\ref{ineq:ci}) is used in a pure Benders decomposition approach. That is, \texttt{BEN} solves to optimality an IP using the current set of inequalities (\ref{ineq:ci}) (i.e., the Benders cuts) found up to the current iteration.  If the solution contains no chordless cycles, its optimality is proven and the procedure stops. Otherwise, a collection of chordless cycles is found and new Benders cuts (i.e., inequalities (\ref{ineq:ci}) violated by the current solution) are added to the model, and the procedure repeats. 

Also of interest is to compare our approach with state-of-the-art heuristics in terms of solution quality, assessing thus how significant the differences between exact and heuristic solutions are. We consider the state-of-the-art heuristic developed 
by \cite{GeoLiu1989} (described in Section \ref{sec:PrimalHeuristic}), which will be henceforth denoted by \texttt{MDO}.

The methodology proposed in this paper will be henceforth denoted by \texttt{BC}, as
it can also be classified as a \textit{branch-and-cut} algorithm.

\subsection{Algorithmic Enhancements}
\label{sec:AlgorithmicEnhancements}

In the first set of experiments, we test the following algorithmic enhancements to \texttt{BC}: (1) using  only inequalities (\ref{ineq:ci}) versus using all inequalities (\ref{ineq:ci})-(\ref{inq:4}); (2) 
separating the inequalities only at integer solutions or at each search-tree node; and (3) 
invoking \texttt{MDO} as a primal heuristic. In particular, \texttt{BC-Base} is an implementation of \texttt{BC} where only inequality (\ref{ineq:ci}) is considered, the separation algorithm is invoked only at integer search-tree nodes, and no primal heuristic is applied.  \texttt{BC-Enh} is an implementation of \texttt{BC} where all enhancements are applied. The caveman graphs are used for this evaluation.

In order to verify whether an individual enhancement leads to a statistically significant reduction in the solution times, a two-sample paired t-test was employed. Specifically, the null hypothesis indicates whether the solutions times are  equivalent with or without the enhancements. 
Solution times can differ by orders of magnitude across the 250 instances of caveman graphs 
(e.g., 0.001 seconds versus nearly 3,600 seconds), so all comparisons were made in logarithmic scale, i.e., we applied $\log(1+t)$ transformations to run times~$t$ (given in seconds).

For enhancement (1), (2), and (3) taken individually, the tests resulted in a p-value of 0.070, 0.00057, and 0.0013, respectively.  
This shows that each enhancement provides considerable reductions in run time, with the heuristic being
perhaps the most effective among them.
When comparing \texttt{BC-Base} versus \texttt{BC-Enh}, the test yields a p-value 0.0000067, showing strong statistical significance of the results indicating reductions on run times caused by the enhancements.

A plot comparing the solution times of \texttt{BC-Base} and \texttt{BC-Enh} is provided in Figure~\ref{fig:enhancements}.  Each point in the scatter plot of Figure~\ref{fig:enhancements}(a) corresponds to one instance, with the radii  indicating the sizes of the cliques and the color representing the number of cliques (dark red/blue corresponding to instances with smallest/largest number of cliques, respectively).  The $x$-axis is the run time in seconds in log-scale for \texttt{BC-Base} and the $y$-axis contains  the respective values for \texttt{BC-Enh}. Figure~\ref{fig:enhancements}(b) presents the cumulative distribution plot of performance of both algorithms, indicating in the $y$-axis how many instances were solved within the amount of time indicated in the $x$-axis. Both figures shows that \texttt{BC-Enh} typically outperforms \texttt{BC-Base}, which becomes more prominent with harder instances.

\begin{figure}[h!]
\centering
\begin{subfigure}{0.45\textwidth}
\centering
\includegraphics[height=26ex]{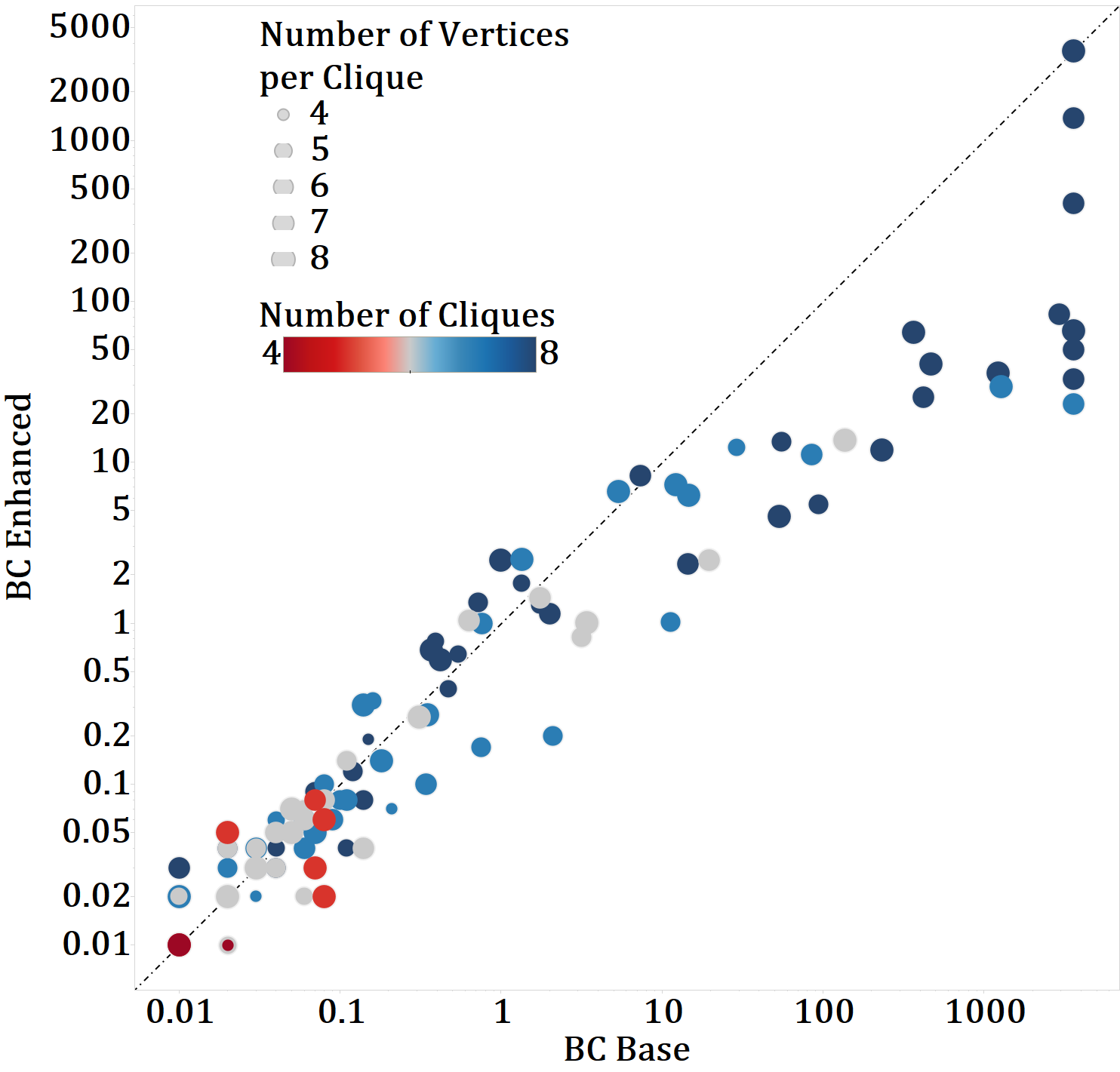}
\caption{Scatter plot}
\end{subfigure}
\hspace{5ex}
\begin{subfigure}{0.45\textwidth}
\centering
\includegraphics[height=26ex]{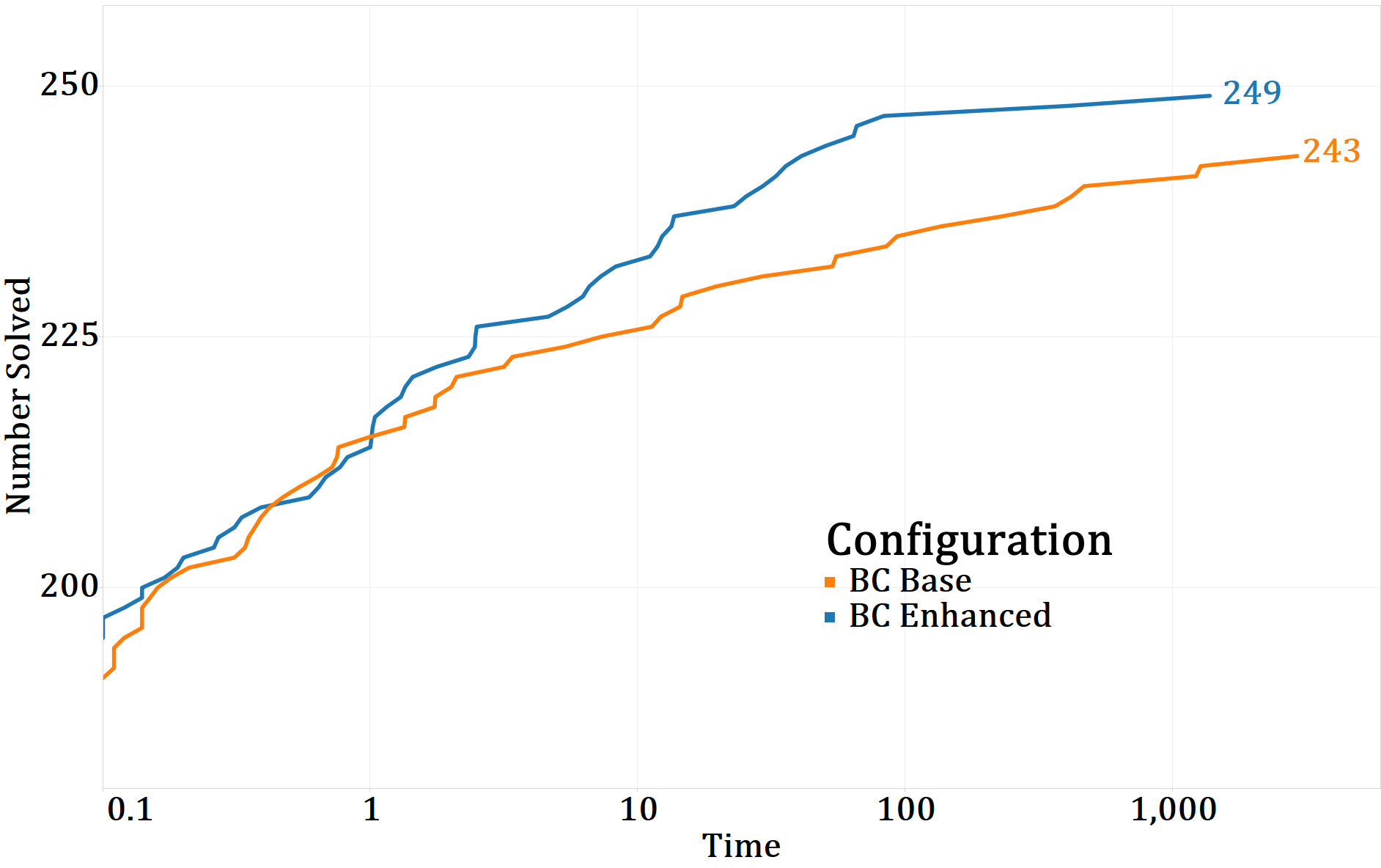}
\caption{Cumulative distribution plots of performance}
\end{subfigure}
\caption{Comparing \texttt{BC} with and without algorithmic enhancements on caveman graphs.  (a) Scatter plot comparing solution times, with the size of the dot corresponding to the number of vertices in each clique and the color (red to blue) the relative number of cliques. (b) Cumulative distribution plot of performance showing number of instances solved by a certain time limit. } 
\label{fig:enhancements}
\end{figure}

The algorithmic enhancements provide clear computational advantages, and so we use them for the remaining experiments, referring to \texttt{BC-Enh} simply as \texttt{BC}.

\subsection{Comparison with Heuristics} 
\label{sec:compHeur}

In the second set of experiments, we investigate the  improvements brought by   
\texttt{BC} upon the solutions obtained alone by \texttt{MDO}. Note that, since \texttt{BC} uses \texttt{MDO} at the root and throughout search, the solutions will always be at least as good as  those achieved by \texttt{MDO}.

Our first comparison, involving the   
relaxed caveman graphs, is presented in Table~\ref{table:cavemanHeur}.  For each configuration, the averages of the upper bounds provided by both \texttt{BC} and \texttt{MDO} are presented, as well as the average percentage decreases in the upper bound from \texttt{MDO} to \texttt{BC} and the average times to compute the chordal completion. As \texttt{MDO} achieves solutions in under a hundredth of a second in all cases, its running times are not reported.  These results readily show the advantage of seeking optimal solutions.  The heuristic can be far from the optimal solution (eventually by up to 70\%) and, on the relaxed caveman graphs, the running time of \texttt{BC} 
is almost always small (only one instance is not proven optimal within 3,600 seconds).

\begin{table}[h]
\begin{footnotesize}
\begin{center}
\caption{Comparison between \texttt{MDO} and \texttt{BC} on caveman graphs.} \label{table:cavemanHeur}
{\def\arraystretch{0.65}
\begin{tabular}{ C{10ex} C{10ex} C{15ex} C{15ex} C{15ex}  C{15ex} }
$\alpha$	&	$\beta$	&	\texttt{BC} UB	&	\texttt{MDO} UB	&	\% Dec	&	\texttt{BC} Time	\\ \hline \hline
4	&	4	&	0.9	&	1.4	&	15	&	0	\\
4	&	5	&	1.7	&	2.8	&	35.1	&	0	\\
4	&	6	&	5.3	&	7.4	&	42.8	&	0	\\
4	&	7	&	13.2	&	17.3	&	25.1	&	0	\\
4	&	8	&	19	&	26.3	&	29.3	&	0	\\ \hline
5	&	4	&	1.6	&	1.8	&	6.3	&	0	\\
5	&	5	&	1.7	&	4.4	&	51.8	&	0	\\
5	&	6	&	7.7	&	10.3	&	40.3	&	0	\\
5	&	7	&	15.5	&	21.7	&	35.7	&	1.3	\\
5	&	8	&	25.1	&	37.1	&	34.8	&	0.5	\\ \hline
6	&	4	&	0.6	&	1.7	&	32.1	&	0	\\
6	&	5	&	4.1	&	8.4	&	54.9	&	0	\\
6	&	6	&	12.2	&	16.9	&	35.4	&	0.1	\\
6	&	7	&	18.1	&	24.2	&	25.3	&	0.1	\\
6	&	8	&	28.9	&	43	&	35.8	&	2.0	\\ \hline
7	&	4	&	1.1	&	2.5	&	50.8	&	0	\\
7	&	5	&	3.6	&	6.5	&	42.9	&	0	\\
7	&	6	&	17.1	&	23.6	&	35.8	&	0.5	\\
7	&	7	&	24.8	&	35.4	&	37.6	&	3.5	\\
7	&	8	&	55.9	&	74.1	&	27.6	&	198.9	\\ \hline
8	&	4	&	0.3	&	3.3	&	70.8	&	0	\\
8	&	5	&	5.5	&	9.9	&	58.7	&	0	\\
8	&	6	&	14.6	&	23.3	&	39.9	&	1.5	\\
8	&	7	&	29.7	&	45	&	37.6	&	5.2	\\
8	&	8	&	52.4	&	69.8	&	26.2	&	382.2	\\
\end{tabular}
}
\end{center}
\end{footnotesize}
\end{table}

The same data for grid graphs, queen graphs, and DIMACS graphs are presented in Tables~\ref{table:grid}, \ref{table:queen}, and~\ref{table:dimacs}, respectively.
These tables first report the graph characteristics, including number of vertices and number of edges, and for all algorithms the resulting lower bounds, upper bounds, and solution times (in seconds if solved to optimality in 3,600s, or a mark "-" otherwise), respectively. Bold face for the upper bound indicates that the algorithm found the best-known solution for the instance.
In general, \texttt{BC} can deliver better solutions than \texttt{MDO} but requires much more time in harder instances, showing therefore the trade-off in computational time and solution quality.

\subsection{Comparison with Other Techniques}
	\label{sec:yucComparison}

This section provides a comparison of \texttt{BC} with \texttt{YUC} and \texttt{BEN}. For these evaluations, 
we employed all instances reported upon in \cite{Yuceoglu2015} and \cite{BerRaghu15}
and compared the solution times and objective function bounds obtained  by all solution methods. The reported numbers for \texttt{YUC} were obtained directly from \cite{Yuceoglu2015}, who uses IBM ILOG CPLEX 12.2 and a processor with similar clock (2.53 GHz), but runs the parallel version of the solver with 4 cores, and not with 1 core, as we do in this work. 

First, we report on grid graphs, which were used in the computational results of \cite{Yuceoglu2015} for~\texttt{YUC}. This approach enhances the \texttt{PEO} formulation with cuts tailored for graphs containing grid structures, so these instances are particularly well-suited for \texttt{YUC}. 
The results are presented in Table~\ref{table:grid}. \texttt{BC} typically finds the best-known solutions, and only in 4 cases out of 22 the relaxation bound  for \texttt{YUC} outperforms that of \texttt{BC}. 

\begin{table}[h]
\begin{footnotesize}
\begin{center}
\caption{Results for grid graphs.  Note: Graph grid5\_5 was not included in the results reported on in \cite{Yuceoglu2015}. }\label{table:grid}
{\def\arraystretch{0.65}
\begin{tabular}{ C{9ex} C{9ex} C{9ex} C{1ex} C{9ex} C{9ex} C{9ex} C{1ex} C{9ex} C{9ex} C{9ex} C{1ex} C{9ex} }
 & Instance & & & & \texttt{YUC} & & & & \texttt{BC}  & &  & \texttt{MDO} \\
 name & $|V|$ & $|E|$ & & $LB$ & $UB$ & $t$ & &  $LB$ & $UB$ & $t$ && $UB$ \\
 \cmidrule{1-3} \cmidrule{5-7} \cmidrule{9-11} \cmidrule{13-13} 
  grid3\_3  	&	9	&	12	&	&	5	&	\textbf{	5	}	&	0.01	&	&	5	&	\textbf{	5	}	&	0	&&	5	\\
  grid3\_4 	&	12	&	17	&	&	9	&	\textbf{	9	}	&	0	&	&	9	&	\textbf{	9	}	&	0.01	&&	9	\\
  grid3\_5 	&	15	&	22	&	&	13	&	\textbf{	13	}	&	0.02	&	&	13	&	\textbf{	13	}	&	0.07	&&	13	\\
  grid3\_6 	&	18	&	27	&	&	17	&	\textbf{	17	}	&	0.02	&	&	17	&	\textbf{	17	}	&	0.17	&&	17	\\
  grid3\_7 	&	21	&	32	&	&	21	&	\textbf{	21	}	&	0.01	&	&	21	&	\textbf{	21	}	&	0.22	&&	21	\\
  grid3\_8 	&	24	&	37	&	&	25	&	\textbf{	25	}	&	0.02	&	&	25	&	\textbf{	25	}	&	1.39	&&	25	\\
  grid3\_9 	&	27	&	42	&	&	29	&	\textbf{	29	}	&	0.02	&	&	29	&	\textbf{	29	}			&	9.13	&&	33	\\
  grid3\_10 	&	30	&	47	&	&	33	&	\textbf{	33	}	&	0.03	&	&	33	&	\textbf{	33	}	&	20.39	&&	37	\\
  grid4\_4 	&	16	&	24	&	&	18	&	\textbf{	18	}	&	1.23	&	&	18	&	\textbf{	18	}	&	2.33	&&	18	\\
  grid4\_5 	&	20	&	31	&	&	25	&	\textbf{	25	}	&	18.11	&	&	25	&	\textbf{	25	}	&	8.35	&&	25	\\
  grid4\_6 	&	24	&	38	&	&	32.2	&	\textbf{	34	}	&	-	&	&	34	&	\textbf{	34	}	&	216.71	&&	34	\\
  grid4\_7 	&	28	&	45	&	&	39	&	\textbf{	41	}	&	-	&	&	41	&	\textbf{	41	}	&	304.85	&&	41	\\
  grid4\_8 	&	32	&	52	&	&	45.5	&		52		&	-	&	&	48.2	&	\textbf{	50	}	&	-	&&	50	\\
  grid4\_9 	&	36	&	59	&	&	52.5	&		58		&	-	&	&	54.2	&	\textbf{	57	}	&	-	&&	57	\\
  grid4\_10  	&	40	&	66	&	&	59.3	&		66		&	-	&	&	59.1	&		66		& -	&&	66	\\
  grid5\_5	&	25	&	40	&	&	*	&		*		&	*	&	&	37	&		37		&	115.98	&&	37	\\
  grid5\_6 	&	30	&	49	&	&	46.2	&		53		&	-	&	&	48.7	&	\textbf{	50	}	&	-	&&	52	\\
  grid5\_7 	&	35	&	58	&	&	56.9	&		65		&	-	&	&	56.5	&	\textbf{	62	}	&	-	&&	68	\\
  grid5\_8 	&	40	&	67	&	&	67.5	&		77		&	-	&	&	65.2	&	\textbf{	75	}	&	-	&&	80	\\
  grid5\_9  	&	45	&	76	&	&	33.3	&		90		&	-	&	&	73.2	&	\textbf{	89	}	&	-	&&	93	\\
  grid6\_6 	&	36	&	60	&	&	60.9	&		77		&	-	&	&	59.2	&	\textbf{	69	}	&	-	&&	71	\\
  grid6\_7 	&	42	&	71	&	&	31	&		94		&	-	&	&	69.1	&	\textbf{	88	}	&	-	&&	92	\\
  grid7\_7 	&	49	&	84	&	&	37	&		125		&	-	&	&	80.3	&	\textbf{	112	}	&	-	&&	119	\\
\end{tabular}
}
\end{center}
\end{footnotesize}
\end{table}

\begin{table}[h]
\begin{footnotesize}
\begin{center}
\caption{Results for queen graphs. } \label{table:queen}
{\def\arraystretch{0.65}
\begin{tabular}{ C{12ex} C{9ex} C{9ex} C{1ex} C{9ex} C{9ex} C{9ex} C{1ex} C{9ex} C{9ex} C{9ex} C{1ex} C{9ex} }
 & Instance & & & & \texttt{YUC} & & & & \texttt{BC}  & &  & \texttt{MDO} \\
 name & $|V|$ & $|E|$ & & $LB$ & $UB$ & $t$ & &  $LB$ & $UB$ & $t$ && $UB$ \\
 \cmidrule{1-3} \cmidrule{5-7} \cmidrule{9-11} \cmidrule{13-13} 
queen3\_3	&	9	&	28	&&	5	&	\textbf{	5	}	&	0	&&	5	&	\textbf{	5	}	&	0	&&	5	\\
queen3\_4	&	12	&	46	&&	12	&	\textbf{	12	}	&	0.01	&&	12	&	\textbf{	12	}	&	0.01	&&	12	\\
queen3\_5	&	15	&	67	&&	22	&	\textbf{	22	}	&	0.31	&&	22	&	\textbf{	22	}	&	0.03	&&	22	\\
queen3\_6	&	18	&	91	&&	36	&	\textbf{	36	}	&	1.03	&&	36	&	\textbf{	36	}	&	0.25	&&	36	\\
queen3\_7	&	21	&	118	&&	53	&	\textbf{	53	}	&	2.17	&&	53	&	\textbf{	53	}	&	0.91	&&	53	\\
queen3\_8	&	24	&	148	&&	74	&	\textbf{	74	}	&	8.49	&&	74	&	\textbf{	74	}	&	2.27	&&	74	\\
queen3\_9	&	27	&	181	&&	98	&	\textbf{	98	}	&	15.77	&&	98	&	\textbf{	98	}	&	4.97	&&	98	\\
queen3\_10	&	30	&	217	&&	126	&	\textbf{	126	}	&	65.91	&&	126	&	\textbf{	126	}	&	22.29	&&	126	\\
queen4\_4	&	16	&	76	&&	26	&	\textbf{	26	}	&	0.19	&&	26	&	\textbf{	26	}	&	0.03	&&	28	\\
queen4\_5	&	20	&	110	&&	51	&	\textbf{	51	}	&	4.54	&&	51	&	\textbf{	51	}	&	0.75	&&	53	\\
queen4\_6	&	24	&	148	&&	83	&	\textbf{	83	}	&	16.54	&&	83	&	\textbf{	83	}	&	6.68	&&	83	\\
queen4\_7	&	28	&	190	&&	119	&	\textbf{	119	}	&	68.22	&&	119	&	\textbf{	119	}	&	34.5	&&	121	\\
queen4\_8	&	32	&	236	&&	164	&	\textbf{	164	}	&	636.28	&&	164	&	\textbf{	164	}	&	445.78	&&	167	\\
queen4\_9	&	36	&	286	&&	209.8	&	\textbf{	217	}	&	-	&&	211.4	&	\textbf{	217	}	&	-	&&	222	\\
queen4\_10	&	40	&	340	&&	255.5	&	\textbf{	278	}	&	-	&&	259.7	&	\textbf{	278	}	&	-	&&	286	\\
queen5\_5	&	25	&	160	&&	93	&	\textbf{	93	}	&	41.03	&&	93	&	\textbf{	93	}	&	14.02	&&	94	\\
queen5\_6	&	30	&	215	&&	144	&	\textbf{	144	}	&	185.81	&&	144	&	\textbf{	144	}	&	186.93	&&	154	\\
queen5\_7	&	35	&	275	&&	203.1	&	\textbf{	214	}	&	-	&&	204.2	&	\textbf{	214	}	&	-	&&	223	\\
queen5\_8	&	40	&	340	&&	265.8	&	\textbf{	293	}	&	-	&&	265.5	&	\textbf{	293	}	&	-	&&	306	\\
queen5\_9	&	45	&	410	&&	339.8	&		393		&	-	&&	338.8	&	\textbf{	386	}	&	-	&&	398	\\
queen5\_10	&	50	&	485	&&	424.9	&		501		&	-	&&	424	&	\textbf{	492	}	&	-	&&	503	\\
queen6\_6	&	36	&	290	&&	214.9	&		232		&	-	&&	218.1	&	\textbf{	231	}	&	-	&&	244	\\
queen6\_7	&	42	&	371	&&	299.2	&		351		&	-	&&	296.4	&	\textbf{	338	}	&	-	&&	352	\\
queen6\_8	&	48	&	458	&&	400.7	&		481		&	-	&&	396.5	&	\textbf{	461	}	&	-	&&	482	\\
queen6\_9	&	54	&	551	&&	521.4	&		622		&	-	&&	514.3	&	\textbf{	619	}	&	-	&&	633	\\
queen6\_10	&	60	&	650	&&	656.7	&	\textbf{	786	}	&	-	&&	646.6	&		787		&	-	&&	826	\\
queen7\_7	&	49	&	476	&&	423.7	&		520		&	-	&&	422.3	&	\textbf{	495	}	&	-	&&	515	\\
queen7\_8	&	56	&	588	&&	577.6	&		710		&	-	&&	567.2	&	\textbf{	680	}	&	-	&&	687	\\
queen7\_9	&	63	&	707	&&	751.8	&		935		&	-	&&	736.8	&	\textbf{	897	}	&	-	&&	919	\\
queen7\_10	&	70	&	833	&&	948.5	&		1177		&	-	&&	926.6	&	\textbf{	1141	}	&	-	&&	1149	\\
queen8\_8	&	64	&	728	&&	782.1	&		965		&	-	&&	766.9	&	\textbf{	939	}	&	-	&&	970
\end{tabular}
}
\end{center}
\end{footnotesize}
\end{table}

Next, Table~\ref{table:queen} reports on queen graphs. 
As previously mentioned, these instances are also  well-suited to \texttt{YUC} 
because of their grid-like structures.  Nonetheless, the results show that \texttt{BC} typically outperforms \texttt{YUC} both in terms of optimality gap and solution time. In particular, for almost all instances, the obtained solution is at least as good as the one found by \texttt{YUC}. In the only exception, the solution obtained by \texttt{BC} contains only one fill edge more than \texttt{YUC}.  

Finally, Table~\ref{table:dimacs} reports on DIMACS graphs.  The 12 instances above the double horizontal line are those reported on in \cite{Yuceoglu2015}, whereas the others are the remaining graphs in the benchmark set with  fewer than 150 vertices. Our results show that instances of the first group are solved orders of magnitude faster by \texttt{BC} and, for those in which \texttt{YUC} was not able to prove optimality,  better objective function bounds are obtained. In particular, \texttt{BC} was able to close entirely the optimality gap of four instances of this dataset that were still open: \texttt{david}, \texttt{miles250}, \texttt{miles750}, and \texttt{myciel5}.
These results can be explained by the fact that DIMACS graphs do not necessarily have grid-like structures, which makes them more challenging  for \texttt{YUC}. 
For the remaining instances, 9 are solved to optimality and for many of the other instances, the best solutions obtained by \texttt{BC} employed substantially fewer fill edges than those obtained by the traditional heuristic \texttt{MDO}.

\begin{table}[h!]
\begin{footnotesize}
\begin{center}
\caption{Results for dimacs graphs.} \label{table:dimacs}
{\def\arraystretch{0.70}
\begin{tabular}{ C{15ex} C{7ex} C{7ex} C{1ex} C{9ex} C{9ex} C{9ex} C{1ex} C{9ex} C{9ex} C{9ex} C{1ex} C{9ex} }
 & Instance & & & & \texttt{YUC} & & & & \texttt{BC}  & &  & \texttt{MDO} \\
 name & $|V|$ & $|E|$ & & $LB$ & $UB$ & $t$ & &  $LB$ & $UB$ & $t$ && $UB$ \\
 \cmidrule{1-3} \cmidrule{5-7} \cmidrule{9-11} \cmidrule{13-13} 
anna	&	138	&	493	&&	47	&	\textbf{	47	}	&	1386.04	&&	47		&	\textbf{	47		}	&	1.02	&&	47	\\
david	&	87	&	406	&&	59.5	&		65		&	-	&&	64		&	\textbf{	64		}	&	0.4	&&	66	\\
games120	&	120	&	638	&&	496.4	&		1626		&	-	&&	886.7		&	\textbf{	1503		}	&	-	&&	1513	\\
huck	&	74	&	301	&&	5	&	\textbf{	5	}	&	2.92	&&	5		&	\textbf{	5		}	&	0.04	&&	9	\\
jean	&	80	&	254	&&	16	&	\textbf{	16	}	&	6.13	&&	16		&	\textbf{	16		}	&	0.09	&&	19	\\
miles250	&	128	&	387	&&	45.7	&		61		&	-	&&	53		&	\textbf{	53		}	&	0.4	&&	61	\\
miles500	&	128	&	1170	&&	196.4	&		447		&	-	&&	327.487		&	\textbf{	376		}	&	-	&&	446	\\
miles750	&	128	&	2113	&&	352.1	&		954		&	-	&&	471		&	\textbf{	471		}	&	537.65	&&	723	\\
myciel3	&	11	&	20	&&	10	&	\textbf{	10	}	&	0	&&	10		&	\textbf{	10		}	&	0	&&	10	\\
myciel4	&	23	&	71	&&	46	&	\textbf{	46	}	&	0.06	&&	46		&	\textbf{	46		}	&	0.03	&&	46	\\
myciel5	&	47	&	236	&&	189.7	&		197		&	-	&&	196		&	\textbf{	196		}	&	28.93	&&	197	\\ \hline \hline
1-FullIns\_3	&	30	&	100	&&		&				&		&&	80		&	\textbf{	80		}	&	2.42	&&	80	\\
1-FullIns\_4	&	93	&	593	&&		&				&		&&	657.9		&	\textbf{	785		}	&	-	&&	839	\\
1-Insertions\_4	&	67	&	232	&&		&				&		&&	303.6		&	\textbf{	365		}	&	-	&&	394	\\
2-FullIns\_3	&	52	&	201	&&		&				&		&&	230.4		&	\textbf{	248		}	&	-	&&	273	\\
2-Insertions\_3	&	37	&	72	&&		&				&		&&	85.1		&	\textbf{	99		}	&	-	&&	103	\\
2-Insertions\_4	&	149	&	541	&&		&				&		&&	659.3		&	\textbf{	1585		}	&	-	&&	1588	\\
3-FullIns\_3	&	80	&	346	&&		&				&		&&	407.1		&	\textbf{	577		}	&	-	&&	661	\\
3-Insertions\_3	&	56	&	110	&&		&				&		&&	118.4		&	\textbf{	192		}	&	-	&&	198	\\
4-FullIns\_3	&	114	&	541	&&		&				&		&&	691.8		&	\textbf{	1094		}	&	-	&&	1274	\\
4-Insertions\_3	&	79	&	156	&&		&				&		&&	155.8		&	\textbf{	330		}	&	-	&&	331	\\
DSJC125.1	&	125	&	736	&&		&				&		&&	1752.3		&	\textbf{	2618		}	&	-	&&	2618	\\
DSJC125.5	&	125	&	3891	&&		&				&		&&	2381.7		&	\textbf{	3240		}	&	-	&&	3240	\\
DSJC125.9	&	125	&	6961	&&		&				&		&&	600.6		&	\textbf{	734		}	&	-	&&	734	\\
miles1000	&	128	&	3216	&&		&				&		&&	535		&	\textbf{	535		}	&	331.2	&&	700	\\
miles1500	&	128	&	5198	&&		&				&		&&	218		&	\textbf{	218		}	&	1.65	&&	308	\\
mug100\_1	&	100	&	166	&&		&				&		&&	64		&	\textbf{	64		}	&	0.3	&&	91	\\
mug100\_25	&	100	&	166	&&		&				&		&&	64		&	\textbf{	64		}	&	0.51	&&	93	\\
mug88\_1	&	88	&	146	&&		&				&		&&	56		&	\textbf{	56		}	&	0.22	&&	82	\\
mug88\_25	&	88	&	146	&&		&				&		&&	56		&	\textbf{	56		}	&	0.49	&&	84	\\
myciel6	&	95	&	755	&&		&				&		&&	741.3		&	\textbf{	753		}	&	-	&&	753	\\
r125.1	&	125	&	209	&&		&				&		&&	11		&	\textbf{	11		}	&	0.17	&&	15	\\
r125.1c	&	125	&	7501	&&		&				&		&&	207		&	\textbf{	207		}	&	26.83	&&	207	\\
r125.5	&	125	&	3838	&&		&				&		&&	895.4		&	\textbf{	1231		}	&	-	&&	1231	
\end{tabular}
}
\end{center}
\end{footnotesize}
\end{table}


We conclude this section by comparing our results with those presented in 
\cite{BerRaghu15}.  With the exception of some queen instances, \texttt{BC} always provides better solutions and objective bounds than \texttt{BEN}. In the exceptional cases, the bounds provided by \texttt{BEN} were slightly better. Note also that \texttt{BEN} does not provide any feasible solution until the algorithm terminates. 

\subsection{Cuts Found}
\label{sec:CutsFounds}

This section provides an analysis of the types of cuts found by \texttt{BC} during the solution process across all experiments. Figure~\ref{fig:distributionCuts}~(a) shows an area plot depicting the distribution of the number of inequalities of each type that was identified and added to the model in \texttt{BC}. We present only the 88 instances for which at least 10,000 cuts were added, where all graph classes were considered, and the instances are ordered by total number of cuts found.
This plot readily shows that most of the cuts added were of type (\ref{inq:2}) and (\ref{inq:4}).

Figure~\ref{fig:distributionCuts}~(b) shows an area plot depicting a similar comparison, but between those cuts added at integer nodes and those added by the threshold separation procedure from \S \ref{sec:solutionMethod}.  For the majority of instances, the cuts are predominantely found through threshold cuts.  The far right portion of the plot corresponds to relatively large instances, hence only a few branching nodes were explored.

\begin{figure}[h!]
\centering
\begin{subfigure}{0.4\textwidth}
\centering
\includegraphics[height=30ex]{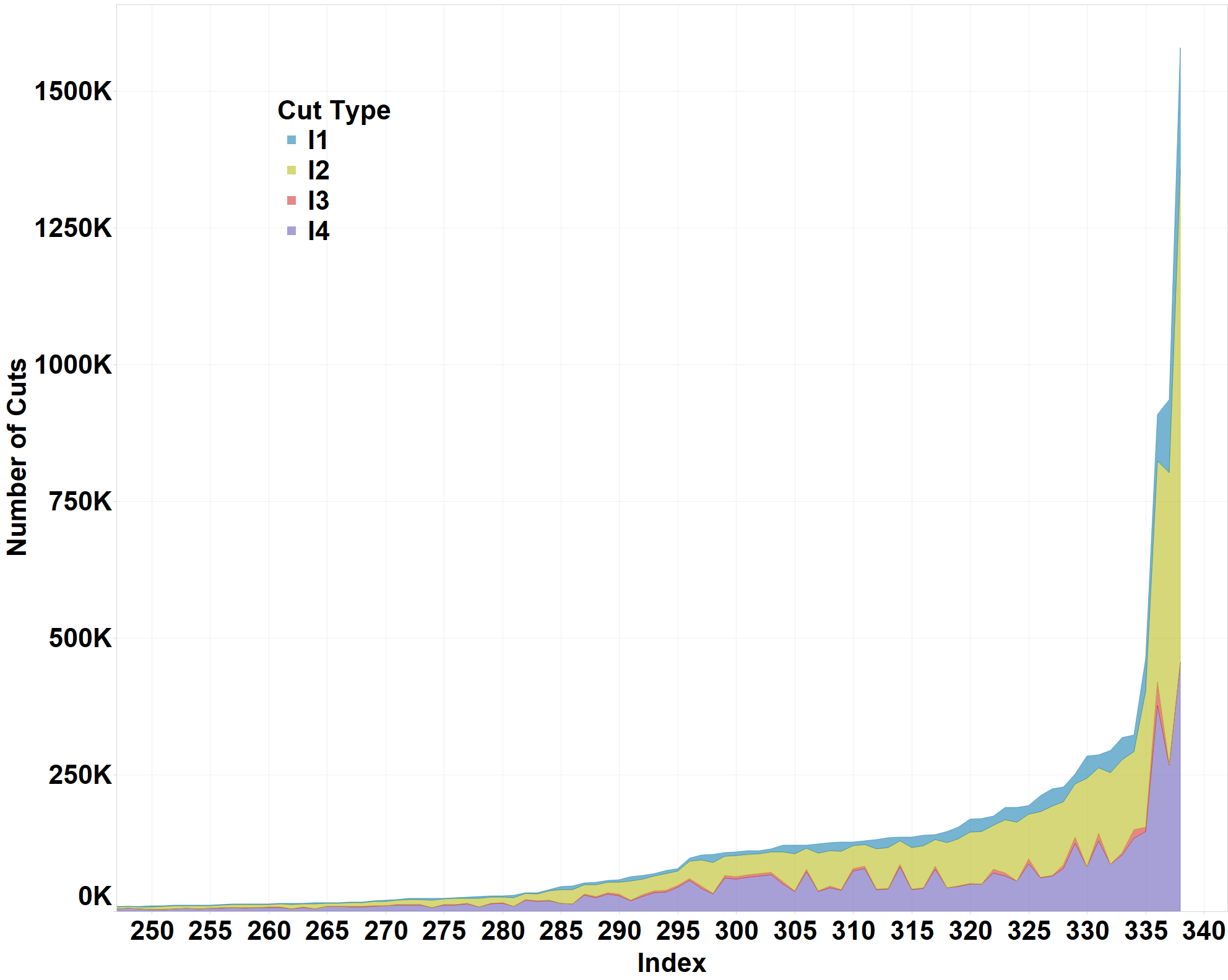}
\caption{Comparing counts of inequalities}
\end{subfigure}
\hspace{5ex}
\begin{subfigure}{0.4\textwidth}
\centering
\includegraphics[height=30ex]{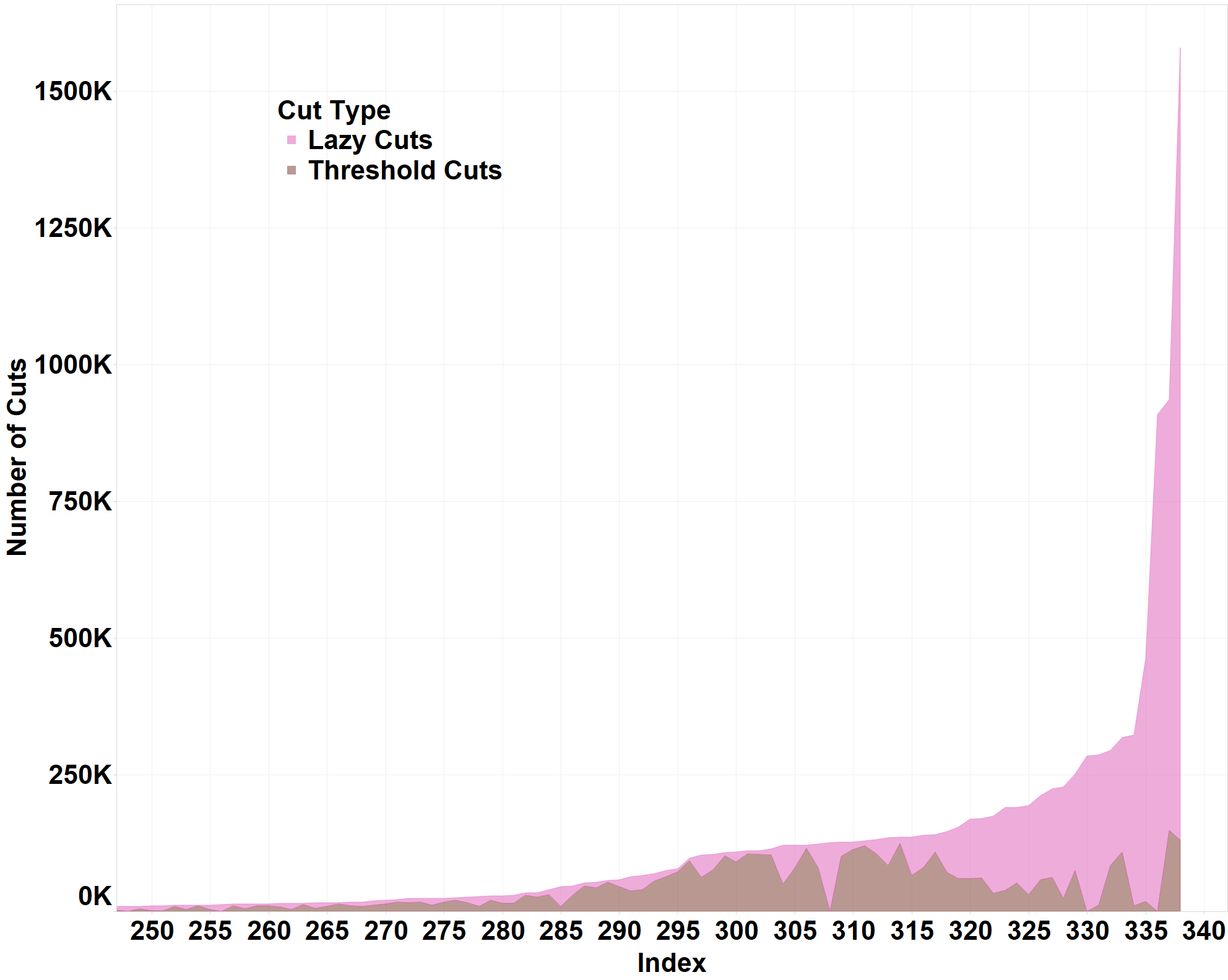}
\caption{Comparing mode of cut identification}
\end{subfigure}
\caption{Area plots displaying types of cuts found and added by \texttt{BC} across all instances.  } 
\label{fig:distributionCuts}
\end{figure}

\section{Conclusion}

In this paper we described a new mathematical programming formulation for the MCCP
and investigated some key properties of its polytope. The constraints employed in  our model correspond to lifted inequalities of induced cycle graphs, and our theoretical results show that this lifting procedure can be generalized to derive other facets of the MCCP polytope of cycle graphs. Finally, we proposed a hybrid solution technique that considers both a lazy-constraint generation and a heuristic separation method based on a threshold rounding, and also presented a simple primal heuristic for the problem. A numerical study indicates that our approach substantially outperforms existing methods, often by orders of magnitude, and, in particular, solves many benchmark graphs to optimality for the first time.



%
%
%


\bibliographystyle{informs2014} 

\begin{thebibliography}{35}
\providecommand{\natexlab}[1]{#1}
\providecommand{\url}[1]{\texttt{#1}}
\providecommand{\urlprefix}{URL }

\bibitem[{Beeri et~al.(1983)Beeri, Fagin, Maier, \protect\BIBand{}
  Yannakakis}]{BeeFagMaiYan83}
Beeri C, Fagin R, Maier D, Yannakakis M (1983) On the desirability of acyclic
  database schemes. \emph{J. ACM} 30(3):479--513, ISSN 0004-5411,
  \urlprefix\url{http://dx.doi.org/10.1145/2402.322389}.

\bibitem[{Bergman \protect\BIBand{} Cire(2016)}]{BerCir2016}
Bergman D, Cire A (2016) Decompositions based on decision diagrams.
  \emph{CPAIOR 2016, to appear} .

\bibitem[{Bergman \protect\BIBand{} Raghunathan(2015)}]{BerRaghu15}
Bergman D, Raghunathan AU (2015) A benders approach to the minimum chordal
  completion problem. \emph{Principles and Practice of Constraint Programming
  – CP 2015}, volume 6308 of \emph{Lecture Notes in Computer Science}, 47--64
  (Springer Berlin Heidelberg).

\bibitem[{Berry et~al.(2006)Berry, Bordat, Heggernes, Simonet,
  \protect\BIBand{} Villanger}]{BerBorHegSimVil2006}
Berry A, Bordat JP, Heggernes P, Simonet G, Villanger Y (2006) A wide-range
  algorithm for minimal triangulation from an arbitrary ordering. \emph{Journal
  of Algorithms} 58(1):33 -- 66, ISSN 0196-6774,
  \urlprefix\url{http://dx.doi.org/http://dx.doi.org/10.1016/j.jalgor.2004.07.001}.

\bibitem[{Berry et~al.(2003)Berry, Heggernes, \protect\BIBand{}
  Simonet}]{BerHegSim2003}
Berry A, Heggernes P, Simonet G (2003) The minimum degree heuristic and the
  minimal triangulation process. Bodlaender H, ed., \emph{Graph-Theoretic
  Concepts in Computer Science}, volume 2880 of \emph{Lecture Notes in Computer
  Science}, 58--70 (Springer Berlin Heidelberg), ISBN 978-3-540-20452-7,
  \urlprefix\url{http://dx.doi.org/10.1007/978-3-540-39890-5_6}.

\bibitem[{Bodlaender et~al.(1998)Bodlaender, Kloks, Kratsch, \protect\BIBand{}
  Mueller}]{BodKloKraHue1998}
Bodlaender HL, Kloks T, Kratsch D, Mueller H (1998) Treewidth and minimum
  fill-in on d-trapezoid graphs.

\bibitem[{Broersma et~al.(1997)Broersma, Dahlhaus, \protect\BIBand{}
  Kloks}]{Moh1997}
Broersma H, Dahlhaus E, Kloks T (1997) Algorithms for the treewidth and minimum
  fill-in of {HHD}-free graphs. Möhring R, ed., \emph{Graph-Theoretic Concepts
  in Computer Science}, volume 1335 of \emph{Lecture Notes in Computer
  Science}, 109--117 (Springer Berlin Heidelberg), ISBN 978-3-540-63757-8,
  \urlprefix\url{http://dx.doi.org/10.1007/BFb0024492}.

\bibitem[{Chang(1996)}]{Cha1996}
Chang MS (1996) Algorithms for maximum matching and minimum fill-in on chordal
  bipartite graphs. Asano T, Igarashi Y, Nagamochi H, Miyano S, Suri S, eds.,
  \emph{Algorithms and Computation}, volume 1178 of \emph{Lecture Notes in
  Computer Science}, 146--155 (Springer Berlin Heidelberg), ISBN
  978-3-540-62048-8, \urlprefix\url{http://dx.doi.org/10.1007/BFb0009490}.

\bibitem[{Chung \protect\BIBand{} Mumford(1994)}]{ChuMum1994}
Chung F, Mumford D (1994) Chordal completions of planar graphs. \emph{Journal
  of Combinatorial Theory, Series B} 62(1):96 -- 106, ISSN 0095-8956,
  \urlprefix\url{http://dx.doi.org/http://dx.doi.org/10.1006/jctb.1994.1056}.

\bibitem[{Codato \protect\BIBand{} Fischetti(2006)}]{Codato2006}
Codato G, Fischetti M (2006) Combinatorial benders' cuts for mixed-integer
  linear programming. \emph{Operations Research} 54(4):756--766,
  \urlprefix\url{http://dx.doi.org/10.1287/opre.1060.0286}.

\bibitem[{Cormen et~al.(2009)Cormen, Leiserson, Rivest, \protect\BIBand{}
  Stein}]{Cormen2009}
Cormen TH, Leiserson CE, Rivest RL, Stein C (2009) \emph{Introduction to
  Algorithms, Third Edition} (The MIT Press), 3rd edition, ISBN 0262033844,
  9780262033848.

\bibitem[{Feremans et~al.(2002)Feremans, Oswald, \protect\BIBand{}
  Reinelt}]{Feremans02}
Feremans C, Oswald M, Reinelt G (2002) A y-formulation for the treewidth.
  Technical report, Heidelberg University.

\bibitem[{Fomin et~al.(2013)Fomin, Philip, \protect\BIBand{}
  Villanger}]{Fomin2013}
Fomin FV, Philip G, Villanger Y (2013) Minimum fill-in of sparse graphs:
  Kernelization and approximation. \emph{Algorithmica} 71(1):1--20, ISSN
  1432-0541, \urlprefix\url{http://dx.doi.org/10.1007/s00453-013-9776-1}.

\bibitem[{Fomin \protect\BIBand{} Villanger(2012)}]{FomVil2012}
Fomin FV, Villanger Y (2012) Subexponential parameterized algorithm for minimum
  fill-in. \emph{Proceedings of the Twenty-third Annual ACM-SIAM Symposium on
  Discrete Algorithms}, 1737--1746, SODA '12 (SIAM),
  \urlprefix\url{http://dl.acm.org/citation.cfm?id=2095116.2095254}.

\bibitem[{Fulkerson \protect\BIBand{} Gross(1965)}]{Fulkerson1965}
Fulkerson DR, Gross OA (1965) Incidence matrices and interval graphs.
  \emph{Pacific J. Math.} 15(3):835--855,
  \urlprefix\url{http://projecteuclid.org/euclid.pjm/1102995572}.

\bibitem[{Garey \protect\BIBand{} Johnson(1979)}]{GarJoh1979}
Garey MR, Johnson DS (1979) \emph{Computers and Intractability: A Guide to the
  Theory of NP-Completeness} (New York, NY, USA: W. H. Freeman \& Co.), ISBN
  0716710447.

\bibitem[{George \protect\BIBand{} Liu(1989)}]{GeoLiu1989}
George A, Liu WH (1989) The evolution of the minimum degree ordering algorithm.
  \emph{SIAM Rev.} 31(1):1--19, ISSN 0036-1445,
  \urlprefix\url{http://dx.doi.org/10.1137/1031001}.

\bibitem[{Grone et~al.(1984)Grone, Johnson, Sá, \protect\BIBand{}
  Wolkowicz}]{GroJohSaWol1984}
Grone R, Johnson CR, Sá EM, Wolkowicz H (1984) Positive definite completions
  of partial hermitian matrices. \emph{Linear Algebra and its Applications}
  58(0):109 -- 124, ISSN 0024-3795,
  \urlprefix\url{http://dx.doi.org/http://dx.doi.org/10.1016/0024-3795(84)90207-6}.

\bibitem[{Heggernes(2006)}]{Heggernes2006}
Heggernes P (2006) Minimal triangulations of graphs: A survey. \emph{Discrete
  Mathematics} 306(3):297 -- 317, ISSN 0012-365X,
  \urlprefix\url{http://dx.doi.org/http://dx.doi.org/10.1016/j.disc.2005.12.003},
  minimal Separation and Minimal Triangulation.

\bibitem[{{IBM ILOG}(2016)}]{CPLEXRef}
{IBM ILOG} (2016) Cplex optimization studio 12.6.3 user manual.

\bibitem[{Judd et~al.(2011)Judd, Kearns, \protect\BIBand{}
  Vorobeychik}]{Judd2011}
Judd S, Kearns M, Vorobeychik Y (2011) Behavioral conflict and fairness in
  social networks. Chen N, Elkind E, Koutsoupias E, eds., \emph{Internet and
  Network Economics}, volume 7090 of \emph{Lecture Notes in Computer Science},
  242--253 (Springer Berlin Heidelberg), ISBN 978-3-642-25509-0,
  \urlprefix\url{http://dx.doi.org/10.1007/978-3-642-25510-6_21}.

\bibitem[{Kaplan et~al.(1999)Kaplan, Shamir, \protect\BIBand{}
  Tarjan}]{Kaplan1999}
Kaplan H, Shamir R, Tarjan RE (1999) Tractability of parameterized completion
  problems on chordal, strongly chordal, and proper interval graphs. \emph{SIAM
  Journal on Computing} 28(5):1906--1922,
  \urlprefix\url{http://dx.doi.org/10.1137/S0097539796303044}.

\bibitem[{Kim et~al.(2011)Kim, Kojima, Mevissen, \protect\BIBand{}
  Yamashita}]{KimKojMevYam2011}
Kim S, Kojima M, Mevissen M, Yamashita M (2011) Exploiting sparsity in linear
  and nonlinear matrix inequalities via positive semidefinite matrix
  completion. \emph{Mathematical Programming} 129(1):33--68, ISSN 0025-5610,
  \urlprefix\url{http://dx.doi.org/10.1007/s10107-010-0402-6}.

\bibitem[{Kloks et~al.(1998)Kloks, Kratsch, \protect\BIBand{}
  Wong}]{KloKraWon1998}
Kloks T, Kratsch D, Wong C (1998) Minimum fill-in on circle and circular-arc
  graphs. \emph{Journal of Algorithms} 28(2):272 -- 289, ISSN 0196-6774,
  \urlprefix\url{http://dx.doi.org/http://dx.doi.org/10.1006/jagm.1998.0936}.

\bibitem[{Lauritzen \protect\BIBand{} Spiegelhalter(1990)}]{Lauritzen1990}
Lauritzen SL, Spiegelhalter DJ (1990) Local computations with probabilities on
  graphical structures and their application to expert systems. Shafer G, Pearl
  J, eds., \emph{Readings in Uncertain Reasoning}, 415--448 (San Francisco, CA,
  USA: Morgan Kaufmann Publishers Inc.), ISBN 1-55860-125-2,
  \urlprefix\url{http://dl.acm.org/citation.cfm?id=84628.85343}.

\bibitem[{Mezzini \protect\BIBand{} Moscarini(2010)}]{MezMos2010}
Mezzini M, Moscarini M (2010) Simple algorithms for minimal triangulation of a
  graph and backward selection of a decomposable {M}arkov network.
  \emph{Theoretical Computer Science} 411(7–9):958 -- 966, ISSN 0304-3975,
  \urlprefix\url{http://dx.doi.org/http://dx.doi.org/10.1016/j.tcs.2009.10.004}.

\bibitem[{Nakata et~al.(2003)Nakata, Fujisawa, Fukuda, Kojima,
  \protect\BIBand{} Murota}]{NakFujFukKojMur2003}
Nakata K, Fujisawa K, Fukuda M, Kojima M, Murota K (2003) Exploiting sparsity
  in semidefinite programming via matrix completion {II}: implementation and
  numerical results. \emph{Mathematical Programming} 95(2):303--327, ISSN
  0025-5610, \urlprefix\url{http://dx.doi.org/10.1007/s10107-002-0351-9}.

\bibitem[{Rollon \protect\BIBand{} Larrosa(2011)}]{Rollon2011}
Rollon E, Larrosa J (2011) \emph{Principles and Practice of Constraint
  Programming -- CP 2011: 17th International Conference, CP 2011, Perugia,
  Italy, September 12-16, 2011. Proceedings}, chapter On Mini-Buckets and the
  Min-fill Elimination Ordering, 759--773 (Berlin, Heidelberg: Springer Berlin
  Heidelberg), ISBN 978-3-642-23786-7,
  \urlprefix\url{http://dx.doi.org/10.1007/978-3-642-23786-7_57}.

\bibitem[{Rose \protect\BIBand{} Tarjan(1978)}]{Rose1978}
Rose DJ, Tarjan RE (1978) Algorithmic aspects of vertex elimination on directed
  graphs. \emph{SIAM Journal on Applied Mathematics} 34(1):176--197,
  \urlprefix\url{http://dx.doi.org/10.1137/0134014}.

\bibitem[{Rose et~al.(1976)Rose, Tarjan, \protect\BIBand{} Lueker}]{Rose1976}
Rose DJ, Tarjan RE, Lueker GS (1976) Algorithmic aspects of vertex elimination
  on graphs. \emph{SIAM Journal on Computing} 5(2):266--283,
  \urlprefix\url{http://dx.doi.org/10.1137/0205021}.

\bibitem[{Rostami et~al.(2015)Rostami, Malucelli, Frey, \protect\BIBand{}
  Buchheim}]{rostami2015}
Rostami B, Malucelli F, Frey D, Buchheim C (2015) On the quadratic shortest
  path problem. \emph{International Symposium on Experimental Algorithms},
  379--390 (Springer).

\bibitem[{Tarjan \protect\BIBand{} Yannakakis(1984)}]{TarYan84}
Tarjan RE, Yannakakis M (1984) Simple linear-time algorithms to test chordality
  of graphs, test acyclicity of hypergraphs, and selectively reduce acyclic
  hypergraphs. \emph{SIAM J. Comput.} 13(3):566--579, ISSN 0097-5397,
  \urlprefix\url{http://dx.doi.org/10.1137/0213035}.

\bibitem[{Vandenberghe \protect\BIBand{} Andersen(2015)}]{Vandenberghe2015}
Vandenberghe L, Andersen MS (2015) Chordal graphs and semidefinite
  optimization. \emph{Foundations and Trends in Optimization} 1(4):241--433,
  ISSN 2167-3888, \urlprefix\url{http://dx.doi.org/10.1561/2400000006}.

\bibitem[{Yannakakis(1981)}]{Yan1981}
Yannakakis M (1981) Computing the minimum fill-in is {NP}-{C}omplete.
  \emph{SIAM Journal on Algebraic Discrete Methods} 2(1):77--79,
  \urlprefix\url{http://dx.doi.org/10.1137/0602010}.

\bibitem[{Y{\"u}ceo{\u{g}}lu(2015)}]{Yuceoglu2015}
Y{\"u}ceo{\u{g}}lu B (2015) \emph{Branch-and-cut algorithms for graph
  problems}. Ph.D. thesis, Maastricht University,
  \urlprefix\url{http://digitalarchive.maastrichtuniversity.nl/fedora/get/guid:bde57abf-9652-45bc-8590-667e7e085074/ASSET1}.

\end{thebibliography}


\ECSwitch


\ECHead{Online Supplement - Proofs of Statements}

\section{Additional Proofs for Section~\ref{sec:cycleGraphFacets}}
\label{sec:facetProofsAdditionalInequalities}

\proof{Facet-defining proof of Proposition \ref{prop:inq1}.}
Let $F^I = \left\{ x \in X(G) : \sum_{ f \: \in \: \interior(C)} x_f = |C|-3\right\}$ and $\mu x \geq \mu_0$ be a valid inequality for $\mathrm{conv}(X(G))$ which is satisfied at equality by each $x \in F^I$. It suffices to show that there exists some $\lambda$ for which $\mu_f = \lambda$ and $\mu_0 = (|C|-3)\lambda$.

Let $x' \in \{0,1\}^{m^c}$ be defined by
\[
x'_f = 
\left\{
\begin{array}{ll}
1, &  \quad f = \{v_0,v_j\}, j = 2, \ldots, k-2; \\
0, &  \quad \mbox{otherwise}.
\end{array}
\right.
\]

\begin{claim}
$x' \in X(G)$, i.e., $G(x')$ is chordal, and $x' \in F^I$.
\end{claim}
\proof{Proof.}
For every $j \in [2,k-1]$, let $\bar{V}_j = \{v_0,v_1,\ldots,v_{j}\}$.
By construction, set~$N_{G \left[\bar{V}_{j} \right]}(v_{j}) = \{v_0,v_{j-1}\}$ induces a clique in $G[\bar{V}_{j}]$.  Therefore, $v_0,v_1,\ldots,\allowbreak v_{k-1}$ is a perfect elimination ordering of~$V(G(x'))$, thereby proving that $G(x')$ is chordal.  Additionally, since exactly $|C|-3$ edges in $\interior(C)$ are in $G(x'), x' \in F^I$.$\Halmos$

Consider now the solutions~$\tilde{x}^i \in \{0,1\}^{m^c}$ for $i = 3, \ldots, k-1$, defined by
\[
\tilde{x}^i_f = 
\left\{
\begin{array}{ll}
1, &  \quad f = \{v_1,v_j\}, j = 3, \ldots, i; \\
1, &  \quad f = \{v_0,v_j\}, j = i,i+1, \ldots, k-2; \\
0, &  \quad \mbox{otherwise.}
\end{array}
\right.
\]
Figure~\ref{fig:inq2proof}~(a) provides a depiction of $G(\tilde{x}^i)$.
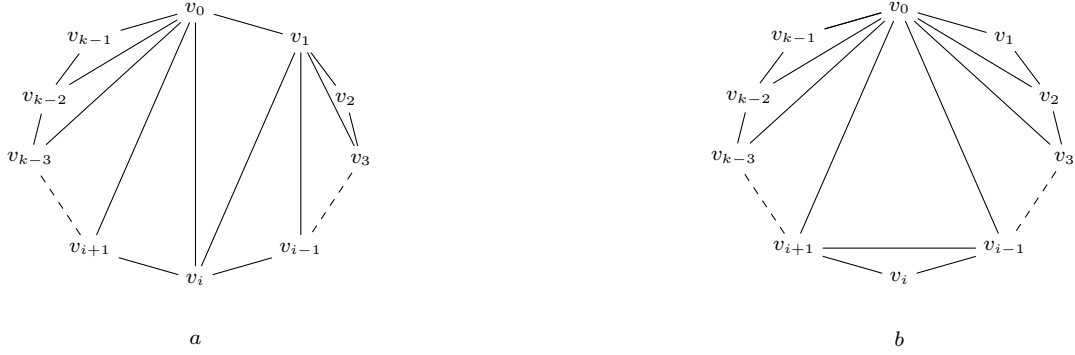
\begin{figure}
\begin{minipage}{0.4\textwidth}
\begin{tikzpicture}[scale=0.4][font=\sffamily,\scriptsize]

\node (0) at (0,0) {$v_0$};
\node (1) at (3.5,-1) {$v_1$};
\node (2) at (5,-3) {$v_2$};
\node (3) at (5.5,-5) {$v_3$};
\node (i-1) at (3.5,-8) {$v_{i-1}$};
\node (i)  at (0,-9) {$v_{i}$};
\node (i+1) at (-3.5,-8) {$v_{i+1}$};
\node (k-1) at (-3.5,-1) {$v_{k-1}$};
\node (k-2) at (-5,-3) {$v_{k-2}$};
\node (k-3) at (-5.5,-5) {$v_{k-3}$};

\node (a)  at (0,-11) {$a$};


\path[-](0) edge (1);
\path[-](1) edge (2);
\path[-](2) edge (3);
\path[dashed](3) edge (i-1);
\path[-](i-1) edge (i);
\path[-](i+1) edge (i);
\path[dashed](i+1) edge (k-3);
\path[-](0) edge (k-1);
\path[-](k-1) edge (k-2);
\path[-](k-2) edge (k-3);

\path[-](3) edge (1);
\path[-](i-1) edge (1);
\path[-](i) edge (1);

\path[-](i) edge (0);

\path[-](k-2) edge (0);
\path[-](k-3) edge (0);
\path[-](i+1) edge (0);
\end{tikzpicture}
\end{minipage}
\hspace{15ex}
\begin{minipage}{0.4\textwidth}
\begin{tikzpicture}[scale=0.4][font=\sffamily,\scriptsize]

\node (0) at (0,0) {$v_0$};
\node (1) at (3.5,-1) {$v_1$};
\node (2) at (5,-3) {$v_2$};
\node (3) at (5.5,-5) {$v_3$};
\node (i-1) at (3.5,-8) {$v_{i-1}$};
\node (i)  at (0,-9) {$v_{i}$};
\node (i+1) at (-3.5,-8) {$v_{i+1}$};
\node (k-1) at (-3.5,-1) {$v_{k-1}$};
\node (k-2) at (-5,-3) {$v_{k-2}$};
\node (k-3) at (-5.5,-5) {$v_{k-3}$};

\node (b)  at (0,-11) {$b$};


\path[-](0) edge (1);
\path[-](1) edge (2);
\path[-](2) edge (3);
\path[dashed](3) edge (i-1);
\path[-](i-1) edge (i);
\path[-](i+1) edge (i);
\path[dashed](i+1) edge (k-3);
\path[-](0) edge (k-1);
\path[-](k-1) edge (k-2);
\path[-](k-2) edge (k-3);

\path[-](0) edge (2);
\path[-](0) edge (3);
\path[-](0) edge (i-1);
\path[-](0) edge (i+1);
\path[-](0) edge (k-3);
\path[-](0) edge (k-2);
\path[-](0) edge (k-1);

\path[-](i-1) edge (i+1);

\end{tikzpicture}
\end{minipage}

\caption{The graph (a) $G(\tilde{x}^i)$ and the graph (b) $G(\hat{x}^i)$ defined in the proof of Proposition~\ref{prop:inq1}.}
\label{fig:inq2proof}
\end{figure}

\begin{claim}
Each $\tilde{x}^i \in X(G)$, i.e., each $G(\tilde{x}^i)$ is chordal, and $\tilde{x}^i \in F^I$.
\end{claim}
\proof{Proof.}
For every $j \in [2,k-1]$, let $\bar{V}_j$ be the set of vertices belonging to the subsequence of $(v_1,\ldots,v_i,v_0,v_{i+1},\ldots,v_{k-1})$ finishing at element~$v_j$. By construction, we have
\[
N_{G\left[\bar{V}_j\right]}(v_j) = 
\left\{
\begin{array}{ll}
\left\{v_1,v_{i}\right\}, &  \quad j = 0; \\
\left\{v_1\right\}, &  \quad j = 2; \\
\left\{v_1,v_{j-1}\right\}, &  \quad 3 \leq j \leq i; \\
\left\{v_0,v_{j-1}\right\}, &  \quad i+1 \leq j \leq k-1,
\end{array}
\right.
\]
which in each case is a clique.  Therefore, $v_1,\ldots,v_i,v_0,v_{i+1},\ldots,v_{k-1}$  is a perfect elimination ordering of~$V(G(\tilde{x}^i))$, showing thus that~$G(\tilde{x}^i)$ is chordal. Moreover, exactly $|C|-3$ edges of~$\interior(C)$ are in $G(\tilde{x}^i)$, so $\tilde{x}^i \in F^I$.$\Halmos$

Let $\lambda_2 = \mu_{\{v_0,v_2\}}$. Solutions $\tilde{x}^i$ and $x'$ belong to~$F_I$, so $\mu \tilde{x}^i = \mu x' = \mu_0$.  By subtracting equation $\mu \tilde{x}^3 = \mu_0$ from $\mu x' = \mu_0$, we obtain $\mu_{\{v_0,v_2\}} = \mu_{\{v_1,v_3\}}$. Additionally, shift operations on the order of the vertices (to the left or to the right) lead to the same cycle~$C$.
Therefore, $\coef{v_0}{v_2} = \coef{v_j}{v_{(j + 2) \mod k}}$ for any $j \in [k-1]$, implying thus that $\coef{v_0}{v_2} = \lambda_2$ for every $f \in \interior(C)$ containing vertices whose indices in~$C$ differ by 2.

The same operation involving $\tilde{x}^{i-1}$ and $\tilde{x}^i$ for $i = 4, \ldots, k-1$ yields $\coef{v_0}{v_{i-1}} = \coef{v_1}{v_{i}}$.  Again, as the ordering around~$C$ can be arbitrarily shifted to the left and to the right, all edges $f \in \interior(C)$ containing vertices whose indices in $C$ differ by $i-1$ have the same coefficient in $\mu$; let $\lambda_{i-1}$ be this common value.   We conclude thus that $\coef{v_j}{v_{j'}} = \lambda_{j' - j}$ for any $f = \{v_j,v_{j'}\}$ (assuming $j < j'$).


Consider now the solutions~$\hat{x}^i \in \{0,1\}^{m^c}$ for $i = 2, \ldots, k-2$, defined by
\[
\hat{x}^i_f = 
\left\{
\begin{array}{ll}
1, &  \quad f = \{v_0,v_j\}, j = 2, \ldots, i-1,i+1,\ldots,k-2 \\
1, &  \quad f = \{v_{i-1},v_{i+1}\} \\
0, &  \quad \mbox{otherwise}
\end{array}
\right.
\]
Figure~\ref{fig:inq2proof}~(b) provides a depiction of $G(\hat{x^i})$.

\begin{claim}
Each $\hat{x}^i \in X(G)$, i.e., each $G(\hat{x}^i)$ is chordal, and $\hat{x}^i \in F^I$.
\end{claim}
\proof{Proof.}
For every $j \in [2,k-1]$, let $\bar{V}_j$ be the set of vertices belonging to the subsequence of $(v_0,v_1,\ldots,v_{i-1},v_{i+1},\ldots, v_{k-1},v_i)$ finishing at element~$v_j$. By construction, we have
\[
N_{G\left[\bar{V}_j\right]}(v_j) = 
\left\{
\begin{array}{ll}
\left\{v_0\right\}, &  \quad j = 2; \\
\left\{v_0,v_{j-1}\right\}, &  \quad 1 \leq j \leq i-1;\\
\left\{v_{i-1},v_{i+1}\right\}, &  \quad j = i,
\end{array}
\right.
\]
which in each case is a clique. Consequently, $(v_0,v_1,\ldots,v_{i-1},v_{i+1},\ldots, v_{k-1},v_i)$ is a perfect elimination ordering, so $\hat{x}^i \in X(G)$. Finally, as $|C|-3$ edges from $\interior(C)$ are included in $G(\hat{x}^i), \hat{x}^i \in F^I$. $\Halmos$

By subtracting $\mu x' = \mu_0$ from $\mu \hat{x}^i = \mu_0$ for any $i = 2,\ldots,k-2$, we obtain $\coef{v_0}{v_i} = \coef{v_{i-1}}{v_{i+1}}$. Therefore, we have that $\lambda_i = \lambda_2$ for any $i = 2, \ldots, k-2$. If $\lambda = \lambda_2$, $\mu x = \mu_0$ can be rewritten $\sum_{f \in \interior(C)} \lambda x_f = \mu_0$.  Finally, substituting $x'$ in this equation yields $\mu_0 =  \left(|C|-3\right) \lambda$, as desired.
$\Halmos$
\endproof

\proof{Facet-defining proof of Proposition \ref{prop:inq2}.}


Let $I:= ax \geq b$ be the inequality of type (\ref{inq:2}) associated with $i = 1$, $F^I$ be the set of points in $\mathrm{conv}(X(G))$ that satisfy $I$ at equality, and $\mu x \geq \mu_0$ be a valid inequality for $\mathrm{conv}(X(G))$ satisfied at equality for each $x \in F^I$.  Let $\mu_{v_{0},v_2} = \lambda$.

\begin{claim}
$\forall i \in \{3,4,\ldots,k-1 \}, \mu_{\{ v_1,v_i \}} = \lambda.$
\end{claim}
\begin{proof}{Proof.}
For $i \in \{3,4,\ldots,k-1 \}$ let $\check{x}^i$ be the solution
\[
\check{x}^i_{f} = 
\left\{
\begin{array}{ll}
1, &  \quad \{v_i,v_j\}, j \in [k-1] \backslash \{i-1,i,i+1\}, \\
0, &  \quad \mbox{otherwise},
\end{array}
\right.
\]
and $\check{y}^i$ be the solution  such that $\check{y}^i_{\{ v_1,v_i \}} = 0$, $\check{y}^i_{\{ v_0,v_2 \}} = 1$, and $\check{y}^i_f = \check{x}^i_f$ for the remaining edges in~$E(G)^c$; both families of solutions are depicted in Figure~\ref{fig:inq4proof}. We have that both $\check{x}^i$ and $\check{y}^i$ belong to $X(G)$, $i \in \{3,4,\ldots,k-1 \}$, as 
$G(\check{x}^i)$ is isomorphic to $G(x')$ and $\check{y}^i$ is isomorphic to $G(\hat{x}^{i-1})$; note that $G(\check{x}^i)$ and $G(\hat{x}^{i-1})$ were  defined and shown to be associated with chordal completions in the proof of Proposition~\ref{prop:inq1}.  Additionally, note that $a_{f}\check{x}^i_{f} = 1$ only for $f = \{ v_1,v_i \}$ and $a_{f}\check{y}^i_{f} = 1$ only for $f = \{ v_0,v_2 \}$, so both solutions satisfy $I$ at equality and, by definition, 
$\mu \check{x}^i = \mu_0  = \mu \check{y}^i$. Therefore, we must have $\mu_{\{ v_1,v_i \}} = \mu_{\{ v_0,v_2 \}} = \lambda$ for $i \in \{3,4,\ldots,k-1 \}$, as desired. $\Halmos$
\begin{figure}
\begin{minipage}{0.4\textwidth}
\begin{tikzpicture}[scale=0.4][font=\sffamily,\scriptsize]

\node (1) at (0,0) {$v_1$};
\node (2) at (3.5,-1) {$v_2$};
\node (3) at (5,-3) {$v_3$};
\node (i-2) at (5,-6) {$v_{i-2}$};
\node (i-1) at (3.5,-8) {$v_{i-1}$};
\node (i)  at (0,-9) {$v_{i}$};
\node (i+1) at (-3.5,-8) {$v_{i+1}$};
\node (i+2) at (-5,-6) {$v_{i+2}$};
\node (k-1) at (-5,-3) {$v_{k-1}$};
\node (0) at (-3.5,-1) {$v_{0}$};

\node (a)  at (0,-11) {$a$};


\path[-](1) edge (2);
\path[-](2) edge (3);
\path[dashed](3) edge (i-2);
\path[-](i-2) edge (i-1);
\path[-](i-1) edge (i);
\path[-](i) edge (i+1);
\path[-](i+1) edge (i+2);
\path[dashed](i+2) edge (k-1);
\path[-](k-1) edge (0);
\path[-](0) edge (1);

\path[-](2) edge (i);
\path[-](3) edge (i);
\path[-](i-2) edge (i);
\path[-](i-1) edge (i);
\path[-](i+1) edge (i);
\path[-](i+2) edge (i);
\path[-](k-1) edge (i);
\path[-](0) edge (i);

\path[-](1) edge (i);

%
%
%
\end{tikzpicture}
\end{minipage}
\hspace{15ex}
\begin{minipage}{0.4\textwidth}
\begin{tikzpicture}[scale=0.4][font=\sffamily,\scriptsize]
\node (1) at (0,0) {$v_1$};
\node (2) at (3.5,-1) {$v_2$};
\node (3) at (5,-3) {$v_3$};
\node (i-2) at (5,-6) {$v_{i-2}$};
\node (i-1) at (3.5,-8) {$v_{i-1}$};
\node (i)  at (0,-9) {$v_{i}$};
\node (i+1) at (-3.5,-8) {$v_{i+1}$};
\node (i+2) at (-5,-6) {$v_{i+2}$};
\node (k-1) at (-5,-3) {$v_{k-1}$};
\node (0) at (-3.5,-1) {$v_{0}$};

\node (b)  at (0,-11) {$b$};


\path[-](1) edge (2);
\path[-](2) edge (3);
\path[dashed](3) edge (i-2);
\path[-](i-2) edge (i-1);
\path[-](i-1) edge (i);
\path[-](i) edge (i+1);
\path[-](i+1) edge (i+2);
\path[dashed](i+2) edge (k-1);
\path[-](k-1) edge (0);
\path[-](0) edge (1);

\path[-](2) edge (i);
\path[-](3) edge (i);
\path[-](i-2) edge (i);
\path[-](i-1) edge (i);
\path[-](i+1) edge (i);
\path[-](i+2) edge (i);
\path[-](k-1) edge (i);
\path[-](0) edge (i);

\path[-](0) edge (2);
\end{tikzpicture}
\end{minipage}
\caption{The graph (a) $G(\check{x}^i)$ and the graph (b) $G(\check{y}^i)$ defined in the proof of Proposition~\ref{prop:inq2}.}
\label{fig:inq4proof}
\end{figure}
\end{proof}

\begin{claim}
For each $f' \in \interior(C) 
\backslash \left(\{v_0,v_2\} \cup \bigcup\limits_{i=3,\ldots,k-1} \{ v_1,v_i \} \right)$,  $\mu_{f'} = 0$.
\end{claim}
\begin{proof}{Proof.}
First, note that $\interior(C) 
\backslash \left(\{v_0,v_2\} \cup \bigcup\limits_{i=3,\ldots,k-1} \{ v_1,v_i \} \right) \neq \emptyset$ only if $k \geq 5$. Let $f' = \left\{v_j,v_{j'}\right\}$ be such an edge, and assume without loss of generality that $v_j \neq v_0$ (i.e., $v_{j'}$ can be equal to~$v_2$).  Let $z(f')$ be the solution in~$\{0,1\}^{m^c}$ presented in
Figure~\ref{fig:inq4proof2} (part a) defined by
\[
z(f')_{f} = 
\left\{
\begin{array}{ll}
1, & \quad f = \{v_i,v_{j'}\}, k = 0, \text{ } 2 \leq i \leq j'-2, \text{ and } j+1 \leq i \leq k-1, \\
1, & \quad f = \{v_{j+1},v_{j'+1}\},\\
1, & \quad f = \{v_j,v_{i}\},  j'+1 \leq i \leq j-2, \\
0, &  \quad \mbox{otherwise}.
\end{array}
\right.
\]
Solution~$z(f')$ satisfies $I$ at equality, as $a_f z(f')_f = 1$ for $f = \{v_1,v_{j'}\}$ and $a_f z(f')_f  = 0$ for all the other edges in $G(z(f'))$. Moreover, $G(z(f'))$ is isomorphic to the graph presented in Figure~\ref{fig:inq2proof} (part a), so $z(f') \in X(G)$.

Let now~$z'(f')$ be the solution in~$\{0,1\}^{m^c}$ such that $z'(f')_{f'} = 1$ and $z'(f')_{f} = z(f')_{f}$ for the remaining edges; this solution is presented in Figure~\ref{fig:inq4proof2} (part b). The same argument used for~$z(f')$ shows that $z'(f')$ satisfies $I$ at equality, and sequence $(v_j,v_{j'},v_{j+1},v_{j+2},\ldots,v_{k-1},v_{0},v_1,\ldots,v_{j'-1},v_{j'+1},\ldots,v_{j-1})$ is a perfect elimination ordering of~$V(G)$ for $G(z'(f'))$, which shows that $z'(f') \in X(G)$.

Finally, because $\mu z'(f') = \mu_0 = \mu z(f')$, it follows that $\mu_{f'} = 0$, as desired. $\Halmos$

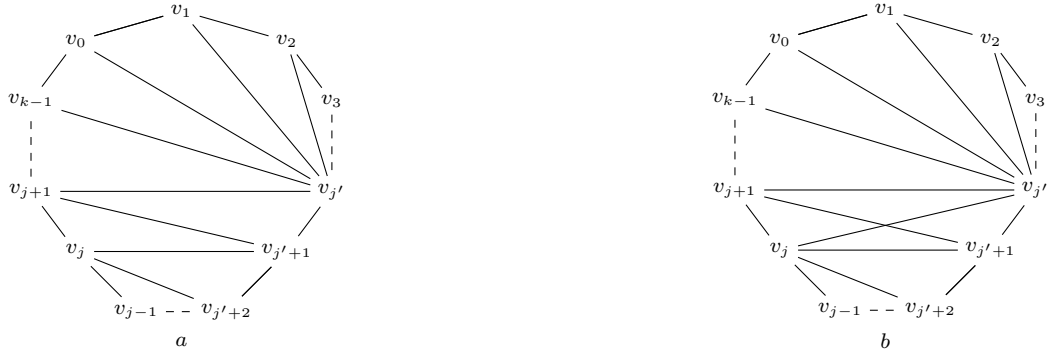
\begin{figure}
\begin{minipage}{0.4\textwidth}
\begin{tikzpicture}[scale=0.4][font=\sffamily,\scriptsize]

\node (1) at (0,0) {$v_1$};
\node (2) at (3.5,-1) {$v_2$};
\node (3) at (5,-3) {$v_3$};
\node (j') at (5,-6) {$v_{j'}$};
\node (j'+1) at (3.5,-8) {$v_{j'+1}$};
\node (j'+2) at (1.5,-10) {$v_{j'+2}$};
\node (j-1) at (-1.5,-10) {$v_{j-1}$};
\node (j+1) at (-5,-6) {$v_{j+1}$};
\node (j) at (-3.5,-8) {$v_{j}$};
\node (k-1) at (-5,-3) {$v_{k-1}$};
\node (0) at (-3.5,-1) {$v_{0}$};

\node (a)  at (0,-11) {$a$};


\path[-](0) edge (1);
\path[-](1) edge (2);
\path[-](2) edge (3);
\path[dashed](3) edge (j');
\path[-](j') edge (j'+1);
\path[-](j'+1) edge (j'+2);
\path[-](j-1) edge (j);
\path[-](j'+1) edge (j'+2);
\path[-](j+1) edge (j');
\path[-](j) edge (j+1);
\path[dashed](j+1) edge (k-1);
\path[-](k-1) edge (0);
\path[-](0) edge (1);
\path[dashed](j-1) edge (j'+2);

\path[-](1) edge (j');
\path[-](2) edge (j');
\path[-](0) edge (j');
\path[-](k-1) edge (j');
\path[-](j) edge (j'+1);
\path[-](j) edge (j'+2);

\path[-](j+1) edge (j'+1);

\end{tikzpicture}
\end{minipage}
\hspace{15ex}
\begin{minipage}{0.4\textwidth}
\begin{tikzpicture}[scale=0.4][font=\sffamily,\scriptsize]

\node (1) at (0,0) {$v_1$};
\node (2) at (3.5,-1) {$v_2$};
\node (3) at (5,-3) {$v_3$};
\node (j') at (5,-6) {$v_{j'}$};
\node (j'+1) at (3.5,-8) {$v_{j'+1}$};
\node (j'+2) at (1.5,-10) {$v_{j'+2}$};
\node (j-1) at (-1.5,-10) {$v_{j-1}$};
\node (j+1) at (-5,-6) {$v_{j+1}$};
\node (j) at (-3.5,-8) {$v_{j}$};
\node (k-1) at (-5,-3) {$v_{k-1}$};
\node (0) at (-3.5,-1) {$v_{0}$};

\node (b)  at (0,-11) {$b$};


\path[-](0) edge (1);
\path[-](1) edge (2);
\path[-](2) edge (3);
\path[dashed](3) edge (j');
\path[-](j') edge (j'+1);
\path[-](j'+1) edge (j'+2);
\path[-](j-1) edge (j);
\path[-](j'+1) edge (j'+2);
\path[-](j+1) edge (j');
\path[-](j) edge (j+1);
\path[dashed](j+1) edge (k-1);
\path[-](k-1) edge (0);
\path[-](0) edge (1);

\path[dashed](j-1) edge (j'+2);

\path[-](1) edge (j');
\path[-](2) edge (j');
\path[-](0) edge (j');
\path[-](k-1) edge (j');
\path[-](j) edge (j'+1);
\path[-](j) edge (j'+2);

\path[-](j) edge (j');
\path[-](j+1) edge (j'+1);

\end{tikzpicture}
\end{minipage}
\caption{The graph (a) $G(z(f'))$ and the graph (b) $G(z'(f'))$ defined in the proof of Proposition~\ref{prop:inq2}.}
\label{fig:inq4proof2}
\end{figure}

\end{proof}

Direct inspection on any solution $\check{x}$ (e.g., $\check{x}^3$) allows us to see  that
$\mu_0 = \lambda$. Therefore, there exists a $\lambda$ such that $\forall f \in E^c, \mu_f = \lambda a_f$ and $\mu_0 = b \lambda$, completing the proof that $I$ is facet-defining. 
$\Halmos$
\endproof

\bigskip

\proof{Facet-defining proof of Proposition \ref{prop:inq3}.}
Let $I:= ax \geq b$ be any inequality (\ref{inq:3}). Moreover, let $F^I$ be the set of points in $\mathrm{conv}(X(G))$ that satisfy $I$ at equality and $\mu x \geq \mu_0$ be a valid inequality for $\mathrm{conv}(X(G))$ satisfied at equality for each $x \in F^I$.  Let $\lambda = \coef{v_0}{v_2}$.

\begin{claim}
$\exists \lambda \neq 0 \mbox{ {such that} } \forall f \in \interior(C) \mbox{ with } d_C(f) = 2,  \mu_f = \lambda$.
\end{claim}
\begin{proof}{Proof.}
 Consider solutions $x'$ and $\tilde{x}^3$ presented in the proof of Proposition~\ref{prop:inq1}; graph~$G(x')$ is 
isormorphic to the one shown  in Figure~\ref{fig:inq4proof} (part a) and has~$v_0$ as the neighbour of all vertices in~$V(G)$, whereas
graph~$G(\tilde{x}^3)$ is shown in Figure~\ref{fig:inq2proof} (part a).
By construction, both solutions are in $F^I$.  Subtracting $\mu x' = \mu_0$ from $\mu \tilde{x}^3 = \mu_0$ and cancelling like terms yields $\coef{v_0}{v_2} = \coef{v_1}{v_3}$.  By applying sequentially this procedure starting from any fill edge of $C$ with $d_C(v_i,v_j) = 2$, we obtain the desired result. 
$\Halmos$
\end{proof}


\begin{claim}
$\forall f \in \interior(C) \mbox{ s.t. } d_C(f) \geq 3, \mu_f = 0$.
\end{claim}
\begin{proof}{Proof.}
Let $\tilde{y}^i$ be the set of solutions given by $\tilde{y}^i = \tilde{x}^i + e^{\{v_1,v_{i+1}\} }$, $3 \leq i \leq k-3$, with $\tilde{x}$ being again the solutions defined in the proof of Proposition~\ref{prop:inq1}.
Each solution~$\tilde{y}^i$ satisfies~$I$ at equality. Moreover, 
$(v_1,v_2,\ldots,v_i,v_0,v_{i+1},v_{i+2},\ldots,v_{k-1})$ is a perfect elimination ordering for each~$G(\tilde{y}^i)$, showing thus that each~$\tilde{y}^i$ is a valid solution.  Therefore, $\mu \tilde{y}^i = \mu_0 = \mu \tilde{x}^i = \mu_0$, and as $\tilde{y}^i$ and $\tilde{x}^i$ only differ on the coordinate corresponding to fill edge $\{v_1,v_{i+1}\}$, we must have $\mu_{\{v_1,v_{i+1}\}} = 0$ for any edge index $i, 3 \leq i \leq k-2$.  This implies, due to cyclic symmetry, that $\mu_f = 0$ for any edge $f  \in \interior(C)$ such that $d_C(f) \geq 3$.
$\Halmos$
\end{proof}

Finally, we have $\mu_0 = \mu x' = \mu_{\{v_0,v_2\}} + \mu_{\{v_0,v_2\}} = 2\lambda $. Therefore, there exists a $\lambda$ such that $\mu_0 = b \lambda$ and $\mu_f = \lambda a_f$ for every $f$  in $E^c$, which shows that~$I$ is facet-defining. $\Halmos$



\endproof

\bigskip

\proof{Facet-defining proof of Proposition \ref{prop:inq4}.}

Let $I:= ax \geq b$ be any inequality (\ref{inq:4}).  Without loss of generality, let $j=0$ and $i$ be any value in $[2,k-3]$.
Let $F^I$ be the set of points in $\mathrm{conv}(X(G))$ that satisfy $I$ at equality, and let $\mu x \geq \mu_0$ be a valid inequality for $\mathrm{conv}(X(G))$ satisfied at equality for each $x \in F^I$.

\begin{claim}
$\mu_{\{v_0,v_i\}} =  \mu_{\{v_{i-1},v_{i+1}\}}  = 0$.
\end{claim}
\begin{proof}{Proof.}
Consider solutions $x'$ and $\hat{x}^i$ presented in the proof of Proposition~\ref{prop:inq1}. Direct inspection allows us to see that both belong to~$F^I$ and differ only on coordinates $\sigma(\{v_0,v_i\})$ and $\sigma(\{v_{i-1},v_{i+1}\})$, so that $\mu x' = \mu_0 = \mu \hat{x}^i$ 
implies that $\mu_{\{v_0,v_i\}} =  \mu_{\{v_{i-1},v_{i+1}\}}$.

Additionally, the solution $x' + e^{\left\{ v_{i-1},v_{i+1} \right\}}$ also belongs to~$F^I$: it satisfies $I$ at equality and the sequence $(v_0,v_1,\ldots,v_{k-1})$ is a perfect elimination order of $V(G)$.  Therefore, as $\mu \left(x' + e^{\left\{ v_{i-1},v_{i+1} \right\}}\right) = \mu x' + \mu_{\left\{ v_{i-1},v_{i+1} \right\}} = \mu_0$, it follows that $\mu_{\left\{ v_{i-1},v_{i+1} \right\}} = 0$.
 $\Halmos$
\end{proof}

Let $C' = (v_0, v_1, \ldots, v_{i-1}, v_{i+1}, v_{i+2}, \ldots, v_{k-1})$ and let $\lambda = \mu_{v_{1},v_{k-1}}$.

\begin{claim}
$\exists \ \lambda \mbox{ such that } \forall f \in \interior(C'), \mu_f = \lambda$.
\end{claim}
\begin{proof}{Proof.}
Let~$x$ be any feasible solution of $X(C' + \{ \{ v_{i-1},v_{i+1} \} \})$  satisfying inequality~(\ref{ineq:ci}) at equality. Let~$y$ be a solution of~$X(G)$ defined as follows:
\[
y_{f} = 
\left\{
\begin{array}{ll}
{x}_f, &  \quad f \in \mathrm{int}(C')  \\
1			, & \quad f = \{v_{i-1}, v_{i+1}\} \\
0, &  \quad \mbox{otherwise}.
\end{array}
\right.
\]
Solution~$y$ belongs to~$X(G)$  because any perfect elimination order of $V(C')$ can be extended into a perfect elimination order for $V(C)$ by putting $v_i$ in the end of the sequence (note that the only neighbors of~$v_i$ are $v_{i-1}$ and $v_{i+1}$, which are necessarily connected). Moreover, by construction, $\sum_{f \in \mathrm{int}(C')}y_f = |C'| - 3 = |C| - 4$, so~$y \in F^I$. Finally, 
as $\sum_{f \in \interior(C) \setminus \{\{v_{i-1},v_{i+1},\{v_j,v_i\} \}\}  }y_f  =  \sum_{f \in \interior(C')}y_f$ for all $j \in [k-1] \setminus \{i\}$, it follows from the arguments used in the proof of Proposition~\ref{prop:inq1} that $\mu_f = \lambda$ for each $f \in \interior(C')$. $\Halmos$
\end{proof}

\begin{claim}
$\forall f = \{v_i,v_{\ell}\}, \ell = 1, 2, \ldots, i-2, i+2, i+3, \ldots, k-2, \mu_{f} = \lambda$.
\end{claim}
\begin{proof}{Proof.}
Fix $\ell' \in \{ 1, 2, \ldots, i-2, i+2, i+3, \ldots, k-2 \}$.  Let $\bar{x}$ be defined by
\[
\bar{x}_{f} = 
\left\{
\begin{array}{ll}
1			, & \quad v_i \in f \\
0, &  \quad \mbox{otherwise}.
\end{array}
\right.
\]
$G(\bar{x})$ is isomorphic to the solution $x'$ defined in the proof of Proposition~\ref{prop:inq1}, so it follows that~$\bar{x}$ is feasible. Moreover, $\bar{x}_f = 1$ for $|C|-4$ edges  in  $ \interior(C) \ \backslash \ \left\{ \left\{ v_{i-1},v_{i+1} \right\} , \left\{ v_0, v_i \right\}  \right\}$,  so we have that $\bar{x} \in F^I$.

Let $\bar{x}^{\ell'}$ be the solution of $X(G)$ such that  
$\bar{x}^{\ell'}_{\{v_{\ell'-1},v_{\ell'+1}\}} = 1$, $\bar{x}^{\ell'}_{\{v_{i},v_{\ell'}\}} = 0$, and $\bar{x}^{\ell'}_f = \bar{x}_f$ for
the remaining edges.
The graph $G(\bar{x}^{\ell'})$ is  isomorphic to 
one of the graphs $G(\hat{x}^{i'})$ defined in the proof of Proposition~\ref{prop:inq1}, and therefore $\bar{x}^{\ell'} \in X(G)$. Moreover,
$\bar{x}^{\ell'}_f = 1$ for $|C|-4$ edges  in  $ \interior(C) \ \backslash \ \left\{ \left\{ v_{i-1},v_{i+1} \right\} , \left\{ v_i, v_{l'} \right\}  \right\}$,  so we have that $\bar{x}^{\ell'} \in F^I$.


Finally, we have $\mu \bar{x} = \mu_0$ and $\mu \bar{x}^{\ell'} = \mu_0$, and the subtraction of these two equalities yields   $\mu_{\{v_{\ell'-1},v_{\ell'+1}\}} = \mu_{\{v_{i},v_{\ell'}\}}$.  Since $\{v_{\ell'-1},v_{\ell'+1}\} \in \mathrm{int}(C'), \lambda = \mu_{\{v_{\ell'-1},v_{\ell'+1}\}}$ and $\mu_{\{v_{i},v_{\ell'}\}} = \lambda$.
$\Halmos$
\end{proof}

From the previous claims, we have that any solution in $F^I$ yields $\mu_0 = \lambda \left( |C| - 4 \right)$, as desired. 
$\Halmos$
\endproof

\section{Additional Proofs for Section \ref{sec:GeneralPolyhedralProperties}}
\label{ec:generalProperties}

\smallskip
%
%

\begin{lemma} \label{lem:ChordalExtensions}
If $G = (V,E)$ is a chordal graph, then the graph $G' = (V \cup w,E \cup \left\{ (w,v) : \forall \: v \in V \right\}$ is also chordal.
\end{lemma}
\proof{Proof.}
Suppose by contradiction that $C = (v_0,v_1,\ldots,v_{k-1}), k \geq 4$, is a chordless cycle in $G'$.  
As~$G$ does not contain chordless cycles,  $V(C)$ cannot be contained in $V(G)$, so $w \in V(C)$. By construction, $w$ is adjacent to all vertices in~$V(G)$ and, in particular, to all vertices in~$V(C)$, contradicting thus the hypothesis that~$C$ is chordless. $\Halmos$
\endproof

Lemma~\ref{lem:ChordalExtensions} can be extended to cliques as opposed
to single vertices, since this addition can be seen as a inductively adding a single vertex one-by-one. This is formalized in the following immediate corollary.

\begin{corollary} \label{cor:ChordalExtensions}
If $G = (V,E)$ is a chordal graph, then the graph $G' = (V \cup W,E \cup \left\{ (w,v) : \forall w \in W, \: v \in V \cup W \backslash \{w\} \right\}$ is also chordal.
\end{corollary}

\medskip

\begin{definition}
An edge $e$ is said to be \emph{critical} in a chordal graph $G = (V,E)$ if $G' = (V,E \backslash e)$ is not chordal (i.e., the removal of $e$ from $G$ creates a chordless cycle).
\end{definition}

\medskip

\begin{lemma} \label{lem:CriticalEdge}
Let $G = (V,E)$ be a chordal graph.  If $e = \{v,w\}$ is critical, then any chordless cycle $C$ emerging after the deletion of $e$  is such that $\{v,w\} \subset V(C)$ and $|C|=4$.
\end{lemma}
\proof{Proof.}
Let~$C$ be a chordless cycle emerging after the deletion of $e$. If either~$v$ or~$w$ does not belong to~$V(C)$, then $C$ is also a chordless cycle in~$G$, so~$G$ is not chordal, a contradiction.

Suppose $|C| > 4$. In this case, $C$ can be written as a sequence $v \sim P_1 \sim w \sim P_2 $, where $P_1$ and $P_2$ are paths in $G$ such that $V(P_1) \cap V(P_2) = \emptyset$ and $v,w \notin V(P_1) \cup V(P_2)$.  Moreover, as $|C| > 4$, at least one of $P_1,P_2$ contains 2 or more vertices.  If $|P_1| > 1$ ($|P_2| > 1$), then the sequence described by path $v \sim P_1 \sim w$ ($w \sim P_2 \sim v$) induces a chordless cycle in $G$, thereby contradicting the assumption that $G$ is chordal. 
$\Halmos$ 
\endproof

\proof{Proof of Theorem \ref{thm:FacetsforInducedSubcycles}.}
This follows directly from Theorem~\ref{thm:FacetsforInducedSubcycles2} presented next. \Halmos
\endproof

\begin{theorem}\label{thm:FacetsforInducedSubcycles2}
Let $G=(V,E)$ and $E' \subseteq E^C$ be such that $G+E'$ is not chordal and $G + E' \backslash \{f\}$ is chordal for every $f \in E'$.  If $ax \ge b$ is facet-defining for $\mathrm{conv}(X(G+E')), a \ge 0$, and $a' \in \mathbb{R}^{\left| E^C(G) \right|}$, with $a'_f = a_f$ if $f \in E^C \backslash E'$ and $a'_f = 0$ otherwise, the inequality 
\[ a'x \geq b \left( \sum_{f \in F_G(C)} x_{f} - |E'| + 1  \right) \]
is facet-defining for $\mathrm{conv}(X(G))$.
\end{theorem}

\proof{Proof of Theorem \ref{thm:FacetsforInducedSubcycles2}}
Let $ax \geq b$ be a facet-defining inequality for $\textnormal{conv}(X(G+E'))$ and $I$ be the corresponding lifted inequality $a'x \geq b\left( \sum_{f \in E'} x_{f} - |E'| + 1  \right)$ for $\textnormal{conv}(X(G))$.

First, we show that~$I$ is valid for~$\textnormal{conv}(X(G))$. Since $a \geq 0$, $I$ can only be violated by a feasible element~$x$ of~$\textnormal{conv}(X(G))$  if  $\sum_{f \in E'}x_f = |E'|$; otherwise, $I$ is trivially satisfied. Moreover, because $ax' \geq b$ is valid for every $x' \in \textnormal{conv}(X(G+E'))$, we have
\[
a'x = \sum_{f \in E^C \setminus E'}a_fx_f \geq b = b\left( \sum_{f \in E'} x_{f} - |E'| + 1  \right),
\]
as desired.

Now we present a set of $|E^c|$ affinely independent vectors of~$\textnormal{conv}(X(G))$  satisfying~$I$ at equality.  For any facet-defining inequality $ax \geq b$ of $\textnormal{conv}(X(G+E'))$, there exists an affinely independent set of vectors $W = \{w^j\}_{j=1}^d \subseteq \{0,1\}^{|E^C \setminus E'|}$ 
that satisfy $ax = b$.  Let $T = \{t^j\}_{j=1}^d \subseteq \{0,1\}^{|E^C|}$ be such that
\[
t^j_f = 
\left\{
\begin{array}{ll}
w^j_f, & \quad f \in E^C \setminus E', \\ 
1		, & \quad f \in E'.
\end{array}
\right.
\]
That is, $t^j$ is an embedding of~$w^j$ in $\{0,1\}^{|E^C|}$ in which coordinates associated with edges in~$E'$ are set to 1.  Note that every $t^j$ belongs to $\textnormal{conv}(X(G))$ because~$G(t^j)$ is isomorphic to $\left(G+ E'\right)(w^j)$, which is chordal. Moreover, by construction, $a't^j = \sum_{f \in E^C \setminus E'} a'_f t^j_{f} = b$ and $\sum_{f \in E'} t^j_{f} = |E'|$ for each~$j = 1, \ldots, d$, so solutions of~$T$ satisfy~$I$ at equality. Finally, note that the embedding operation in the elements of~$W$ is such that~$T$ is also affinely independent.

Let $Z= \{z^f\}_{f \in E'} \subseteq \{0,1\}^{|E^C|}$ be such that
\[ 
z^f_{f'} = 
\left\{
\begin{array}{ll}
1 & \quad f' \in E' \setminus  f, \\ 
0		, & \quad \text{otherwise}.
\end{array}
\right.
\]
As $G + E' \setminus \{f\}$ is chordal by hypothesis, it follows that each solution~$z^f$ belongs to~$\textnormal{conv}(X(G))$. Moreover, by construction, $a'z^f = \sum_{f' \in E^C \setminus E'} a'_{f'} t^j_{f'} = 0$ and $\sum_{f' \in E'} z^f_{f'} = |E'|-1$ for each~$f \in E'$, so each solution of~$Z$ satisfies~$I$ at equality. Let $\alpha_f, f \in E^C \setminus E'$, and $\beta_f, f \in E'$, be constants for which
\[ 
\sum_{f \in E^C \setminus E'} \alpha_j t^j + \sum_{f \in E'} \beta_{f} z^f = 0, \quad \sum_{f \in E^C \setminus E'} \alpha_j + \sum_{f \in E'} \beta_{f} = 0.  
\]
For each $f \in E'$, we have~$z^f_f = 0$, whereas $s_f = 1$ for $s \in T \cup Z \setminus \{z^f\}$. Therefore, we must have $\sum_{f' \in E^C \setminus E'} \alpha_{f'}  + \sum_{f' \in E' \setminus \{f\} } \beta_{f'}  = 0$ and, as a consequence, $\beta_{f} = 0$ for each $f \in E'$. Finally, as $T$ is affinely independent, we have that $\alpha_f = 0$, $f \in E^C \setminus E'$. It follows that~$I$ is a faced-defining inequality for $\textnormal{conv}(X)$, as desired. $\Halmos$
\endproof

\proof{Proof of Theorem \ref{thm:ChordFacets}.}
Without loss of generality, let $f^* = \{ v_0,v_{t-1} \}, t < k$, be the chord considered and~$I$ be the associated lifted inequality $a'x - bx_{f^*} \geq 0$ for~$X(G)$.
For any vector $x \in \{0,1\}^{|E^c|}$ and set $E' \subseteq E^c$, let $x[E']$ be the projection of~$x$ onto the coordinates corresponding to fill edges in $E '$.

First, we claim that~$I$ is valid for~$X(G)$.  Take any solution $x^0 \in X(G)$. If $x^0_{f^*} = 0$, then~$I$ reduces to $a'x^0 \geq 0$, which must be satisfied because $a' \geq 0$ and $x^0 \geq 0$.  If $x^0_{f^*} = 1$, then $I$ reduces to $a'x \geq b$.  Since $G(x^0)$ is chordal, by Lemma~\ref{lem:InducedSubgraph} we have that 
$G(x^0)[C']$
is also chordal, and therefore $x^0[\interior(C')] \in X(G')$. Since $I$ is facet-defining for 
$\mathrm{conv}(X(G'))$, we have $a'x^0 = ax^0[\interior(C')] \geq b$.  As $x^0$ was chosen arbitrarily among all feasible solutions in~$X(G)$, it follows that~$I$ is valid for~$X(G)$. 

Let  $C' = \left(v_0,v_1,\ldots,v_{t-1}\right),C'' = \left(v_0,v_{t-1},v_{t},\ldots,v_{k-1}\right)$, and 
\[ 
\mathrm{Cross}(f^*) = \{ f : f \cap \{ v_1,v_2,\ldots,v_{t-2} \} \neq \emptyset  , f \cap \{ v_{t},v_{t+1},\ldots,v_{k-1}  \} \neq \emptyset \};
\]
that is, $\mathrm{Cross}(f^*)$ contains all fill edges in~$\interior(C)$ containing exactly one vertex incident in~$C' \setminus C''$ and one vertex incident in~$C'' \setminus C'$. 
Set~$\interior(C)$ can therefore be partitioned as follows:
\[
\interior(C) = 
 \:
\interior(C')
\: \dot{\cup} \:
\interior(C'')
\: \dot{\cup} \:
f^* 
\: \dot{\cup} \:
\mathrm{Cross}(f^*)
\]

Let $F^I$ be the set of points in $\mathrm{conv}(X(G))$ that satisfy $I$ at equality, and $\mu x \geq \mu_0$ be a valid inequality for $\mathrm{conv}(X(G))$ satisfied at equality by each $x \in F^I$.  Inequality $\mu x \geq \mu_0$ can be written as
\[ 
\sum_{f \in \mathrm{int}(C')} \mu_f x_f + \sum_{f \in \mathrm{int}(C'')} \mu_f x_f + \mu_{f^*} x_{f^*} + \sum_{f \in \mathrm{Cross}(f^*)} \mu_f x_f \geq \mu_0 
\]

\begin{claim}
For every $f$ in $\mathrm{Cross}(f^*),  \mu_f = 0$.
\end{claim}
\begin{proof}{Proof.}
Take any vector $\tilde{w}^b$ in $X(G')$ 
such that $a\tilde{w}^b = b$. Moreover, let us assume that the fill in set associated with~$\tilde{w}^b$ is minimal; note that if $\tilde{w}^b$ does not satisfy this condition, then it can be substituted for some other feasible solution $w'$, $aw' = a\tilde{w}^b = b$, associated with a subset of the fill in edges represented by~$\tilde{w}^b$.  

From Proposition~\ref{prop:inq3}, it follows that~$G'(\tilde{w}^b)$ must contain an edge $\{v_{b-1},v_{b+1}\}$. Moreover, from Lemma~\ref{lem:InducedSubgraph}, we have that $w' = \tilde{w}^b[E^c(G') \setminus \{ \{v_a,v_b\}: v_a \in V(G') \}]$ is associated with a chordal completion of~$G'[V(G) \setminus v_b]$. Because $v_{b-1}$ and $v_{b+1}$ are the only neighbours of~$v_b$ in~$G'$, the edges of~$w'$ are sufficient to make~$G'$ chordal; therefore, we have that the neighbours of~$v_b$ in~$G'[\tilde{w}^b]$ are exactly its neighbours in~$C'$.


Fix $f' \in \mathrm{Cross}(f^*), f' = \{v_a,v_b\}, 1 \leq a \leq t-2, t \leq b \leq k-1$.  Let $\tilde{w}$ in $X(G)$ be defined by
\[
\tilde{w}_f = 
\left\{
\begin{array}{ll}
\tilde{w}^b_f, & \quad f \in \interior(C') \\ 
1			, & \quad f = f^* \\
1		, & \quad f \in \interior(C'') \\
0			, & \quad f = \{ v_a,v_b \} \\
1		, & \quad f \in \mathrm{Cross}(f^*) \backslash \{ v_a,v_b \}
\end{array}
\right.
\]
By Lemma~\ref{lem:ChordalExtensions}, $G\left(\tilde{w} + e^{\{ v_a,v_b \}} \right)$ is a chordal graph. We claim that $\{ v_a,v_b \}$ cannot be critical, and therefore $G\left( T^1\left(\tilde{w}^b\right) \right)$ is chordal. Suppose by contradiction that this is not true. Then,  
upon the removal of $\{ v_a,v_b \}$, by Lemma~\ref{lem:CriticalEdge} there must exists vertices $v',v''$ for which $(v_a,v',v_b,v'')$ is a chordless cycle.  This can only happen if there exists a pair of vertices in~$N(v_b)$  which are not adjacent.  However, $N(v_b) = \{ v_{b-1}, v_{b+1} \} \cup \{ v_t, v_{t+1}, \ldots, v_{k-1} \}$, which, by construction, is a clique. 

Therefore, we have that $\tilde{w}$ and $\tilde{w} + e^{\{ v_a,v_b \}}$ belong to~$\mathrm{conv}(X(G))$.  Additionally, both solutions satisfy~$I$ at equality and belong thus to~$F^I$.  Finally, as 
\[
\mu(\tilde{w} + e^{\{ v_a,v_b \}} - \tilde{w}) = \mu_{\{ v_a,v_b \}} = 0,
\]
it follows that $\mu_{f'} = 0$ for every $f' \in \mathrm{Cross}(f^*)$. $\Halmos$
\end{proof}


\begin{claim}
For every~$f \in \mathrm{int}(C'')$,  $\mu_f = 0$.
\end{claim}
\begin{proof}{Proof.}
Fix $f' = \{ z_1,z_2 \} \in \interior(C'')$ and any solution $\tilde{w}^b$ in $X(G')$ 
such that $a\tilde{w}^b = b$. Let $\tilde{w}$ be defined by 
\[
\tilde{w}_f = 
\left\{
\begin{array}{ll}
\tilde{w}^b_f, & \quad f \in \interior(C') \\ 
1			, & \quad f = f^* \\
0		, & \quad f  \in \mathrm{Cross}(f^*) \\
0			, & \quad f = f' \\
1		, & \quad f \in \interior(C'') \backslash f'
\end{array}
\right.
\]
We claim that $G(\tilde{w} + e^{f'})$ is chordal.  Consider the ordering~$\pi$ of the vertices in $V(G)$ consisting of a perfect elimination order of the vertices in $V(C')$ (which must exists because $G[V(C')](\tilde{w}^b)$ is chordal), followed by an arbitrary ordering of the remaining vertices. Because the neighbourhood of each vertex in $V(C'') \setminus V(C')$ is a clique in~$V(C'')$, it follows by construction that~$\pi$ is a perfect elimination ordering for the vertices of~$G(\tilde{w} + e^{f'})$.


We claim now that~$G(\tilde{w})$ is also chordal.  If not, by Lemma~\ref{lem:CriticalEdge} there must exist a chordless cycle $(z_1,v',z_2,v'')$ created upon the removal of $f'$ from $G(\tilde{w} + e^{f'})$.  
At least one among $z_1$ and $z_2$ is contained in $\{v_t, v_{t+1}, \ldots, v_{k-1}\}$;  let~$z_1$ be one such vertex.  The neighborhood of $z_1$ in $G(\tilde{w})$ is $V(C'')$, and as~$G(\tilde{w})[V(C'')]$ is a clique, we must have $\{v',v''\} \in E(G(\tilde{w}))$, a contradiction.

Therefore,  we have that $\tilde{w}$ and $\tilde{w} + e^{f'}$ belong to $\mathrm{conv}(X(G))$ and, by construction, to~$F^I$. Similar arguments to those used in the previous claim allow us to conclude that
$\mu_{f'} = 0$ for every $f'\in \mathrm{int}(C'')$.$\Halmos$
\end{proof}

\begin{claim}
$\mu_0 = 0$.
\end{claim}
\begin{proof}{Proof.}
Consider the solution $\hat{w}$ defined by
\[
\hat{w}_f = 
\left\{
\begin{array}{ll}
1, & \quad v_{k-1} \in f \\ 
0, & \quad \mbox{ otherwise }		
\end{array}
\right.
\]
This solution is isomorphic to the solution~$x'$ constructed in the proof of Proposition~\ref{prop:inq1}, so $G(\hat{w})$ is chordal. By construction, because $f^* = \{ v_0,v_{t-1} \}$ for $t < k$,
$\hat{w}_{f^*} = 0$. Moreover, as $a_f = 0$ for $f \notin \interior(C')$, we have $a \hat{w} = 0$,
and therefore $a \hat{w} - b \hat{w}_{f^*} = 0$.  Substituting $\hat{w}$ into $\mu x = \mu_0$ yields
\[ 
\mu_0 = \mu \hat{w} =   \sum_{f \in \mathrm{int}(C')} \mu_f \hat{w}_f + \mu_{f^*} \hat{w}_{f^*} = 0,
\]
as desired. $\Halmos$
\end{proof}

\begin{claim}
There is a $\lambda \in \mathbb{R}$ such that $\mu_{f^*} = -\lambda b$ and $\mu_f = \lambda a_f$
for every $f$ in $\interior(C')$.
\end{claim}
\begin{proof}{Proof.}
Let $\tilde{F}^I$ be the subset of $F^I$ containing only solutions~$x$ such that $x_f = 1$ for every edge $f$ which does not belong to $\interior(C')$.
For every $x \in \tilde{F}^I$, we have
\begin{eqnarray*} 
\mu x = \sum_{f \in \interior(C')} \mu_f x_f + \mu_{f^*} 1  = 0 \implies \\ 
\sum_{f \in \interior(C')} \mu_f x_f = - \mu_{f^*}.
\end{eqnarray*}
Consequently, we have that every solution~$y$ in $X(G')$ that satisfies $ay = b$ must also satisfy
$\mu[\interior(C')] y = - \mu_{f^*}$.
As $ay \geq b$ is facet-defining for~$X(G')$, there exists some $\lambda$ such that $-\mu_{f^*} = \lambda b$ and $\mu_f = \lambda' a_f$ for every $f$ in $E(G')^c$, as desired. $\Halmos$
\end{proof}

From the previous claims, we conclude that~$I$ is a facet-defining inequality for~$\mathrm{conv}(X(G))$. $\Halmos$
\endproof

\bigskip

\section{Additional Proofs for Section~\ref{sec:SeparationComplexity}}
\label{ec:separationProofsAlgorithms}



\medskip

We first provide a lemma.

\begin{lemma}
\label{lemma:minimali1}
For any fractional point~$x \in [0,1]^{m^c}$, if~$x \notin \textnormal{conv}(X(G))$, then there is a chordless cycle~$C$ in $G + E(x)$ whose associated inequality of type~(\ref{ineq:ci}) is violated by~$x$.
\end{lemma}
\begin{proof}{Proof.}
Suppose by contradiction that this claim does not hold, and let~$C$ be a cycle in $G$ associated with a violated inequality of type~(\ref{ineq:ci}) such that $\interior(C) \cap E(G)$ is minimum. Set $\interior(C)$ must contain at least one edge~$e$ in~$E(G)$, so let~$C'$ and $C''$ be the sub-cycles of~$C$ such that $V(C') \cap V(C'') = e$, $V(C') \cup V(C'') = V(C)$,  and $E(C') \cap E(C'') = \{e\}$; by construction, we have  $|C'| + |C''| = |C| + 2$. 

If~$x$ satisfies the inequalities~(\ref{ineq:ci}) associated with~$C'$ and~$C''$, we have
\[
\sum_{e \in \interior(C)}x_e \ge \sum_{e \in \interior(C')}x_e + \sum_{e \in \interior(C'')}x_e + 1 \geq |C'| - 3 + |C''| - 3 + 1 = |C| + 2 -5 = |C| -3,
\]
contradicting hence the fact that~$C$ does not satisfy inequality~\ref{ineq:ci}. Therefore, $x$ must violate inequality~(\ref{ineq:ci}) for $C'$ or $C''$; let us assume that the violation holds for~$C'$. If~$\interior(C')$ does not contain any edge in~$E(G)$, we have a contradiction. Otherwise,  we must have $|\interior(C') \cap E(G)| < |\interior(C) \cap E(G)|$, which contradicts the selection of~$C$.
$\Halmos$
\end{proof}

\proof{Proof of Theorem \ref{thm:separation}(a).}
We show this result by proving that the $(-3)$-Quadratic Shortest Cycle Problem (or $(-3)$-QSCP), defined below, can be reduced to the the separation of the simplified version of inequalities~\ref{ineq:ci}. Lemma~\ref{lemma:minimali1} allows us to conclude that these two problems are equivalent, so the main step of the proof consists of showing that $(-3)$-QSCP is $NP$-complete.

We define the Quadratic Shortest Cycle Problem (QSCP) as follows: we are given an undirected graph $G = (V, E)$ and a quadratic cost function $q:V \times V \rightarrow [0,1]$ such that $q(u,v) = 0$ if $(u,v) \in E$; that is, the quadratic cost associated with~$\{u,v\}$  can be different from zero only if $\{u,v\} \notin E$. For any cycle $C$ in $G$, let $int^*(C) = E(G[C])^C$, that is, edge $\{u,v\}$ belongs to $int^*(C)$ if $u \neq v$ and $(u,v) \notin E$.
A feasible solution for an instance of QSCP consists of a simple chordless cycle $C = (v_1,v_2,\ldots, v_{|C|})$ 
whose cost $p(C)$ is given by
\[
p(C) = \sum_{\{u,v\} \in int^*(C)}q(u,v)   -|C|.  
\]
Finally, $\alpha$-QSCP is  the decision version of QSCP  in which the goal is to decide whether there is 
a simple chordless cycle~$C$ such that $p(C) < \alpha$.

\begin{lemma}
$(-3)$-QSCP is $NP$-complete.
\end{lemma}
\begin{proof}{Proof.}
Our proof employs a reduction of the Quadratic Assignment Problem (QAP) to $(-3)$-QSCP. This strategy is based on the reduction used by~\cite{rostami2015} to show that the Quadratic Shortest Path Problem is strongly NP-hard.

\paragraph*{QAP description:}
For an arbitrary instance~$I$ of QAP, let $F$ and~$L$ be the set of facilities and locations, respectively, with $n = |F| = |L|$, and let $C$, $D$, and~$A$ be the $n \times n$ matrices in $\mathbb{R}^+$ describing the flow between facilities, the distance between locations, and the cost of assigning facilities to locations, respectively; recall that linear costs are given by entries of~$A$, whereas quadratic costs are associated with the multiplication of one entry in~$C$ (referring to the flow of products between facilities) by some other entry in~$D$ (representing the distance between locations). Finally, we are given a value $\beta$, and the goal is to decide whether the instance of QAP admits an assignment whose cost is smaller than~$\beta$. This problem is known to be NP-complete (see e.g.,~\cite{GarJoh1979}).

Let 
\[
M = \max\left(
\max\limits_{\substack{1 \leq f,f' \leq |F|\\ 1 \leq l,l' \leq |L|}}C_{f,f'}D_{l,l'},  
\max\limits_{\substack{1 \leq f \leq |F|\\ 1 \leq l \leq |L|}}A_{f,l}  
\right),
\] 
that is, $M$ is the largest individual penalty that may compose the cost of a feasible assignment. Any feasible solution consists of $n$ assignments, so it is subject to not more than~$n^2$ and~$n$ quadratic and linear penalizations, respectively. Therefore, no feasible assignment has an objective value larger than $K = 2Mn^2$.

\paragraph*{Vertices and edges:}

From~$I$, we construct an instance~$I'$ of $(-3)$-QSCP associated with a graph $G = (V,E)$ and a quadratic cost $q:V \times V \rightarrow [0,1]$ as follows. Let us assume w.l.o.g. that there is some (arbitrary) ordering between facilities, that is, $F = \{f_1,f_2,\ldots,f_{n}\}$.

Set $V$ contains one \textit{assignment vertex} $a_{f,l}$ for each pair $(f,l) \in F \times L$; these vertices can be interpreted as the assignment of facilities to locations. We say that a pair of assignment vertices belong to the same \textit{block} if they are associated with the same facility.

For technical reasons, $V$ contains three types of auxiliary variables. We have \textit{type-$z$} vertices~$z_1$ and $z_2$ and \textit{type-$y$} variables $y_1, y_2, \ldots, y_{n}$ whose usage will become clear next. Additionally, for each pair of assignments $(f_i,l)$ and $(f_{i+1},l')$,  $1 \leq i < n$ and $l,l' \in L$, we have a \textit{connection vertex} $c_{f_i,f_{i+1},l,l'}$. $V$ also contains connection vertices $c_{z_1,f_1,\emptyset,l}$ and $c_{f_n,y_{n-1},l,\emptyset}$ for all $l$ in $L$; by an abuse of notation, we might use $c_{f_{0},f_{1},l,l'}$ instead of $c_{z_1,f_1,\emptyset,l'}$ (i.e., substitute $(z_1,\emptyset)$ for $(f_0,l)$) and $c_{f_{n},f_{n+1},l,l'}$ instead of $c_{f_n,z_n,l,\emptyset}$ ($(z_n,\emptyset)$ for $(f_{n+1},l')$ )  in situations where the correct notation can be easily inferred from the context. 
A pair of connection vertices is said to belong to the same \textit{block} if they have the 
same first facility index.

Each assignment vertex~$a_{f_i,l}$ composes edges with connection vertices $c_{f_i,f_{i+1},l,l'}$ and  $c_{f_{i-1},f_i,l'',l}$ for all $l', l''$ in $L$. Moreover, $z_1$ and assignment vertices~$a_{f_1,l}$ are connected to $c_{z_1,f_1,\emptyset,l}$, whereas~$y_{n}$ and assignment vertices~$a_{f_n,l}$  are connected to $c_{f_n,y_{n},l,\emptyset}$, $l \in L$. Note that $a_{f_i,l}$ and $a_{f_{i+1},l'}$ are the only neighbours of $c_{f_i,f_{i+1},l,l'}$, that is, all connection vertices have degree~$2$. 
Finally, $\{z_1,z_2\}$, $\{z_2,y_1\}$, $\{y_1,y_2\}, \ldots, \{y_{n-1},y_{n}\}$ are also edges of~$E$. 
An example of graph associated with an instance of QAP with~$n=3$ is presented in Figure~\ref{fig:i1construction}. 

\begin{figure}
\centering

\begin{tikzpicture}[scale=1.5][font=\sffamily,\scriptsize]

\node (1) [draw, circle] [draw] at (0,0) {$z_1$};

\node (2) [draw]  at (1,1)  {$c_{z_1,f_1,\emptyset,l_1}$};
\node (3) [draw] [draw] at (1,0) {$c_{z_1,f_1,\emptyset,l_2}$};
\node (4) [draw] [draw] at (1,-1) {$c_{z_1,f_1,\emptyset,l_3}$};

\path[-](1) edge (2);
\path[-](1) edge (3);
\path[-](1) edge (4);

\node (5) [draw,circle] [draw] at (2,2) {$a_{f_1,l_1}$};
\node (6) [draw,circle] [draw] at (2,0) {$a_{f_1,l_2}$};
\node (7) [draw,circle] [draw] at (2,-2) {$a_{f_1,l_3}$};

\path[-](2) edge (5);
\path[-](3) edge (6);
\path[-](4) edge (7);

\node (8) [draw] [draw] at (3,2.5) {$c_{f_1,f_2,l_1,l_2}$};
\node (9) [draw] [draw] at (3,1.5) {$c_{f_1,f_2,l_1,l_3}$};

\path[-](5) edge (8);
\path[-](5) edge (9);

\node (10) [draw] [draw] at (3,0.5) {$c_{f_1,f_2,l_2,l_1}$};
\node (11) [draw] [draw] at (3,-0.5) {$c_{f_1,f_2,l_2,l_3}$};

\path[-](6) edge (10);
\path[-](6) edge (11);

\node (12) [draw] [draw] at (3,-1.5) {$c_{f_1,f_2,l_3,l_1}$};
\node (13) [draw] [draw] at (3,-2.5) {$c_{f_1,f_2,l_3,l_2}$};

\path[-](7) edge (12);
\path[-](7) edge (13);

\node (14) [draw,circle] [draw] at (4,2) {$a_{f_2,l_1}$};
\path[-](10) edge (14);
\path[-](12) edge (14);

\node (15) [draw,circle] [draw] at (4,0) {$a_{f_2,l_2}$};
\path[-](8) edge (15);
\path[-](13) edge (15);

\node (16) [draw,circle] [draw] at (4,-2) {$a_{f_2,l_3}$};
\path[-](9) edge (16);
\path[-](11) edge (16);

\node (17) [draw] [draw] at (5,2.5) {$c_{f_2,f_3,l_1,l_2}$};
\node (18) [draw] [draw] at (5,1.5) {$c_{f_2,f_3,l_1,l_3}$};

\path[-](14) edge (17);
\path[-](14) edge (18);

\node (19) [draw] [draw] at (5,0.5) {$c_{f_2,f_3,l_2,l_1}$};
\node (20) [draw] [draw] at (5,-0.5) {$c_{f_2,f_3,l_2,l_3}$};

\path[-](15) edge (19);
\path[-](15) edge (20);

\node (21) [draw] [draw] at (5,-1.5) {$c_{f_2,f_3,l_3,l_1}$};
\node (22) [draw] [draw] at (5,-2.5) {$c_{f_2,f_3,l_3,l_2}$};

\path[-](16) edge (21);
\path[-](16) edge (22);

\node (23) [draw,circle] [draw] at (6,2) {$a_{f_3,l_1}$};
\path[-](19) edge (23);
\path[-](21) edge (23);

\node (24) [draw,circle] [draw] at (6,0) {$a_{f_3,l_2}$};
\path[-](17) edge (24);
\path[-](22) edge (24);

\node (25) [draw,circle] [draw] at (6,-2) {$a_{f_3,l_3}$};
\path[-](18) edge (25);
\path[-](20) edge (25);

\node (26) [draw] [draw] at (7,1)  {$c_{f_3,y_3,l_1,\emptyset}$};
\path[-](23) edge (26);

\node (27) [draw] [draw] at (7,0) {$c_{f_3,y_3,l_2,\emptyset}$};
\path[-](24) edge (27);

\node (28) [draw] [draw] at (7,-1) {$c_{f_3,y_3,l_3,\emptyset}$};
\path[-](25) edge (28);

\node (29) [draw,circle] [draw] at (8,0) {$y_3$};
\path[-](26) edge (29);
\path[-](27) edge (29);
\path[-](28) edge (29);

\node (30) [draw,circle] [draw] at (1,-3) {$z_2$};
\path[-](1) edge (30);

\node (31) [draw,circle] [draw] at (4,-3) {$y_1$};
\path[-](30) edge (31);

\node (32) [draw,circle] [draw] at (7,-3) {$y_2$};
\path[-](31) edge (32);

\path[-](32) edge (29);

\end{tikzpicture}

\caption{Construction for QAP instance with $n = 3$.}
\label{fig:i1construction}
\end{figure}
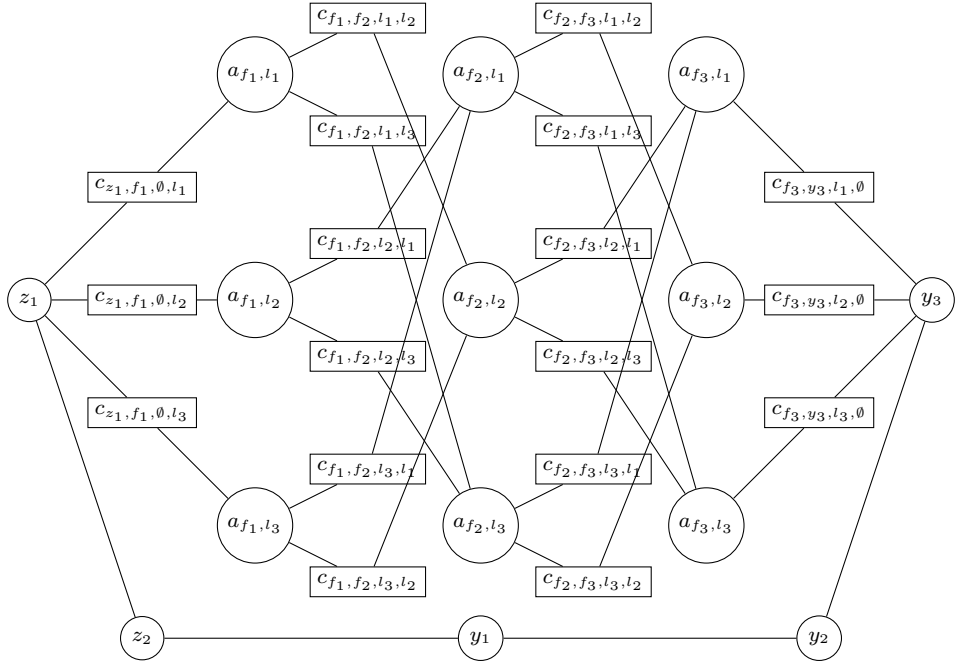

\paragraph{Penalties:}

The goal of the construction is to enforce every algorithm deciding $(-3)$-QSCP to deliver a cycle $C$ as solution for $I'$ if and only if $I$ admits an assignment whose cost is inferior to~$\beta$. 
We say that~$C$ is associated with the solution of~$I$ containing each assignment~$(f_i,l_j)$ such that $a_{f_i,l_j}$ belongs to~$V(C)$.

Let~\texttt{Alg} be an algorithm deciding $(-3)$-QSCP. \texttt{Alg} can return~$C$ only if 
\begin{eqnarray*}
p(C) = p_L(C) + p_Q(C) 
+ q^*(C) < -3,
\end{eqnarray*}
where~$q^*(C)$ is the sum of $-|C|$ with all additional costs that will be incorporated in our construction and~$p_L(C)$ and~$p_Q(C)$ denote the linear and the quadratic costs of~$I$ mapped into~$C$, respectively. For technical reasons described below, the original (linear and quadratic) costs of~$I$ will be divided by~$K$ in $I'$. As a consequence, the assignment associated with $C$ is a solution of $I$ if it is feasible and
\[
p_L(C) + p_Q(C) < \frac{\beta}{K},
\]
so we just need to define~$q^*(C)$ in a way that
\[
-3 - q^*(C) = \frac{\beta}{K} \implies q^*(C) = -3 - \frac{\beta}{K}.
\]
Note that $\beta \leq K$; otherwise, any feasible assignment decides~$I$.

In summary, the costs composing~$q^*(C)$ should guarantee that \texttt{Alg} returns a cycle if and only if the associated assignment in $I$ is feasible with cost inferior to~$\beta$.
For this, our construction restricts the set of cycles \texttt{Alg} may select to \textit{matching cycles}, which are cycles in~$G$ of size $3n + 3$ that pass through all type-$y$ and type-$z$ vertices and are associated with a feasible assignment for the QAP instance whose cost is below~$\beta$. The rest of the proof shows how the function~$q$ enforces the satisfaction of these conditions. 

\paragraph{Original QAP penalties:}
For each pair of assignment vertices~$a_{f,l}$ and $a_{f',l'}$, $f \neq f'$ and $l \neq l'$, we have the \textit{assignment cost}
\[
q(a_{f,l},a_{f',l'}) = \frac{C_{f,f'}D_{l,l'}}{K},
\]
which represents the quadratic cost of $I$ associated with assignments~$(f,l)$ and~$(f',l')$. 

Additionally, we have \textit{linear costs}
\[
q(z_1,a_{f,l}) = \frac{A_{f,l}}{K},
\]
that is, $q(z_1,a_{f,l})$ contains the linear cost of~$I$ associated with assignment~$(f,l)$. 

Note that the scaling factor $\frac{1}{K}$ enforces all values to belong to~$[0,\frac{1}{2n^2}]$ and, consequently, the sum of these penalties is bounded by $1$ for any feasible assignment in~$I$. In particular, any cycle~$C$ of~$G$ associated with a feasible assignment is such that
$0 \leq p_L(C) + p_Q(C) \leq 1$.

\paragraph{Infeasibility penalties:}

Cycle $C$ cannot be a matching cycle 
if $V(C)$ contains one or more pairs of  assignment vertices sharing the same location or facility.
In order to avoid these configurations, we set \textit{assignment conflict costs}
\[
q(a_{f,l},a_{f',l'}) = 1 
\]
for every pair of assignment vertices~$a_{f,l}$ and~$a_{f',l'}$ such that either $f = f'$ or $l = l'$.  
Note that this penalty is not smaller  than the (scaled) cost of any feasible solution  of~$I$. 

Similar penalizations will be applied to 
pair of connection vertices belonging to the same block. That is, given connection vertices $c_{f_{i},f_{i+1},l,l'}$ and $c_{f_{i},f_{i+1},l'',l'''}$, $1 \leq i \leq n$ and $l,l',l'',l''' \in L$, we have
\textit{transition conflict costs}
\[
q(c_{f_{i},f_{i+1},l,l'},c_{f_{i},f_{i+1},l'',l'''}) = 1. 
\]

\paragraph{Sub-cycle elimination:}

In order to avoid the selection of cycles which do not pass through type-$y$ and type-$z$ vertices, we penalize pairs of connection vertices belonging to consecutive blocks. That is, given connection vertices $ c_{f_i,f_{i+1},l,l'}$ and $c_{f_{i+1},f_{i+2},l'',l'''}$, $0 \leq i < n$ and $l,l',l'',l''' \in L$, we have \textit{transition penalties}
\[
q(c_{f_i,f_{i+1},l,l'},c_{f_{i+1},f_{i+2},l'',l'''}) = 1. 
\]
Note that this penalty incurs $n$ times in matching cycles.

\paragraph{Compensation penalties:}

Penalties described here are used to compensate for the inclusion of vertices and to make the solution respect the upper bound~$\beta$ associated with~$I$.

The inclusion of connection vertices is compensated by their quadratic costs with~$z_2$. That is, for every connection vertex~$c_{f_i,f_{i+1},l,l'}$, $0 \leq i < n-1$ and $l,l' \in L$ (note that connection vertices who are neighbours of~$y_n$ are excluded), we have \textit{connection-covering costs}
\[
q(z_2,c_{f_i,f_{i+1},l,l'}) = 1. 
\]
In matching cycles, connection-covering costs incur $n$ times.

The costs of type-$y$ vertices are  covered by quadratic assignments involving $z_1$ and connection vertices~$c_{f_i,f_{i+1},l,l'}$, $1 \leq i < n-1$ and $l,l' \in L$ (note that we are excluding the neighbours of~$z_1$ and$y_n$). These $y$-covering costs are given by
\[
q(z_1,c_{f_i,f_{i+1},l,l'}) = 1. 
\]
In matching cycles, these costs incur $n-1$ times.

So far, the sum of the compensation penalties with the transition penalties for any matching cycle~$C$ is equal to $n + n + n-1$. Because $|C| = 3n+3$, there is a deficit of $1 - \frac{\beta}{K} < 1$ in $p(C)$. For this, we employ the quadratic cost of~$y_n$ and~$z_1$, that is,
\[
q(z_1,y_n) = 1 - \frac{\beta}{K}.
\] 
Finally, all the remaining 
costs that have not been explicitly presented are set to zero.

\begin{lemma}\label{lemma:nosubcycles}
Every cycle  delivered by \texttt{Alg} 
must contain all type-$z$ and type-$y$ vertices.
\end{lemma}
\begin{proof}{Proof.} 
Let us assume that~$C$ does not include some type-$y$ or type-$z$ vertex; by construction, a cycle in $G$ contains either all vertices in $\{z_2, y_1, y_2, \ldots, y_{n-1}\}$ or none of them, so $C$ may only contain $z_1$, $y_{n}$, assignment vertices, and connection vertices. Consequently, $C$ belongs to a bipartite region of~$G$ (with one part being composed of connection vertices), so~$|C|$ must be even and larger than~$4$. See Figure~\ref{fig:subcycle} for an example with $|C| = 12$. Set~$V(C)$ can be partitioned as follows:

\begin{figure}
\centering
\begin{tikzpicture}[scale=1.15][font=\sffamily,\scriptsize]

\node (1) [draw,circle] [draw] at (0,0) {$z_1$};

\node (2)  [draw] at (1,1) {$c_{z_1,\emptyset,f_1,l_1}$};

\node (3)  [draw] at (1,-1) {$c_{z_1,\emptyset,f_1,l_4}$};

\node (4) [draw,circle] [draw] at (2,2) {$a_{f_1,l_1}$};

\node (5) [draw,circle] [draw] at (2,-2) {$a_{f_1,l_4}$};

\node (6)  [draw] at (3,1) {$c_{f_1,l_1,f_2,l_2}$};

\node (7)  [draw] at (3,-1) {$c_{f_1,l_4,f_2,l_5}$};

\node (8) [draw,circle] [draw] at (4,2) {$a_{f_2,l_2}$};

\node (9) [draw,circle] [draw] at (4,-2) {$a_{f_2,l_5}$};

\node (10) [draw] at (5,1) {$c_{f_2,l_2,f_3,l_3}$};

\node (11) [draw] at (5,-1) {$c_{f_2,l_5,f_3,l_3}$};

\node (12) [draw,circle] [draw] at (6,0) {$a_{f_3,l_3}$};

\path[-](1) edge (2);
\path[-](1) edge (3);

\path[-](2) edge (4);
\path[-](3) edge (5);

\path[-](4) edge (6);
\path[-](5) edge (7);

\path[-](6) edge (8);
\path[-](7) edge (9);

\path[-](8) edge (10);
\path[-](9) edge (11);

\path[-](10) edge (12);
\path[-](11) edge (12);
\end{tikzpicture}
\caption{Sub-cycle of~$G$ with 12 vertices}
\label{fig:subcycle}
\end{figure}
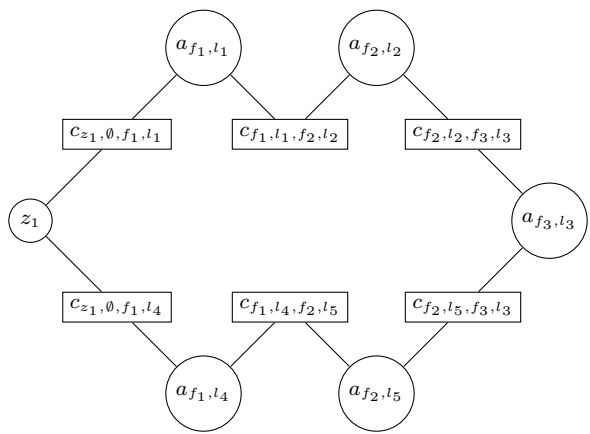

\begin{enumerate}
\item $|C|/2$ connection vertices: Each vertex in this category belongs to the same block as at least some  other connection vertex.
\item 2 vertices, which may be assignment vertices, $y_{n}$, or $z_1$: These are the vertices
located in the extremities of the cycle according to the topology presented in Figure~\ref{fig:subcycle};  on the left, we have either $z_1$ or the assignment vertex with facility of lower index in $C$, whereas on the right we have either~$y_{n}$ or the assignment vertex with facility of highest index in $C$. These vertices stay either alone in their blocks (this is necessarily the case of $z_1$ and $y_{n}$) or together with other assignment vertices in the same block. 
\item $|C|/2 - 2$ assignment vertices: These vertices are located in the middle of the cycle and stay in the same block with at least one other assignment vertex, so the assignment conflict costs involving these vertices is at least $|C|/4 - 1$. 
\end{enumerate}


Connection vertices are distributed among~$k \geq 2$ (consecutive) blocks with $b_i \geq 2$ elements each, $1 \leq i \leq k$. 
These vertices are associated with transition conflict penalties and transition penalties, and the sum 
$p_c(C)$ of all penalties associated with them is 
\[
p_c(C) = \sum_{1 \leq i \leq k} \frac{b_i(b_i - 1)}{2} + \sum_{1 \leq i < k} b_i b_{i+1}. 
\]
Because $2 \leq b_i \leq \frac{|C|}{2} - 2$, we have $b_i b_{i+1} \geq b_i + b_{i+1}$ and $b_i^2 \geq 2b_i$. Therefore, 
\begin{eqnarray*}
p_c(C) &\geq& 
\sum_{1 \leq i \leq k} \frac{b_i(b_i - 1)}{2} + \sum_{1 \leq i < k} b_i b_{i+1} \geq
 \sum_{1 \leq i \leq k} \frac{b_i^2}{2} - \sum_{1 \leq i \leq k} \frac{b_i}{2}  + \sum_{1 \leq i < k} (b_i + b_{i+1}) \\
 &\geq&
 \sum_{1 \leq i \leq k} 2b_i - \sum_{1 \leq i \leq k} \frac{b_i}{2}  + \sum_{1 < i < k} b_{i}
\geq
|C| - \frac{b_1 + b_k}{2} \geq |C| - \frac{|C|}{4} = \frac{3|C|}{4}.
\end{eqnarray*}
By summing all penalties, we have 
\begin{eqnarray*}
p(C) \geq -|C| + |C|/4 - 1 + p_c(C)  > -\frac{3|C|}{4} - 1 + \frac{3|C|}{4} > -1,
\end{eqnarray*}
which is clearly larger than~$-3$. Therefore, \texttt{Alg} can only return cycles~$C$
containing all type-$z$ and type-$y$ vertices. \Halmos 
\end{proof}

\begin{lemma}
\label{lemma:nobadassignment}
Every cycle  delivered by  \texttt{Alg} 
is associated with a feasible assignment.
\end{lemma}
\begin{proof}{Proof.}
Let us suppose by contradiction that \texttt{Alg} delivers a cycle~$C$ which is not associated with a feasible assignment. From Lemma~\ref{lemma:nosubcycles}, 
it follows that every cycle delivered by  \texttt{Alg}  necessarily contains at least one assignment vertex containing each facility. Thus, if the assignment associated with~$C$ is infeasible, then some location is  being assigned to at least two different facilities. Compensation penalties are not affected by this, so 
\[ q^*(C) = -3 - \frac{\beta}{K}. \]
As each assignment conflict cost is equal to 1, we have
\begin{eqnarray*}
p(C) &=& p_L(C) + p_Q(C) + q^*(C) \\
p(C) &\geq& 1 + q^*(C) 
	 \geq 1 -3 - \frac{\beta}{K},
\end{eqnarray*}
and as $\frac{\beta}{K} \leq 1$, $p(C) \geq -3$, and therefore~$C$ cannot be delivered by \texttt{Alg}. 
$\Halmos$
\end{proof}

The previous lemmas show that  \texttt{Alg}  decides $(-3)$-QSCP positively on~$G$ using~$C$ only if~$C$ is a matching cycle.   A similar process can be used in order to construct a solution for $(-3)$-QSCP on~$G$ given a feasible solution for the QAP instance; namely, just take the cycle containing the associated assignment vertices, the connection vertices uniquely determined by the assignment vertices, and all type-$y$ and type-$z$ vertices. Finally, as QAP is NP-complete and $(-3)$-QSCP is clearly in NP, it follows that $(-3)$-QSCP is NP-complete.\Halmos
\endproof

\medskip

We conclude by reducing the $(-3)$-QSCP to the separation of~\ref{ineq:ci}. Let~$I$ be an instance of $(-3)$-QSCP associated with graph~$G = (V,E)$ and quadratic cost function $q: V \times V \rightarrow [0,1]$. We reduce~$I$ to an instance~$I'$ of the separation of~\ref{ineq:ci} associated with the same graph~$G = (V,E)$. The (potentially fractional) solution $x \in X(G)$ is derived from the quadratic cost function~$q$ of~$I$ as follows: if $e = \{u,v\} \in E^C$, $x_e = q(u,v)$. Note that $x$ is valid, since $x_e \in [0,1]$ for all $e \in E$ and $x_e$ is not defined if $e \in E$.

By construction, any chordless cycle~$C$ in~$G$ has a cost~$c(C)$ in~$I$ deciding $(-3)$-QSCP positively has a cost~$c(C)$ such that
\[
c(C) = \sum_{\{u,v\} \in E^C(C)} q(u,v) = \sum_{f \in int(C) } x_{f}  < |C| - 3,
\]
that is, if~$C$ decides~$I$ positively, then~$C$ also decides~$I'$ positively. The same argument shows that if~$C$ decides~$I'$ positively, then~$C$ is also a valid certificate for~$I$. Finally, from Lemma~\ref{lemma:minimali1}, we know that the separation of~(\ref{ineq:ci}) can be restricted to cycles which are  chordless in~$G$, so we conclude that deciding whether~$I$ has a solution is equivalent to deciding whether~$I'$ has a solution. Thus, we conclude that the separation of~(\ref{ineq:ci}) is NP-complete. \Halmos
\endproof

\bigskip

\proof{Proof of Theorem \ref{thm:separation}(b).}
All coefficients of inequalities (\ref{inq:2}) are non-negative, so we are able to apply Theorem~\ref{thm:FacetsforInducedSubcycles} in order to obtain the following inequalities:
\begin{eqnarray}\label{eq:i4-extended2}
x_{ \{ v_{i-1},v_{i+1} \} } +  \sum_{f : v_i \in f, \{v_{i-1},v_{i+1}\} \cap f = \emptyset } x_{f} &\geq&   \sum_{f \in F(C)} x_{f} - |F(C)| + 1 \nonumber \\ 
&\geq& 1 - \sum_{f \in F(C)} (1 - x_{f}). 
\end{eqnarray} 
Note that inequality~(\ref{eq:i4-extended2}) is trivially satisfied if $\sum_{f \in F(C)} x_{f} < |F(C)|$, as the right-hand side expression becomes zero and all coefficients on the left are non-negative.

Let~$x \in X(G)$ be a fractional solution; in abuse of notation, if $\{u,v\} \in E(G)$, we assume that~$x_{\{u,v\}} = 1$. 
By construction, solution~$x$ violates the inequality~(\ref{eq:i4-extended2}) associated with
cycle $C = (v_{i-1},v_i,v_{i+1},v_1,v_2,\ldots,v_n)$ if and only if
\begin{eqnarray*} 
x_{ \{ v_{i-1},v_{i+1} \} }  &+&  \sum_{\substack{t \in C' \setminus \{v_{i-1},v_i,v_{i+1}\} } } x_{\{v_i,t\}} < \: 1 - \sum_{f \in F(C')} (1 - x_{f}) \implies \nonumber  \\
x_{ \{ v_{i-1},v_{i+1} \} }  &+&  \sum_{\substack{t \in C' \setminus \{v_{i-1},v_i,v_{i+1}\} } } x_{\{v_i,t\}} + \sum_{f \in F(C')} (1 - x_{f})  < \: 1 \implies \nonumber  \\
x_{ \{ v_{i-1},v_{i+1} \} } &+& 
  \sum_{\substack{ v_i \notin\{v_j,v_k\} \in F(C') }} \left(\frac{x_{\{v_i,v_j\}} + x_{\{v_i,v_k\}}}{2}\right)  + \sum_{\substack{ v_i \notin \{v_j,v_k\} \in F(C') }} (1 - x_{\{v_j,v_k\}}) + \\  
&&  \left(1 - \frac{3x_{\{v_{i-1},v_i\}}}{2}\right) + \left(1 -  \frac{3x_{\{v_i,v_{i+1}\}}}{2}\right)  
\end{eqnarray*} 

Let $G^{i} = (V(G),E(G))$ be a complete weighted direct graph such that, 
for each edge $e = \{v_j,v_k\}$ in $E(G)$, 
\[
w(e) = \begin{cases}
1 - x_{\{v_j,v_k\}} +  \frac{x_{\{v_i,v_j\}} + x_{\{v_i,v_k\}}}{2}, & \text{if } v_i \notin \{v_j,v_k\}\\
+\infty, & \text{otherwise.}
\end{cases}
\]

In order to separate inequalities~(\ref{eq:i4-extended2}), it suffices to find a  path in~$G^i$ connecting~$v_{i+1}$ to~$v_{i-1}$ not passing through~$v_i$ whose length is inferior to $1 - x_{\{v_{i-1},v_{i+1}\}} - \left(1 - \frac{3x_{v_{i-1},v_i}}{2}\right) - \left(1 - \frac{3x_{v_i,v_{i+1}}}{2}\right)$. If such a path exists, then, in particular, any shortest path in~$G^i$ connecting~$v_{i+1}$ to~$v_{i-1}$ while avoiding~$v_i$ also satisfies this property, so the verification can be done in polynomial time for each sequence $(v_{i-1},v_i,v_{i+1})$ (e.g., the running time for simple implementations of Dijsktra's algorithm is $O(|V(G)|^2$). The number of sequences for which this verification needs to be performed is $O(|V(G)|^3)$,
so we conclude that inequality~(\ref{eq:i4-extended2}) can be separated in polynomial time. $\Halmos$
\endproof

\bigskip

\proof{Proof of Theorem \ref{thm:separation}(c).} All coefficients of inequalities (\ref{inq:3}) are non-negative, so we are able to apply Theorem~\ref{thm:FacetsforInducedSubcycles} in order to obtain the following inequalities: 
\begin{eqnarray}\label{eq:i3-extended2}
\sum_{f : \{ \{v_i,v_j\}: d_C(v_i,v_j) = 2\}} x_{f} &\geq&   2\left(\sum_{f \in F(C)} x_{f} - |F(C)| + 1\right) \nonumber \implies \\
\sum_{f : \{ \{v_i,v_j\}: d_C(v_i,v_j) = 2\}} x_{f} &+& 2\sum_{f \in F(C)} (1 - x_{f})   \geq   2.
\end{eqnarray} 
Note that inequality~(\ref{eq:i3-extended2}) is trivially satisfied if $\sum_{f \in F(C)} x_{f} < |F(C)|$, as the right-hand side expression becomes zero and all coefficients on the left are non-negative.

Let~$x \in X(G)$ be a fractional solution; in abuse of notation, if $\{u,v\} \in E(G)$, we assume that~$x_{\{u,v\}} = 1$. 
Let~$D$ be a weighted directed graph such that, for each two-set $\{v_i,v_j\}$ in $V(G)$, there is one vertex in~$V(D)$ labelled by pair~$(v_i,v_j)$ and other labelled by pair~$(v_j,v_i)$. Moreover, for each pair of vertices $(v_i,v_j)$ and $(v_j,v_k)$ in $V(D)$, $v_i \neq v_k$, we define an arc~$a = ((v_i,v_j), (v_j,v_k))$ in $A(D)$ whose weight is given by 
\[
w_D(a) = w_D((v_i,v_j), (v_j,v_k)) = (1 - x_{\{v_i,v_j\}}) + x_{\{v_i,v_k\}} + (1 - x_{\{v_j,v_k\}});
\]
the first and the third terms of~$w_D(a)$ can be interpreted as penalties associated with the absence of edges $\{v_i,v_j\}$ and $\{v_j,v_k\}$ in~$G$, whereas the second penalizes the existence of  edge $\{v_i,v_k\}$. Finally, note that path $(v_1,v_2,\ldots,v_k)$ in~$G$ is associated with path $((v_1,v_2),(v_2,v_3),\ldots,(v_{k-1},v_{k}))$ in~$D$ (and vice-versa).

Let $u,v,w,$ and $x$ be vertices in $V(G)$ such that $P_G = (u,v,w,x)$ is a path in~$G$, and
let $P_D = (a_1,a_2,a_3)$ be the associated path in~$D$,
with $a_1 = (u,v)$, $a_2 = (v,w)$, and $a_3 = (w,x)$. Let $P_D' = ( v_1, v_2, \ldots, v_n )$ be other path in~$D$ such that $v_1 = (x,z_1)$, $v_i = (z_i,z_{i+1}), 1 \leq i < n$, and $v_n = (z_n,u)$; note that, by construction, arc $((a_3,v_1),(v_n,a_1))$ belongs to~$A(D)$. If all elements in $\{u,v,w,x,z_1,\ldots,z_n\}$ are pairwise different, then $C_D = (a_1,a_2,a_3,v_1,\ldots,v_n)$ is a directed cycle in~$D$ associated with cycle $C_G = (u,v,w,x,z_1,\ldots,z_n)$ in $G$. 

We claim that $C_D$ is a directed cycle in~$D$ if~$P_D'$ is a shortest path in~$D$ connecting~$(x,z')$ to~$(z'',u)$, with  $z',z''$ in $V(G)$, 
 whose internal vertices are not associated with edges in~$G$ containing vertices in $\{u,v,w,x\}$ and
 such that $|P_D'|$ is minimal. As $w_D(a) \geq 0$ for all $a \in A(D)$, it follows from the last condition and from the fact that~$P_D'$ is a shortest path that all elements in $\{z_1,\ldots,z_n\}$ are necessarily pairwise different.


The sum of the costs of all edges cycle~$C_D$ is given by
\begin{eqnarray}
\sum\limits_{i \in [1,2]} w((a_i,a_{i+1})) + w( (a_3,v_1) ) &+&
\sum\limits_{i \in [1,n-1]} w((v_i,v_{i+1})) + w( (v_n,a_1) ) = \nonumber\\
\sum_{f \in \{\{a,b\}:d_{C_G}(a,b) = 2\}}x_f &+& 2\sum_{f \in F(C_G)}(1 - x_f). \nonumber
\end{eqnarray}
Therefore, solution~$x$ does not respect the inequality~(\ref{eq:i3-extended2}) associated with  cycle~$C_G$ in $G$ containing path $P = (u,v,w,x)$ if and only if the weight of~$P_D'$ is smaller than~$2$.

The number of tuples for which this verification needs to be performed is~$O(|V(G)|^4)$, and the identification (and construction) of a path~$P_D'$ with the desired features  can be performed in polynomial time (e.g., $O(|V(D)|^2) = O(|V(G)|^4)$ using Dijkstra's algorithm); therefore, we conclude that Inequalities~\ref{eq:i3-extended2} can be separated in polynomial time. $\Halmos$
\endproof


\bigskip

\proof{Proof of Theorem \ref{thm:separation}(d).} 
We show this result by employing a construction that is very similar to the one used in the proof of Theorem~\ref{thm:separation}(a). More precisely, we introduce QSCP$^*$, a variation of QSCP that is more convenient for proving the hardness of (the simplified version of) (\ref{inq:4}), and show that the  addition of a single compensation penalty to the construction used 
in the proof of Theorem~\ref{thm:separation}(a) yields the desired result.



Similarly to the QSCP, in the Adapted Quadratic Shortest Cycle Problem (QSCP$^*$) we are given an  undirected graph $G = (V, E)$ and a quadratic cost function $q:V \times V \rightarrow [0,1]$ such that $q(u,v) = 0$ if $(u,v) \in E$. A feasible solution for an instance of QSCP$^*$ consists of a simple chordless cycle $C = (v_1,v_2,\ldots, v_{|C|})$  whose cost $p^*(C)$ is given by
\[p^*(C) = \sum_{\{u,v\} \in int^*(C)}q(u,v) -|C| - \max\limits_{\substack{i,j\\i \neq j\\d_C(\{v_j,v_i\}) \geq 2}  } (q(v_{j-1},v_{j+1}) + q(v_j,v_i)).  \]
Finally, $\alpha$-QSCP$^*$ is  the decision version of QSCP$^*$  in which the goal is to decide whether~$G$ admits a simple chordless cycle~$C$ such that $p^*(C) < \alpha$. 


The present proof also relies on a reduction of QAP to $(-4)$-QSCP$^*$. Let
$I$ be an arbitrary instance of QAP of size~$n$ and $M$ be the largest individual (i.e., linear or quadratic) penalty that may compose the cost of a feasible assignment. Note that~$K = 2Mn^2$ is an upper bound on the objective value of any feasible solution of~$I$.



We show how to adapt the construction presented in the proof of Theorem~\ref{thm:separation}(a) in order to construct an instance~$I'$ of $(-4)$-QSCP$^*$ that admits solution if and only if the associated instance~$I$ of QAP admits an assignment whose cost is inferior to~$\beta$. If \texttt{Alg} is an algorithm that decides $(-4)$-QSCP$^*$, then it will only return a cycle~$C$ containing the linear  costs~$p_L(C)$ and the quadratic costs~$p_Q(C)$ of $I$ if
\begin{eqnarray*}
p^*(C) = p_L(C) + p_Q(C) + q^*(C) -  \max\limits_{\substack{v_i,v_j \in V(C)\\v_i \neq v_j\\d_C(\{v_j,v_i\}) \geq 2}  } (q(v_{j-1},v_{j+1}) + q(v_{j},v_{i}) )< -4,
\end{eqnarray*}
where~$q^*(C)$ is the sum of $-|C|$ with additional costs incorporated by our construction. 
The costs of~$I$ are divided by~$K$ in~$I'$, so the assignment associated with $C$ is a solution if
\[
p_L(C) + p_Q(C) < \frac{\beta}{K}.
\]
Therefore, we define~$q^*(C)$ in the following way:
\begin{eqnarray*}
q^*(C) = -4 - \frac{\beta}{K} + \max\limits_{\substack{v_i,v_j \in V(C)\\v_i \neq v_j\\d_C(\{v_j,v_i\}) \geq 2}  }\left(q(v_{j-1},v_{j+1}) + q(v_{j},v_{i})\right).
\end{eqnarray*}
Note that the difference between $q^*(C)$ in this proof and~$q^*(C)$ in the proof of Theorem~\ref{thm:separation}(a) is $-1 + \max_{i,j}\left(q(v_{j-1},v_{j+1}) + q(v_{j},v_{i})\right)$. Moreover, in the original construction, one can see by inspection that if~$C$ is a matching cycle, then
\[
\max\limits_{\substack{v_i,v_j \in V(C)\\v_i \neq v_j\\d_C(\{v_j,v_i\}) \geq 2}  } \left( q(v_{j-1},v_{j+1}) + q(v_{j},v_{i})\right) \leq 1+\frac{M}{K}.
\]


In order to guarantee equality in the inequality above for every matching cycle, 
we set 
\[
q(a_{f_1,l_k},y_1) = \frac{M}{K}
\]
for all $l_k \in L$. By definition, $M \leq K$, so $q(a_{f_1,l_k},y_1) \leq 1$. With this modification, we have  $q(v_{j-1},v_{j+1})= 1$ and $q(v_j,v_i) = \frac{M}{K}$ 
for $v_{j-1} = c_{z_1,f_{1},\emptyset,l_k}$, $v_j = a_{f_1,l_k}$, $v_{j+1} = c_{f_1,f_{2},l_k,l_z}$, and $v_i = y_1$, $l_z,l_k \in L$. Moreover, for every matching cycle~$C$, by direct substitution we have 
\begin{eqnarray*}
q^*(C) &=& -3 - \frac{\beta}{K} + \frac{M}{K}.
\end{eqnarray*}

The arguments used in the proof of Lemma~\ref{lemma:nosubcycles} also apply to the present construction, so \texttt{Alg} can only select cycles that include all type-$y$ and type-$z$ vertices. For Lemma~\ref{lemma:nobadassignment}, note that if~$C$ contains assignments involving the same location or facility, then
\begin{eqnarray*}
p^*(C) &=& p_L(C) + p_Q(C) + q^*(C) -  \max\limits_{\substack{v_i,v_j \in V(C)\\v_i \neq v_j\\d_C(\{v_j,v_i\}) \geq 2}  } (q(v_{j-1},v_{j+1}) + q(v_{j},v_{i}) ) \\
	 &\geq& 1 -3 - \frac{\beta}{K} + \frac{M}{K} > -3 	  
\end{eqnarray*}
so \texttt{Alg} cannot decide $(-4)$-QSCP positively on~$G$ using~$C$ if~$C$ is associated with an infeasible assignment. 

Finally, the arguments used in the proof of Theorem \ref{thm:separation}(a) to show that $(-3)$-QSCP is NP-complete can be used in an identical way in order to show that $(-4)$-QSCP is NP-complete, and the problem of deciding the separation of (\ref{inq:4}) can be reduced to $(-4)$-QSCP in the same way the separation of (\ref{ineq:ci}) was reduced to $(-3)$-QSCP, so we conclude that the separation of (\ref{inq:4}) is also NP-complete. $\Halmos$
\end{proof}

\end{document}